%% file: main.tex
\documentclass[acmsmall]{acmart}
\AtBeginDocument{%
  }

\setcopyright{none}

\input{include}
\input{macro}

\begin{document}

\title{Contravariant Universe Subtyping, Formalized, Mechanized}

\author{Shengyi Jiang}
\email{syjiang@cs.hku.hk}
\orcid{0000-0002-4443-0753}
\affiliation{%
  \institution{The University of Hong Kong}
  \city{Hong Kong}
  \country{China}
}

\author{Jason Z. S. Hu}
\authornote{This research was conducted independently of Amazon.}
\email{fdhzs2010@hotmail.com}
\orcid{0000-0001-6710-6262}
\affiliation{%
  \institution{Amazon}
  \city{Seattle}
  \state{Washington}
  \country{USA}
}

\citestyle{acmauthoryear}

\begin{abstract}
  In a type theory with a cumulative universe hierarchy, \emph{subtyping} is often incorporated to lift lower universe levels to higher ones. %
  In practice, kernels of proof assistants typically implement
  universe subtyping conservatively: subtyping only propagates covariantly into the
  output types of $\Pi$ types, while input types must be equivalent. %
  A full-blown subtyping relation also compares the input types contravariantly. %
  However, contravariant subtyping is not justified by the existing set-theoretic model, and poses challenges in the development of a complete subtyping algorithm. 
  
  We present \mclttsub, a complete mechanized study of contravariant universe
  subtyping for a core Martin-L\"of type theory (MLTT), built as an extension of
  the verified McTT kernel in Rocq. %
  We re-establish the syntactic meta-theory and the soundness and completeness properties of the normalization-by-evaluation (NbE) algorithm. %
  To complete the verified pipeline for a type-checker, we give two complete and sound subtyping algorithms. %
  The first one \emph{renormalizes} output types in every $\Pi$ iteration. %
  Despite a simple completeness proof, this algorithm incurs additional computational cost. %
  The second algorithm hence circumvents this cost by avoiding the renormalization at all. %
  It simply recursively compares two already normalized types structurally. %
  Nevertheless, a completeness proof of this algorithm in a semantic approach requires to reason about normalization in sub-contexts, a capability beyond what the current semantic gadget offers. %
  Instead, we conduct a program equivalence proof through bisimulation relations. %
  Using this technique, we prove that the NbE algorithm returns equal normal forms in sub-contexts up to universe levels, and successfully conclude the completeness of the subtyping algorithm. %
\end{abstract}

\maketitle

\input{sections/introduction}
\input{sections/overview}
\input{sections/syntax}
\input{sections/nbe_per}
\input{sections/algo_typing}
\input{sections/algo_renorm}
\input{sections/algo_renorm_free}
\input{sections/discussion}
\input{sections/related}
\input{sections/conclusion}

\bibliographystyle{ACM-Reference-Format}
\bibliography{ref}

\appendix
\input{sections/appendix/rules}

\input{sections/appendix/mechanization}
\end{document}

%% file: include.tex
\PassOptionsToPackage{dvipsnames}{xcolor}

\usepackage[skins]{tcolorbox}
\usepackage{amsmath}

\usepackage{amssymb}
\usepackage{amsthm}
\usepackage{tikz-cd}
\usepackage{ifthen}
\usepackage{xspace}
\usepackage{xparse}
\usepackage{xstring}
\usepackage{mathpartir}
\usepackage{hyperref}
\usepackage{cleveref}
\usepackage{enumitem}
\usepackage{stmaryrd}
\usepackage{multicol}
\usepackage{tikz-cd}
\usepackage{mathtools}
\usepackage{stackengine}
\usepackage{scalerel}
\usepackage[normalem]{ulem}
\usepackage{wasysym}
\usepackage{fontawesome}
\usepackage{xcolor}
\usepackage{tcolorbox}
\usepackage{accents}
\usepackage{multirow}
\usepackage{array,longtable}
\usepackage{turnstile}
\usepackage{listings}
\usepackage{subcaption}
\usepackage{float}
\usepackage{needspace}
\usepackage{bbold}
\usepackage{amsmath}%
\usepackage{graphicx}%
\usepackage{contour}
\usepackage{pgf}%
\usepackage{fancyvrb}
\usepackage{lipsum}

\hypersetup{
  urlbordercolor = red,
  pdfborder={1 0 0},
}

\usetikzlibrary{matrix}
\usetikzlibrary{calc}
\usetikzlibrary{positioning}
\usetikzlibrary{decorations.pathreplacing}
\tikzstyle{block} = [draw, rectangle, minimum height=3em, minimum width=6em]
\usetikzlibrary{shapes.multipart}
\usetikzlibrary{fit}

\newcommand{\figrulesep}{\vspace{-3pt}}
\newcommand{\figtailspace}{\vspace{-5pt}}

\DeclareFontFamily{OT1}{pzc}{}
\DeclareFontShape{OT1}{pzc}{m}{it}{<-> s * [1.10] pzcmi7t}{}
\DeclareMathAlphabet{\mathpzc}{OT1}{pzc}{m}{it}

\crefname{figure}{fig.}{fig.}
\Crefname{figure}{Fig.}{Fig.}
\Crefname{section}{Sec.}{Sec.}

\lstdefinelanguage{ocaml}
{
language=ml,
keywords={match,let,rec,with},
morecomment=[s][\color{gray}]{(**}{*)},
literate={=}{{{\color{violet}=}}}1{|}{{\({\color{violet}|}\)}}1{->}{{{\color{violet}->}}}1
}

\definecolor{light-gray}{gray}{0.85}
\lstdefinelanguage{mints}%
{%
  keywords={box,unbox,letbox,let,in,quote,run,splice,zero,succ,match,with,nil,Ty,Some,None,if,then,else,
      fun,rec,end,return,forall},%
  sensitive,%
  moredelim=[is][\color{blue}]{bb<<}{>>bb},%
  moredelim=[is][\color{black}]{bl<<}{>>bl},%
  moredelim=*[is][\color{blue}]{b<<}{>>b},%
  moredelim=*[is][\color{teal}]{t<<}{>>t},%
  morecomment=[l][\color{gray}]{--},%
}
\newcolumntype{L}{>{$}l<{$}}
  \newcolumntype{R}{>{$}r<{$}}
\DefineVerbatimEnvironment{codeexample}{Verbatim}{fontsize=\small,commandchars=\\\{\}}

\contourlength{0.8pt}

\usepackage[nolist,nohyperlinks]{acronym}

\newenvironment{mathparshrink}
{\begin{@empty}%
    \begin{mathpar}}
      {\end{mathpar}\end{@empty}}

\tcbuselibrary{theorems}

\DeclareMathAlphabet{\mathpzc}{OT1}{pzc}{m}{it}

\DeclareUnicodeCharacter{2115}{$\mathbb{N}$}
\DeclareUnicodeCharacter{FFFD}{$\mathbb{U}$}
\DeclareUnicodeCharacter{21D2}{$\Rightarrow$}
\DeclareUnicodeCharacter{2113}{$\mathpzc{l}$}
\DeclareUnicodeCharacter{2248}{$\approx$}
\DeclareUnicodeCharacter{2200}{$\forall$}
\DeclareUnicodeCharacter{2218}{$\circ$}
\DeclareUnicodeCharacter{2294}{$\sqcup$}
\DeclareUnicodeCharacter{393}{$\Gamma$}
\DeclareUnicodeCharacter{3B1}{$\alpha$}
\DeclareUnicodeCharacter{3B2}{$\beta$}
\DeclareUnicodeCharacter{3B7}{$\eta$}
\DeclareUnicodeCharacter{3BB}{$\lambda$}
\DeclareUnicodeCharacter{3BC}{$\mu$}
\DeclareUnicodeCharacter{2080}{$_0$}
\DeclareUnicodeCharacter{2081}{$_1$}
\DeclareUnicodeCharacter{2082}{$_2$}
\DeclareUnicodeCharacter{2261}{$\equiv$}
\DeclareUnicodeCharacter{2E1}{$^l$}
\DeclareUnicodeCharacter{2B3}{$^r$}
\DeclareUnicodeCharacter{2032}{$'$}
\DeclareUnicodeCharacter{21D4}{$\Leftrightarrow$}
\DeclareUnicodeCharacter{224A}{$\approxeq$}
\DeclareUnicodeCharacter{27E8}{$\langle$}
\DeclareUnicodeCharacter{27E9}{$\rangle$}
\DeclareUnicodeCharacter{27F6}{$\longrightarrow$}
\DeclareUnicodeCharacter{220E}{$\blacksquare$}
\DeclareUnicodeCharacter{27FA}{$\Leftrightarrow$}
\DeclareUnicodeCharacter{22A4}{$\top$}
\DeclareUnicodeCharacter{22A5}{$\bot$}
\DeclareUnicodeCharacter{03A0}{$\Pi$}
\DeclareUnicodeCharacter{03A3}{$\Sigma$}
\DeclareUnicodeCharacter{03C0}{$\pi$}
\DeclareUnicodeCharacter{03C1}{$\rho$}
\DeclareUnicodeCharacter{21A6}{$\mapsto$}
\DeclareUnicodeCharacter{2199}{$\swarrow$}
\DeclareUnicodeCharacter{22A2}{$\vdash$}
\DeclareUnicodeCharacter{2236}{$\colon$}
\DeclareUnicodeCharacter{2208}{$\in$}
\DeclareUnicodeCharacter{27E6}{$\llbracket$}
\DeclareUnicodeCharacter{27E7}{$\rrbracket$}

%% file: macro.tex
\newcommand{\mycomment}[1]{}

\newcommand*\mynote[3]
{
  \ifvmode
    \textcolor{#2}{#1: #3}
  \else
    \textcolor{#2}{~#1: #3}
  \fi
}

\newcommand{\hlmath}[2][gray!40]{%
  \colorbox{#1}{$\displaystyle#2$}}

\newcommand{\ca}[1]{\textcolor{SteelBlue}{#1}}
\newcommand{\cb}[1]{\textcolor{FireBrick}{#1}}

\newcommand{\codelevel}{\textsc{c}\xspace}
\newcommand{\proglevel}{\textsc{d}\xspace}

\newcommand{\mcltt}{\textsc{McTT}\xspace}
\newcommand{\mclttsub}{\textsc{McTT}$_{\subseteq}$\xspace}

\newcommand{\coq}{\textsc{Rocq}\xspace}
\newcommand{\agda}{\textsc{Agda}\xspace}

\makeatletter
\newcommand{\inlineitem}[1][]{%
\ifnum\enit@type=\tw@
    {\descriptionlabel{#1}}
  \hspace{\labelsep}%
\else
  \ifnum\enit@type=\z@
       \refstepcounter{\@listctr}\fi
    \quad\@itemlabel\hspace{\labelsep}%
\fi}
\makeatother

\theoremstyle{remark}

\makeatletter
\def\namedlabel#1#2{\begingroup
  \def\@currentlabel{#2}%
  \label{#1}\endgroup
}
\makeatother

\makeatletter
\newcommand{\customlabel}[2]{%
  \protected@write \@auxout {}{\string \newlabel {#1}{{#2}{\thepage}{#2}{#1}{}} }%
  \hypertarget{#1}{}
}
\makeatother

\makeatletter
\DeclareFontFamily{U}{MnSymbolA}{}
\DeclareFontShape{U}{MnSymbolA}{m}{n}{
  <-6>  MnSymbolA5
  <6-7>  MnSymbolA6
  <7-8>  MnSymbolA7
  <8-9>  MnSymbolA8
  <9-10> MnSymbolA9
  <10-12> MnSymbolA10
  <12->   MnSymbolA12}{}
\DeclareFontShape{U}{MnSymbolA}{b}{n}{
  <-6>  MnSymbolA-Bold5
  <6-7>  MnSymbolA-Bold6
  <7-8>  MnSymbolA-Bold7
  <8-9>  MnSymbolA-Bold8
  <9-10> MnSymbolA-Bold9
  <10-12> MnSymbolA-Bold10
  <12->   MnSymbolA-Bold12}{}
\DeclareSymbolFont{MnSyA}{U}{MnSymbolA}{m}{n}
\SetSymbolFont{MnSyA}{bold}{U}{MnSymbolA}{b}{n}

\DeclareRobustCommand{\overleftharpoon}{\mathpalette{\overarrow@\leftharpoonfill@}}
\DeclareRobustCommand{\overrightharpoon}{\mathpalette{\overarrow@\rightharpoonfill@}}
\def\leftharpoonfill@{\arrowfill@\leftharpoondown\mn@relbar\mn@relbar}
\def\rightharpoonfill@{\arrowfill@\mn@relbar\mn@relbar\rightharpoonup}

\DeclareMathSymbol{\leftharpoondown}{\mathrel}{MnSyA}{'112}
\DeclareMathSymbol{\rightharpoonup}{\mathrel}{MnSyA}{'100}
\DeclareMathSymbol{\mn@relbar}{\mathrel}{MnSyA}{'320}
\makeatother

\DeclareFontFamily{U}{mathx}{}
\DeclareFontShape{U}{mathx}{m}{n}{<-> mathx10}{}
\DeclareSymbolFont{mathx}{U}{mathx}{m}{n}
\DeclareMathAccent{\widecheck}{0}{mathx}{"71}

\newcommand{\vect}[1]{\overrightarrow{#1}}

\newcommand{\vGamma}{\vect \Gamma}

\newcommand{\usynrec}{\texttt{rec}}

\newcommand{\usynlam}[3]{\lambda (#1:#2). #3}
\newcommand{\usynPi}[3]{\Pi(#1 : #2). #3}

\newcommand{\STo}{\Rightarrow}
\newcommand{\D}{\ensuremath{\mathcal{D}}}

\newcommand{\Nat}{\ensuremath{\texttt{Nat}}}
\newcommand{\Int}{\ensuremath{\texttt{Int}}}

\newcommand{\trec}{\usynrec}
\newcommand{\recn}[4]{\trec (#1)~#2~(#3)~#4}
\newcommand{\Nf}{\textsf{Nf}}
\newcommand{\Ne}{\textsf{Ne}}
\newcommand{\Nfd}{\textsl{Nf}}
\newcommand{\Ned}{\textsl{Ne}}
\newcommand{\Tyd}{\textsl{Ty}}

\newcommand{\Env}{\textsf{Env}}

\newcommand{\N}{\mathbb{N}}
\newcommand{\intp}[1]{\ensuremath{\llbracket #1 \rrbracket}}

\newcommand{\Rnf}{\textsf{R}^{\Nf}}
\newcommand{\Rne}{\textsf{R}^{\Ne}}
\newcommand{\Rty}{\textsf{R}^{\textsf{Ty}}}

\newcommand{\semsubtyp}[3]{\ensuremath{#1 <:_{#3} #2}}

\newcommand{\sep}{\;|\;}

\newcommand{\nbe}{\textsf{nbe}}

\newcommand{\Ty}[1]{\ensuremath{\Se@{#1}}}
\newcommand{\tEl}{\texttt{El}}

\newcommand\Uc{\mathcal{U}}
\newcommand\El{\textbf{El}}

\newcommand{\dapp}[2]{\ensuremath{#1 \cdot #2}}

\newcommand{\Dom}{\textsf{D}}
\newcommand{\DomNe}{\Dom^{\Ne}}
\newcommand{\DomNf}{\Dom^{\Nf}}

\newcommand{\at}[1]{\scalebox{0.75}{\color{gray}$\,@~#1$}}

\newcommand{\vertrule}[1][1ex]{\rule{.45pt}{#1}}

\newcommand{\tttstileA}{\ensuremath{\mathrel{\raisebox{0.01pt}{\vertrule[1.6ex]}\hspace{.15em}{\vertrule[1.6ex]}{\hspace{.1em}\vDash}}}}
\newcommand{\sttstileA}{\ensuremath{\mathrel{\raisebox{0.01pt}{\vertrule[1.6ex]}{\hspace{.1em}\vDash}}}}

\DeclareDocumentCommand{\genjudge}{ m m m } {
  \ensuremath{#1 #2 #3}
}

\DeclareDocumentCommand{\singlejudge}{ o m m } {
  \IfNoValueTF {#1}
  {\genjudge{\Gamma}{#2}{#3}}
  {\genjudge{#1}{#2}{#3}}
}

\DeclareDocumentCommand{\dualjudge}{o o m m m } {
  \IfNoValueTF {#1}
  {\IfNoValueTF {#2}
    {\genjudge{\Psi #3 \Gamma}{#4}{#5}}
    {\genjudge{\Psi #3 #2}{#4}{#5}}}
  {\IfNoValueTF {#2}
    {\genjudge{#1 #3 \Gamma}{#4}{#5}}
    {\genjudge{#1 #3 #2}{#4}{#5}}}
}

\DeclareDocumentCommand{\triplejudge}{o o o m m } {
  \IfNoValueTF {#1}
  {\IfNoValueTF {#2}
    {\IfNoValueTF {#3}
      {\genjudge{L \sep \Psi; \Gamma}{#4}{#5}}
      {\genjudge{#3 \sep \Psi; \Gamma}{#4}{#5}}}
    {\IfNoValueTF {#3}
      {\genjudge{L \sep \Psi; #2}{#4}{#5}}
      {\genjudge{#3 \sep \Psi; #2}{#4}{#5}}}}
  {\IfNoValueTF {#2}
    {\IfNoValueTF {#3}
      {\genjudge{L \sep #1; \Gamma}{#4}{#5}}
      {\genjudge{#3 \sep #1; \Gamma}{#4}{#5}}}
    {\IfNoValueTF {#3}
      {\genjudge{L \sep #1; #2}{#4}{#5}}
      {\genjudge{#3 \sep #1; #2}{#4}{#5}}}}
}

\DeclareDocumentCommand{\dcjudge}{o o m m } {
  \dualjudge[#1][#2]{;}{#3}{#4}
}

\DeclareDocumentCommand{\dSemjudgebf}{ o o m } {
  \dcjudge[#1][#2]{\sttstileA}{#3}
}

\DeclareDocumentCommand{\judge}{ o m } {
  \singlejudge[#1]{\vdash}{#2}
}

\DeclareDocumentCommand{\gequiv}{ o m m } {
  \IfNoValueTF {#1}
  {\judge[L]{#2 \approx #3}}
  {\judge[#1]{#2 \approx #3}}
}

\DeclareDocumentCommand{\cont}{ o m } {
  \square (\judge[#1]{#2})
}

\DeclareDocumentCommand{\rjudge}{ o m } {
  \IfNoValueTF {#1}
  {\singlejudge[\vGamma]{\vdash_r}{#2}}
  {\singlejudge[#1]{\vdash_r}{#2}}
}

\DeclareDocumentCommand{\semjudge}{ o m } {
  \singlejudge[#1]{\vDash}{#2}
}

\DeclareDocumentCommand{\Semjudge}{ o m } {
  \singlejudge[#1]{\Vdash}{#2}
}

\DeclareDocumentCommand{\mjudge}{ o m } {
  \IfNoValueTF {#1}
  {\singlejudge[\vGamma]{\vdash}{#2}}
  {\singlejudge[#1]{\vdash}{#2}}
}

\DeclareDocumentCommand{\djudge}{ o o m } {
  \dualjudge[#1][#2]{;}{\vdash}{#3}
}

\DeclareDocumentCommand{\ljudge}{ o m m } {
  \IfNoValueTF {#1}
  {\genjudge{\Psi}{\vdash_{#2}}{#3}}
  {\genjudge{#1}{\vdash_{#2}}{#3}}
}

\DeclareDocumentCommand{\ljudgel}{ o m m } {
  \IfNoValueTF {#1}
  {\genjudge{\Gamma}{\vdash_{#2}}{#3}}
  {\genjudge{#1}{\vdash_{#2}}{#3}}
}

\DeclareDocumentCommand{\DTyp}{o m m} {
  \ensuremath{(\ljudgel[#1]{#2}{\at{#3}})}
}

\DeclareDocumentCommand{\DTrm}{o m m m} {
  \ensuremath{(\ljudgel[#1]{#2}{#3\at{#4}})}
}

\DeclareDocumentCommand{\CTyp}{o m} {
  \ensuremath{\square \DTyp[#1]{\codelevel}{#2}}
}

\DeclareDocumentCommand{\CTrm}{o m m} {
  \ensuremath{\square \DTrm[#1]{\codelevel}{#2}{#3}}
}

\DeclareDocumentCommand{\TPI}{m o m m m}{
  \ensuremath{(#1 : \DTyp[#2]{\proglevel}{#3}) \STo^{#4} #5}
}

\DeclareDocumentCommand{\lpgenjudge}{ o o m m } {
  \IfNoValueTF {#1}
  {\IfNoValueTF {#2}
    {\genjudge{L \sep \Psi}{#3}{#4}}
    {\genjudge{#2 \sep \Psi}{#3}{#4}}}
  {\IfNoValueTF {#2}
    {\genjudge{L \sep #1}{#3}{#4}}
    {\genjudge{#2 \sep #1}{#3}{#4}}}
}

\DeclareDocumentCommand{\pjudge}{ o o m } {
  \lpgenjudge[#1][#2]{\vdash}{#3}
}

\DeclareDocumentCommand{\ptyping}{o o m m } {
  \pjudge[#1][#2]{#3 : #4}
}

\DeclareDocumentCommand{\ptyequiv}{o o m m m } {
  \pjudge[#1][#2]{#3 \approx #4 : #5}
}

\DeclareDocumentCommand{\lpjudge}{ o o m m } {
  \pjudge[#1][#2]{_{#3}#4}
}

\DeclareDocumentCommand{\lpequiv}{ o o m m m } {
  \lpjudge[#1][#2]{#3}{#4 \approx #5}
}

\DeclareDocumentCommand{\lpcobj}{ o o m m m } {
  \lpjudge[#1][#2]{#3}{#4 : #5}
}

\DeclareDocumentCommand{\tjudge}{ o o o m } {
  \triplejudge[#1][#2][#3]{\vdash}{#4}
}

\DeclareDocumentCommand{\ltjudge}{o o o m m} {
  \tjudge[#1][#2][#3]{_{#4} #5}
}

\DeclareDocumentCommand{\lttypwf}{o o o m m m} {
  \ltjudge[#1][#2][#3]{#4}{#5 \at{#6}}
}

\DeclareDocumentCommand{\lttyping}{o o o m m m m} {
  \ltjudge[#1][#2][#3]{#4}{#5 : #6 \at{#7}}
}

\DeclareDocumentCommand{\lttypingd}{o o o m m m m} {
  \ltjudge[#1][#2][#3]{#4}{#5 : \Ty{#6}\;(\at{#7})}
}

\DeclareDocumentCommand{\ltsubst}{o o o m m m} {
  \ltjudge[#1][#2][#3]{#4}{#5 : #6}
}

\DeclareDocumentCommand{\ltsubstv}{o o o m m m} {
  \ltjudge[#1][#2][#3]{#4}{^{\#} #5 : #6}
}

\DeclareDocumentCommand{\ltsubeq}{o o o m m m m} {
  \ltjudge[#1][#2][#3]{#4}{#5 \approx #6 : #7}
}

\DeclareDocumentCommand{\lttypeq}{o o o m m m m} {
  \ltjudge[#1][#2][#3]{#4}{#5 \approx #6 \at{#7}}
}

\DeclareDocumentCommand{\lttyequiv}{o o o m m m m m} {
  \ltjudge[#1][#2][#3]{#4}{#5 \approx #6 : #7 \at{#8}}
}

\DeclareDocumentCommand{\lttyequivv}{o o o m m m m m} {
  \ltjudge[#1][#2][#3]{#4}{^{\#} #5 \approx #6 : #7 \at{#8}}
}

\DeclareDocumentCommand{\lttyequivd}{o o o m m m m m} {
  \ltjudge[#1][#2][#3]{#4}{#5 \approx #6 : \Ty{#7}\;(\at{#8})}
}

\DeclareDocumentCommand{\lsemjudge}{ o m m } {
  \IfNoValueTF {#1}
  {\genjudge{\Psi}{\sttstileA^{#2}}{#3}}
  {\genjudge{#1}{\sttstileA^{#2}}{#3}}
}

\DeclareDocumentCommand{\lsemvjudge}{ o m m } {
  \IfNoValueTF {#1}
  {\genjudge{\Psi}{\Vdash^{#2}}{#3}}
  {\genjudge{#1}{\Vdash^{#2}}{#3}}
}

\DeclareDocumentCommand{\lmjudge}{ o o m m } {
  \djudge[#1][#2]{_{#3} #4}
}

\DeclareDocumentCommand{\dpjudge}{ o o m } {
  \lpgenjudge[#1][#2]{\sttstileA}{#3}
}

\DeclareDocumentCommand{\dtjudge}{ o o o m } {
  \triplejudge[#1][#2][#3]{\tttstileA}{#4}
}

\DeclareDocumentCommand{\ldpjudge}{ o o m m m } {
  \dpjudge[#1][#2]{_{#3}^{#4} #5}
}

\DeclareDocumentCommand{\ldjudge}{ o o o m m m } {
  \dtjudge[#1][#2][#3]{_{#4}^{#5} #6}
}

\DeclareDocumentCommand{\dpSjudge}{ o o m } {
  \lpgenjudge[#1][#2]{\sttstileA}{#3}
}

\DeclareDocumentCommand{\dSjudge}{ o o o m } {
  \triplejudge[#1][#2][#3]{\sttstileA}{#4}
}

\DeclareDocumentCommand{\dsjudge}{ o o o m } {
  \triplejudge[#1][#2][#3]{\Vdash}{#4}
}

\DeclareDocumentCommand{\ldpSjudge}{ o o m m m } {
  \dpSjudge[#1][#2]{_{#3}^{#4} #5}
}

\DeclareDocumentCommand{\ldSjudge}{ o o o m m m } {
  \dSjudge[#1][#2][#3]{_{#4}^{#5} #6}
}

\DeclareDocumentCommand{\dsemjudge}{ o o m } {
  \dualjudge[#1][#2]{;}{\sttstileA}{#3}
}

\DeclareDocumentCommand{\dSemjudge}{ o o m } {
  \dualjudge[#1][#2]{;}{\Vdash}{#3}
}

\DeclareDocumentCommand{\wmjudge}{ o m } {
  \IfNoValueTF {#1}
  {\vGamma \vdash_{\! w} #2}
  {#1 \vdash_{\! w} #2}
}

\DeclareDocumentCommand{\msemjudge}{ o m } {
  \IfNoValueTF {#1}
  {\vGamma \vDash #2}
  {#1 \vDash #2}
}

\DeclareDocumentCommand{\mSemjudge}{ o m } {
  \IfNoValueTF {#1}
  {\vGamma \Vdash #2}
  {#1 \Vdash #2}
}

\DeclareDocumentCommand{\typing}{ o m m } {
  \judge[#1]{#2 : #3}
}

\DeclareDocumentCommand{\subtyp}{ o m m } {
  \judge[#1]{#2 \subseteq #3}
}

\DeclareDocumentCommand{\rtyping}{ o m m } {
  \rjudge[#1]{#2 : #3}
}

\DeclareDocumentCommand{\wtyping}{ o m m } {
  \judge[#1]{_w #2 : #3}
}

\DeclareDocumentCommand{\semtyp}{ o m m } {
  \semjudge[#1]{#2 : #3}
}

\DeclareDocumentCommand{\semvtyp}{ o m m } {
  \IfNoValueTF {#1}
  {\semtyp[\Psi]{ #2}{#3}}
  {\semtyp[#1]{ #2}{#3}}
}

\DeclareDocumentCommand{\Semtyp}{ o m m } {
  \Semjudge[#1]{#2 : #3}
}

\DeclareDocumentCommand{\tyequiv}{ o m m m } {
  \judge[#1]{#2 \approx #3 : #4}
}

\DeclareDocumentCommand{\semtyeq}{ o m m m } {
  \semjudge[#1]{#2 \approx #3 : #4}
}

\DeclareDocumentCommand{\semvtyeq}{ o m m m } {
  \IfNoValueTF {#1}
  {\semjudge[\Psi]{ #2 \approx #3 : #4}}
  {\semjudge[#1]{ #2 \approx #3 : #4}}
}

\DeclareDocumentCommand{\atyping}{ o m m m } {
\judge[#1]{#2 :^{#3} #4}
}
\DeclareDocumentCommand{\atyequiv}{ o m m m m } {
\judge[#1]{#2 \approx #3 :^{#4} #5}
}

\DeclareDocumentCommand{\lSemtypPrime}{ o o m m m } {
  \dSemjudgebf[#1][#2]{_{#3} #4 : #5}
}

\DeclareDocumentCommand{\mtyping}{ o m m } {
  \mjudge[#1]{#2 : #3}
}

\DeclareDocumentCommand{\lmtyping}{ o m m m } {
  \mjudge[#1]{_{#2} #3 : #4}
}

\DeclareDocumentCommand{\dtyping}{ o o m m } {
  \djudge[#1][#2]{#3 : #4}
}

\DeclareDocumentCommand{\ltyping}{ o o m m m } {
  \djudge[#1][#2]{_{#3} #4 : #5}
}

\DeclareDocumentCommand{\wmtyping}{ o m m } {
  \wmjudge[#1]{#2 : #3}
}

\DeclareDocumentCommand{\mSemtyp}{ o m m } {
  \mSemjudge[#1]{#2 : #3}
}

\DeclareDocumentCommand{\lsemtyp}{ o o m m m } {
  \dsemjudge[#1][#2]{_{#3} #4 : #5}
}

\DeclareDocumentCommand{\lsemvtyp}{ o o m m m } {
  \dSemjudge[#1][#2]{_{#3} #4 : #5}
}

\DeclareDocumentCommand{\lSemtyp}{ o o m m m } {
  \dSemjudge[#1][#2]{_{#3} #4 : #5}
}

\DeclareDocumentCommand{\mrarr}{ o m m } {
  \mjudge[#1]{#2 > #3}
}

\DeclareDocumentCommand{\wmrarr}{ o m m } {
  \wmjudge[#1]{#2 > #3}
}

\DeclareDocumentCommand{\mtyequiv}{ o m m m } {
  \mjudge[#1]{#2 \approx #3 : #4}
}

\DeclareDocumentCommand{\dtyequiv}{ o o m m m } {
  \djudge[#1][#2]{#3 \approx #4 : #5}
}

\DeclareDocumentCommand{\ltyequiv}{ o o m m m m } {
  \djudge[#1][#2]{_{#3} #4 \approx #5 : #6}
}

\DeclareDocumentCommand{\ltygneeq}{ o o m m m } {
  \djudge[#1][#2]{_1 #3 \sim #4 : #5}
}

\DeclareDocumentCommand{\ltygteq}{ o o m m m } {
  \djudge[#1][#2]{_1 #3 \simeq #4 : #5}
}

\DeclareDocumentCommand{\ltyred}{ o o m m m } {
  \djudge[#1][#2]{_1 #3 \rightsquigarrow #4 : #5}
}

\DeclareDocumentCommand{\ltyreds}{ o o m m m } {
  \djudge[#1][#2]{_1 #3 \rightsquigarrow^\ast #4 : #5}
}

\DeclareDocumentCommand{\dtconv}{ o o m m m } {
  \djudge[#1][#2]{_1 #3 \ \hat{\Longleftrightarrow}\ #4 : #5}
}

\DeclareDocumentCommand{\dtconvnf}{ o o m m m } {
  \djudge[#1][#2]{_1 #3 \Longleftrightarrow #4 : #5}
}

\DeclareDocumentCommand{\dtconvne}{ o o m m m } {
  \djudge[#1][#2]{_1 #3 \longleftrightarrow #4 : #5}
}

\DeclareDocumentCommand{\ttypred}{ o o o m m m m } {
  \ltjudge[#1][#2][#3]{#4}{ #5 \rightsquigarrow #6 \at{#7}}
}

\DeclareDocumentCommand{\ttypreds}{ o o o m m m m } {
  \ltjudge[#1][#2][#3]{#4}{ #5 \rightsquigarrow^\ast #6 \at{#7}}
}

\DeclareDocumentCommand{\ttrmred}{ o o o m m m m m } {
  \ltjudge[#1][#2][#3]{#4}{ #5 \rightsquigarrow #6 : #7 \at{#8}}
}

\DeclareDocumentCommand{\ttrmreds}{ o o o m m m m m } {
  \ltjudge[#1][#2][#3]{#4}{ #5 \rightsquigarrow^\ast #6 : #7 \at{#8}}
}

\DeclareDocumentCommand{\tconvtyp}{ o o o m m m m } {
  \ltjudge[#1][#2][#3]{#4}{ #5 \ \hat{\Longleftrightarrow}\ #6 \at{#7}}
}

\DeclareDocumentCommand{\tconvtypnf}{ o o o m m m m } {
  \ltjudge[#1][#2][#3]{#4}{ #5 \Longleftrightarrow #6 \at{#7}}
}

\DeclareDocumentCommand{\tconvtypne}{ o o o m m m m } {
  \ltjudge[#1][#2][#3]{#4}{ #5 \longleftrightarrow #6 \at{#7}}
}

\DeclareDocumentCommand{\tconvctx}{o o m m m}{
  \lpjudge[#1][#2]{#3}{#4 \ \hat{\Longleftrightarrow}\ #5}
}

\DeclareDocumentCommand{\tconvtrm}{ o o o m m m m m } {
  \ltjudge[#1][#2][#3]{#4}{ #5 \ \hat{\Longleftrightarrow}\ #6 : #7 \at{#8}}
}

\DeclareDocumentCommand{\tconvtrmnf}{ o o o m m m m m } {
  \ltjudge[#1][#2][#3]{#4}{ #5 \Longleftrightarrow #6 : #7 \at{#8}}
}

\DeclareDocumentCommand{\tconvtrmne}{ o o o m m m m m } {
  \ltjudge[#1][#2][#3]{#4}{ #5 \ \hat{\longleftrightarrow}\ #6 : #7 \at{#8}}
}

\DeclareDocumentCommand{\tconvtrmnee}{ o o o m m m m m } {
  \ltjudge[#1][#2][#3]{#4}{ #5 \longleftrightarrow #6 : #7 \at{#8}}
}

\DeclareDocumentCommand{\tconvsub}{ o o o m m m m } {
  \ltjudge[#1][#2][#3]{#4}{ #5 \ \hat{\Longleftrightarrow}\ #6 : #7}
}

\DeclareDocumentCommand{\lttypgneeq}{ o o o m m m m } {
  \ltjudge[#1][#2][#3]{#4}{#5 \sim #6 \at{#7}}
}

\DeclareDocumentCommand{\lttypgeq}{ o o o m m m m } {
  \ltjudge[#1][#2][#3]{#4}{#5 \simeq #6 \at{#7}}
}

\DeclareDocumentCommand{\lttrmgneeq}{ o o o m m m m m } {
  \ltjudge[#1][#2][#3]{#4}{#5 \sim #6 : #7 \at{#8}}
}

\DeclareDocumentCommand{\lttrmgeq}{ o o o m m m m m } {
  \ltjudge[#1][#2][#3]{#4}{#5 \simeq #6 : #7 \at{#8}}
}

\DeclareDocumentCommand{\ltctxgeq}{o o m m m}{
  \lpjudge[#1][#2]{#3}{#4 \simeq #5}
}

\DeclareDocumentCommand{\ltsubgeq}{ o o o m m m m } {
  \ltjudge[#1][#2][#3]{#4}{ #5 \simeq #6 : #7}
}

\DeclareDocumentCommand{\lsemtyp}{ o o m m m } {
  \dsemjudge[#1][#2]{_{#3} #4 : #5}
}

\DeclareDocumentCommand{\dsemtyeq}{ o o m m m } {
  \dsemjudge[#1][#2]{#3 \approx #4 : #5}
}

\DeclareDocumentCommand{\ldtypwf}{ o o o m m m m } {
  \ldjudge[#1][#2][#3]{#4}{#5}{#6 \at{#7}}
}

\DeclareDocumentCommand{\ldtypeq}{ o o o m m m m m } {
  \ldjudge[#1][#2][#3]{#4}{#5}{#6 \approx #7 \at{#8}}
}

\DeclareDocumentCommand{\ldtyeqnat}{ o o o m m m } {
  \ldjudge[#1][#2][#3]{#4}{}{#5 \approx #6 : \Nat}
}

\DeclareDocumentCommand{\ldtynfeqnat}{ o o o m m m } {
  \ldjudge[#1][#2][#3]{#4}{}{#5 \simeq #6 : \Nat}
}

\DeclareDocumentCommand{\ldtynfeqm}{ o o o m m m } {
  \ldjudge[#1][#2][#3]{}{}{#4 \simeq #5 : #6}
}

\DeclareDocumentCommand{\ldtyping}{ o o o m m m m m } {
  \ldjudge[#1][#2][#3]{#4}{#5}{#6 : #7 \at{#8}}
}

\DeclareDocumentCommand{\ldtyequivt}{ o o o m m m m m m } {
  \ldjudge[#1][#2][#3]{#4}{#5}{#6 \approx #7 : #8 \at{#9}}
}

\DeclareDocumentCommand{\ldtyequiv}{ o o o m m m m m } {
  \ldjudge[#1][#2][#3]{#4}{#5}{#6 \approx #7 : \tEl(#8)}
}

\DeclareDocumentCommand{\ldctxwf}{ o o m m m } {
  \ldpjudge[#1][#2]{#3}{#4}{#5}
}

\DeclareDocumentCommand{\ldctxeq}{ o o m m m m } {
  \ldpjudge[#1][#2]{#3}{#4}{#5 \approx #6}
}

\DeclareDocumentCommand{\ldsubeq}{ o o o m m m m m } {
  \ldjudge[#1][#2][#3]{#4}{#5}{#6 \approx #7 : #8}
}

\DeclareDocumentCommand{\ldsubst}{ o o o m m m m } {
  \ldjudge[#1][#2][#3]{#4}{#5}{#6 : #7}
}

\DeclareDocumentCommand{\spjudge}{ o o m } {
  \lpgenjudge[#1][#2]{\Vdash}{#3}
}

\DeclareDocumentCommand{\stjudge}{ o o o m } {
  \triplejudge[#1][#2][#3]{\Vdash}{#4}
}

\DeclareDocumentCommand{\lspjudge}{ o o m m } {
  \spjudge[#1][#2]{_{#3} #4}
}

\DeclareDocumentCommand{\lsjudge}{ o o o m m } {
  \stjudge[#1][#2][#3]{_{#4} #5}
}

\DeclareDocumentCommand{\lsgctx}{ o m } {
  \IfNoValueTF {#1}
  {\genjudge{L}{\Vdash}{#2}}
  {\genjudge{#1}{\Vdash}{#2}}
}

\DeclareDocumentCommand{\lstypwf}{ o o o m m m } {
  \lsjudge[#1][#2][#3]{#4}{#5 \at{#6}}
}

\DeclareDocumentCommand{\lstypeq}{ o o o m m m m } {
  \lsjudge[#1][#2][#3]{#4}{#5 \approx #6 \at{#7}}
}

\DeclareDocumentCommand{\lstyping}{ o o o m m m m } {
  \lsjudge[#1][#2][#3]{#4}{#5 : #6 \at{#7}}
}

\DeclareDocumentCommand{\lstyequiv}{ o o o m m m m m } {
  \lsjudge[#1][#2][#3]{#4}{#5 \approx #6 : #7 \at{#8}}
}

\DeclareDocumentCommand{\lsgsubeq}{ o o m m m } {
  \lspjudge[#1][#2]{}{#3 \approx #4 : #5}
}

\DeclareDocumentCommand{\lsgsubst}{ o o m m } {
  \lspjudge[#1][#2]{}{#3 : #4}
}

\DeclareDocumentCommand{\lsctxwf}{ o o m m } {
  \lspjudge[#1][#2]{#3}{#4}
}

\DeclareDocumentCommand{\lsctxeq}{ o o m m m } {
  \lspjudge[#1][#2]{#3}{#4 \approx #5}
}

\DeclareDocumentCommand{\lssubeq}{ o o o m m m m } {
  \lsjudge[#1][#2][#3]{#4}{#5 \approx #6 : #7}
}

\DeclareDocumentCommand{\lssubst}{ o o o m m m } {
  \lsjudge[#1][#2][#3]{#4}{#5 : #6}
}

\DeclareDocumentCommand{\ldStypeq} { o o m m m m m m } {
  \ldSjudge[#1][#3][#2]{#4}{#5}{#6 \approx #7 \at{#8}}
}

\DeclareDocumentCommand{\ldStyequiv} { o o m m m m m m } {
  \ldSjudge[#1][#3][#2]{#4}{#5}{#6 \approx #7 : \tEl(#8)}
}

\DeclareDocumentCommand{\ldStypwfge} { o o o m m m m } {
  \ldjudge[#1][#2][#3]{\ge #4}{#5}{#6 \at{#7}}
}

\DeclareDocumentCommand{\ldStypingge} { o o o m m m m m } {
  \ldjudge[#1][#2][#3]{\ge #4}{#5}{#6 : #7 \at{#8}}
}

\DeclareDocumentCommand{\ldStypeqge} { o o o m m m m m } {
  \ldjudge[#1][#2][#3]{\ge #4}{#5}{#6 \approx #7 \at{#8}}
}

\DeclareDocumentCommand{\ldStyequivge} { o o o m m m m m m } {
  \ldjudge[#1][#2][#3]{\ge #4}{#5}{#6 \approx #7 : #8 \at{#9}}
}

\DeclareDocumentCommand{\ldSctxeqge} { o o m m m m } {
  \ldpSjudge[#1][#2]{\ge #3}{#4}{#5 \approx #6}
}

\DeclareDocumentCommand{\ldSctxwfge} { o o m m m } {
  \ldpSjudge[#1][#2]{\ge #3}{#4}{#5}
}

\DeclareDocumentCommand{\ldSsubeqge} { o o o m m m m m } {
  \ldjudge[#1][#2][#3]{\ge #4}{#5}{#6 \approx #7 : #8}
}

\DeclareDocumentCommand{\ldSsubstge} { o o o m m m m } {
  \ldjudge[#1][#2][#3]{\ge #4}{#5}{#6 : #7}
}

\DeclareDocumentCommand{\ldSctxeq} { o o m m m m } {
  \ldpSjudge[#1][#2]{#3}{#4}{#5 \approx #6}
}

\DeclareDocumentCommand{\ltSjudge}{ o o o m m m } {
  \triplejudge[#1][#2][#3]{\vDash_{#4}^{#5} #6}
}

\DeclareDocumentCommand{\ltStypwf} { o o m m m m m } {
  \ltSjudge[#1][#3][#2]{#4}{#5}{#6 \at{#7}}
}

\DeclareDocumentCommand{\ltStypeq} { o o m m m m m m } {
  \ltSjudge[#1][#3][#2]{#4}{#5}{#6 \approx #7 \at{#8}}
}

\DeclareDocumentCommand{\ltSctxwf} { o m m } {
  \lpgenjudge[#2][#1]{\vDash_{p}^{p}}{#3}
}

\DeclareDocumentCommand{\ltSctxeq} { o m m m } {
  \lpgenjudge[#2][#1]{\vDash_{p}^{p}}{#3 \approx #4}
}

\DeclareDocumentCommand{\ldSsubeq} { o o m m m m m m } {
  \ldSjudge[#1][#3][#2]{#4}{#5}{#6 \approx #7 : #8}
}

\DeclareDocumentCommand{\ldSgctx}{ o m } {
  \IfNoValueTF {#1}
  {\genjudge{L}{\sttstileA}{#2}}
  {\genjudge{#1}{\sttstileA}{#2}}
}

\DeclareDocumentCommand{\ldSgctxeq}{ o m m } {
  \IfNoValueTF {#1}
  {\genjudge{L}{\vDash}{#2 \approx #3}}
  {\genjudge{#1}{\vDash}{#2 \approx #3}}
}

\DeclareDocumentCommand{\ldSgsubst}{ o o m m } {
  \lpgenjudge[#1][#2]{\sttstileA}{#3 : #4}
}

\DeclareDocumentCommand{\ldSgsubeq}{ o o m m m } {
  \lpgenjudge[#1][#2]{\sttstileA}{#3 \approx #4 : #5}
}

\DeclareDocumentCommand{\ldssubst} { o o o m m m  } {
  \dsjudge[#1][#2][#3]{_{#4} #5 : #6}
}

\DeclareDocumentCommand{\ldstypwf} { o o o m m m  } {
  \dsjudge[#1][#2][#3]{_{#4} #5 \at{#6}}
}

\DeclareDocumentCommand{\ldstyping} { o o o m m m m } {
  \dsjudge[#1][#2][#3]{_{#4} #5 : #6 \at{#7}}
}

\DeclareDocumentCommand{\ltsemjudge} { o o o m m m } {
  \triplejudge[#1][#2][#3]{\sttstile{}{}_{#4}^{#5}}{#6}
}

\DeclareDocumentCommand{\ltsemtypwf} { o o o m m m m } {
  \ltsemjudge[#1][#2][#3]{#4}{#5}{#6 \at{#7}}
}

\DeclareDocumentCommand{\ltsemtyping} { o o o m m m m m } {
  \ltsemjudge[#1][#2][#3]{#4}{#5}{#6 : #7 \at{#8}}
}

\DeclareDocumentCommand{\ltsemsubst} { o o o m m m m } {
  \ltsemjudge[#1][#2][#3]{#4}{#5}{#6 : #7}
}

\DeclareDocumentCommand{\lsemtyeqp}{ o o m m m m m } {
  \dsemjudge[#1][#2]{_{#3}^{#4} {#5} \approx #6 : #7}
}

\DeclareDocumentCommand{\lsemtyeq}{ o o m m m m } {
\dsemjudge[#1][#2]{_1^{#3} {#4} \approx #5 : #6}
}

\DeclareDocumentCommand{\lsemtypingc}{ o o m m } {
\dsemjudge[#1][#2]{_0^0 {#3} : #4}
}

\DeclareDocumentCommand{\lsemtypingp}{ o o m m } {
  \dualjudge[#1][#2]{;}{\vDash_1^0}{#3 : #4}
}

\DeclareDocumentCommand{\lsemvtyeq}{ o o m m m m } {
  \dSemjudge[#1][#2]{_{#3} #4 \approx #5 : #6}
}

\DeclareDocumentCommand{\msemtyeq}{ o m m m } {
  \msemjudge[#1]{#2 \approx #3 : #4}
}
\DeclareDocumentCommand{\msemtyp}{ o m m } {
  \msemjudge[#1]{#2 : #3}
}

\DeclareDocumentCommand{\mgluty}{ o m o m } {
  \IfNoValueTF {#3}
  {\mjudge[#1]{#2 \; \circledR \; #4}}
  {\mjudge[#1]{#2 \; \circledR_{#3} \; #4}}
}

\DeclareDocumentCommand{\mglutm}{ o m m m o m } {
  \IfNoValueTF {#5}
  {\mjudge[#1]{#2 : #3 \; \circledR \; #4 \in \El(#6)}}
  {\mjudge[#1]{#2 : #3 \; \circledR_{#5} \; #4 \in \El_{#5}(#6)}}
}

\DeclareDocumentCommand{\mglunat}{ o m m } {
  \mjudge[#1]{#2 : \Nat \; \circledR \; #3 \in Nat}
}

\DeclareDocumentCommand{\mglutms}{ o m m m } {
  \mjudge[#1]{#2 : #3 \; \circledR \; #4}
}

\DeclareDocumentCommand{\mglutyu}{ o m o m } {
  \IfNoValueTF {#3}
  {\mjudge[#1]{#2 \; \overline{\circledR} \; #4}}
  {\mjudge[#1]{#2 \; \overline{\circledR}_{#3} \; #4}}
}

\DeclareDocumentCommand{\mglutmu}{ o m m m o m } {
  \IfNoValueTF {#5}
  {\mjudge[#1]{#2 : #3 \; \overline{\circledR} \; #4 \in \El(#6)}}
  {\mjudge[#1]{#2 : #3 \; \overline{\circledR}_{#5} \; #4 \in \El_{#5}(#6)}}
}

\DeclareDocumentCommand{\mglutmd}{ o m m m o m } {
  \IfNoValueTF {#5}
  {\mjudge[#1]{#2 : #3 \; \underline{\circledR} \; #4 \in \El(#6)}}
  {\mjudge[#1]{#2 : #3 \; \underline{\circledR}_{#5} \; #4 \in \El_{#5}(#6)}}
}

\DeclareDocumentCommand{\glunat}{ o m o m } {
  \IfNoValueTF {#3}
  {\judge[#1]{#2 \; \circledR \; #4 \in Nat}}
  {\judge[#1]{#2 \; \circledR_{#3} \; #4 \in Nat}}
}

\DeclareDocumentCommand{\gluty}{ o m o m } {
  \IfNoValueTF {#3}
  {\judge[#1]{#2 \; \circledR \; #4}}
  {\judge[#1]{#2 \; \circledR_{#3} \; #4}}
}

\DeclareDocumentCommand{\glutm}{ o m m m o m } {
  \IfNoValueTF {#5}
  {\judge[#1]{#3 : #2 \; \circledR \; #4 \in #6}}
  {\judge[#1]{#3 : #2 \; \circledR_{#5} \; #4 \in #6}}
}

\DeclareDocumentCommand{\asubtyp}{ o m m } {
  \judge[#1]{_r #2 \subseteq #3}
}
\DeclareDocumentCommand{\agsub}{ o m m } {
  \judge[#1]{_A #2 \subseteq #3}
}

\DeclareDocumentCommand{\ainfer}{ o m m } {
  \judge[#1]{_A #2 \Rightarrow #3}
}
\DeclareDocumentCommand{\acheck}{ o m m } {
  \judge[#1]{_A #2 \Leftarrow #3}
}

\newcommand{\uname}{~{::}~}

\tcbuselibrary{skins}

\newtcbox{\mybox}
{hbox, on line, colback=white, enhanced, frame hidden, boxrule=0pt, top=-1.5pt, bottom=-1.5pt, right=-2pt, left=-2pt, sharp corners, borderline={0.5pt}{0pt}{black,dashed}}

\newtcbox{\mysbox}
{hbox, on line, colback=white, enhanced, frame hidden, baseline=1.6mm, valign=center, height from=0.55cm to 0.65cm, boxrule=0pt, top=-1.2pt, bottom=-0.7pt, right=-1pt, left=-1pt, sharp corners, borderline={0.5pt}{0pt}{black}}

\newtcbox\tcbtp{hbox, on line, colback=light-gray, enhanced, frame hidden, boxrule=0pt, top=-2pt, bottom=-2pt, right=-2.5pt, left=-2.5pt, sharp corners}

\DeclareDocumentCommand{\ulvlhelper}{o m}{
  \IfNoValueTF{#1}
  {
    ^{{\scriptsize\medmuskip=2mu\tcbtp{\ensuremath{#2}}}
      }}
  {\ifnum#1=1
      ^{#2}
    \else{}
    \fi}
}

\DeclareDocumentCommand{\ucolor}{o m}{
  \IfNoValueTF{#1}
  {#2}
  {\ifnum#1=1
      \ca{#2}
    \else{\cb{#2}}
    \fi}
}

\DeclareDocumentCommand{\ulvl}{o o}{
  \IfNoValueTF{#1}
  {\ulvlhelper[#2]{i}}
  {\ulvlhelper[#2]{#1}}
}

\DeclareDocumentCommand{\uwfctx}{o o} {
  \IfNoValueTF{#1}
  {\ensuremath{\ucolor[#2]{\vdash \Gamma}}}
  {\ensuremath{\vdash #1}}
}
\DeclareDocumentCommand{\ueqctx}{o o} {
  \IfNoValueTF{#1}
  {\IfNoValueTF{#2}
    {\ensuremath{\vdash \Gamma \equiv \Delta}}
    {\ensuremath{\vdash \Gamma \equiv #2}}
  }
  {\IfNoValueTF{#2}
    {\ensuremath{\vdash #1 \equiv \Delta}}
    {\ensuremath{\vdash #1 \equiv #2}}}
}
\DeclareDocumentCommand{\utyping}{ o m o m o o} {
  {\ucolor[#6]{\judge[#1]{#2 :\ulvl[#3][#5] #4}}}
}
\DeclareDocumentCommand{\uutyping}{ o m m } {
  {\judge[#1]{#2 : #3}}
}
\DeclareDocumentCommand{\usubst}{ o m m } {
  {\judge[#1]{#2 : #3}}
}
\DeclareDocumentCommand{\usubstwk}{ o m m } {
  \IfNoValueTF {#1}
  {\Gamma \vdash #2 : #3}
  {#1 \vdash #2 : #3}
}
\DeclareDocumentCommand{\ueqtyp}{ o m m o m o o} {
  \ucolor[#7]{\judge[#1]{#2 \equiv #3 :\ulvl[#4][#6] #5}}
}
\DeclareDocumentCommand{\uueqtyp}{ o m m m } {
  \judge[#1]{#2 \equiv #3 : #4}
}

\DeclareDocumentCommand{\utypingsound}{ m m o m o} {
  {#1 \Vdash {#2 :\ulvl[#3][#5] #4}}
}

\DeclareDocumentCommand{\utypingcomp}{ m m o m} {
  {#1 \vDash {#2 :\ulvl[#3][1] #4}}
}
\DeclareDocumentCommand{\ueqtypcomp}{ m m m o m o} {
  {#1 \vDash {#2 \equiv #3 :\ulvl[#4][#6] #5}}
}

\DeclareDocumentCommand{\ueqeractx}{o o} {
  \IfNoValueTF{#1}
  {\IfNoValueTF{#2}
    {\ensuremath{\ca{\Gamma} \lfloor\equiv\rfloor \cb{\Gamma'}}}
    {\ensuremath{\ca{Gamma} \lfloor\equiv\rfloor #2}}
  }
  {\IfNoValueTF{#2}
    {\ensuremath{\ca{#1} \lfloor\equiv\rfloor \cb{#1'}}}
    {\ensuremath{\ca{#1} \lfloor\equiv\rfloor \cb{#2}}}}
}
\DeclareDocumentCommand{\ueqsubst}{ o m m m } {
  \judge[#1]{#2 \equiv #3 : #4}
}
\DeclareDocumentCommand{\ueqeratyp}{ o m m } {
  {\judge[#1]{\ca{#2} \lfloor \equiv \rfloor \cb{#3}}}
}
\DeclareDocumentCommand{\ueqerasubst}{ o m m } {
  {\judge[#1]{\ca{#2} \lfloor \equiv \rfloor \cb{#3}}}
}

\DeclareDocumentCommand{\ueralvl}{m o}{
  \IfNoValueTF{#2}
  {\lfloor \ca{#1} \rfloor = \cb{#1'}}
  {\lfloor #1 \rfloor = #2}
}

\newcommand{\shortsearrow}{%
  \mathrel{\vcenter{\hbox{\scalebox{0.7}{$\searrow$}}}}%
}

\DeclareDocumentCommand{\uset}{o}{
  \IfNoValueTF{#1}
  {\texttt{Set}_i}
  {\texttt{Set}_{#1}}
}

\newcommand{\ulogrelt}[3]{}

\newcommand{\synid}{\textsf{id}}
\newcommand{\synwk}{\textsf{wk}}
\newcommand{\synze}{\texttt{zero}}
\newcommand{\synsu}[1]{\texttt{succ}~#1}

\newcommand{\wfctx}[1]{\ensuremath{\vdash #1}}
\newcommand{\wfctxsub}[2]{\ensuremath{\vdash #1 \subseteq #2}}
\newcommand{\wfctxeq}[2]{\ensuremath{\vdash #1 \approx #2}}
\DeclareDocumentCommand{\substyping}{ o m m } {
  \judge[#1]{#2 : #3}
}
\DeclareDocumentCommand{\subeq}{ o m m m } {
  \judge[#1]{#2 \approx #3 : #4}
}

\newcommand{\dpi}[3]{\ensuremath{\Pi(#1, #2, #3)}}
\newcommand{\dlam}[2]{\ensuremath{\Lambda(#1, #2)}}
\newcommand{\dU}[1]{\ensuremath{\textsf{U}@#1}}
\newcommand{\dN}{\ensuremath{\textsf{N}}}
\newcommand{\dze}{\ensuremath{\textsf{ze}}}
\newcommand{\dsu}[1]{\ensuremath{\textsf{su}(#1)}}
\newcommand{\dneut}[2]{\ensuremath{\uparrow^{#1}\!(#2)}}
\newcommand{\dreif}[2]{\ensuremath{\downarrow^{#1}\!(#2)}}
\newcommand{\dvar}[1]{\ensuremath{l_{#1}}}
\newcommand{\dnatrec}[5]{\ensuremath{\textsf{rec}(#1, #2, #3, #4, #5)}}
\newcommand{\denvext}[2]{\ensuremath{\textsf{ext}(#1, #2)}}
\newcommand{\denvdrop}[1]{\ensuremath{\textsf{drop}(#1)}}

\newcommand{\evald}[3]{\ensuremath{\intp{#1}(#2) \shortsearrow #3}}
\newcommand{\evalf}[2]{\ensuremath{\intp{#1}(#2)}}
\newcommand{\evalsd}[3]{\ensuremath{\intp{#1}_{\mathsf{s}}(#2) \shortsearrow #3}}
\newcommand{\rnf}[3]{\ensuremath{\textsf{R}^{\textsf{Nf}}_{#1}(#2) \shortsearrow #3}}
\newcommand{\rne}[3]{\ensuremath{\textsf{R}^{\textsf{Ne}}_{#1}(#2) \shortsearrow #3}}
\newcommand{\rtyd}[3]{\ensuremath{\textsf{R}^{\textsf{Ty}}_{#1}(#2) \shortsearrow #3}}
\newcommand{\initenv}[1]{\ensuremath{\uparrow^{#1}}}

\newcommand{\nbed}[4]{\ensuremath{\nbe^{#2}_{#1}(#3) \shortsearrow #4}}
\newcommand{\nbetyd}[3]{\ensuremath{\nbe^{\mathsf{ty}}_{#1}(#2) \shortsearrow #3}}
\newcommand{\nbetyf}[2]{\ensuremath{\nbe^{\mathsf{ty}}_{#1}(#2)}}

\newcommand{\inper}[3]{\ensuremath{#1 \approx #2 \in #3}}
\newcommand{\perm}[4]{\ensuremath{#2 \approx #3 \in \Uc_{#1} \shortsearrow #4}}
\newcommand{\peruni}[3]{\ensuremath{#2 \approx #3 \in \Uc_{#1}}}
\newcommand{\Neud}{\ensuremath{\textsl{Neu}}}
\newcommand{\Natd}{\ensuremath{\textsl{Nat}}}
\newcommand{\perctxd}[3]{\ensuremath{#1 \approx #2 \shortsearrow #3}}
\newcommand{\perctxsub}[2]{\ensuremath{#1 <: #2}}
\newcommand{\releq}{\ensuremath{\Leftrightarrow}}

\DeclareDocumentCommand{\semsub}{ o m m } {
  \semjudge[#1]{#2 \subseteq #3}
}
\newcommand{\semctx}[1]{\ensuremath{\vDash #1}}
\newcommand{\semctxeq}[2]{\ensuremath{\vDash #1 \approx #2}}
\newcommand{\semctxsub}[2]{\ensuremath{\vDash #1 \subseteq #2}}

\DeclareDocumentCommand{\sndtyping}{ o m m } {
  \Semjudge[#1]{#2 : #3}
}
\DeclareDocumentCommand{\sndsubtyping}{ o m m } {
  \Semjudge[#1]{#2 : #3}
}
\newcommand{\sndctx}[1]{\ensuremath{\Vdash #1}}
\newcommand{\glumodels}[4]{\ensuremath{#4 \in \Uc^{*}_{#1} \shortsearrow #2 \shortsearrow #3}}
\DeclareDocumentCommand{\gluwk}{ o m m } {
  \judge[#1]{_{w} #2 : #3}
}

\DeclareDocumentCommand{\anfsub}{ o m m } {
  \judge[#1]{^{\Nf}_{r} #2 \subseteq #3}
}
\newcommand{\rfsub}[2]{\ensuremath{\vdash^{\Nf}_{rf} #1 \subseteq #2}}
\DeclareDocumentCommand{\arfsub}{ o m m } {
  \judge[#1]{_{rf} #2 \subseteq #3}
}
\DeclareDocumentCommand{\arfinfer}{ o m m } {
  \judge[#1]{_{rf} #2 \Rightarrow #3}
}
\DeclareDocumentCommand{\arfcheck}{ o m m } {
  \judge[#1]{_{rf} #2 \Leftarrow #3}
}

\newcommand{\nfaeq}[2]{\ensuremath{#1 \simeq #2}}
\newcommand{\compat}[2]{\ensuremath{#1 \mathbin{\dot{\sim}} #2}}
\newcommand{\dcompat}[2]{\ensuremath{#1 \mathbin{\dot{\sim}} #2}}
\newcommand{\daeq}[2]{\ensuremath{#1 \mathbin{\dot{\simeq}} #2}}
\newcommand{\nfsize}[1]{\ensuremath{|#1|}}
\newcommand{\dcompatrel}[3]{\ensuremath{#1 \mathbin{\dot{\sim}} #2 \shortsearrow #3}}

\newcommand{\qsub}[1]{\ensuremath{\mathsf{q}(#1)}}

\newcommand{\drecd}[6]{\ensuremath{\textsf{rec}~#1~\llbracket #2 \mid #3 \mid #4 \rrbracket_{#5} \shortsearrow #6}}

\newcommand{\gluctx}[2]{\ensuremath{#1 \in \mathcal{G} \shortsearrow #2}}

\newenvironment{rulepar}
  {\par\addvspace{1.5pt}\small\begin{mathparshrink}}
  {\end{mathparshrink}\par\addvspace{1.5pt}}
\DeclareDocumentCommand{\sndsub}{ o m m m }{\judge[#1]{#2 \; \circledR \; #3 \in #4}}

\newcommand{\khlmath}[1]{\hlmath{#1}}
\newcommand{\phlmath}[1]{{\color{gray!60}\ensuremath{#1}}}

\newcommand{\mechx}[2]{#1}
\newcommand{\nbesize}[1]{\ensuremath{|#1|}}

\setlist[itemize]{leftmargin=*}
\setlist[enumerate]{leftmargin=*}

\usetikzlibrary{arrows.meta,calc}

%% file: sections/introduction.tex
\section{Introduction}\label{sec:intro}

Type theory is the theoretical foundation behind several popular proof assistants, e.g. Rocq~\citep{rocq}, Agda~\citep{agda},
Lean~\citep{DBLP:conf/cade/MouraKADR15,DBLP:conf/cade/Moura021}, and
Idris~\citep{DBLP:journals/jfp/Brady13}. %
In a type theory, the programming language and the proof language are unified as one, coined as the principle of Curry-Howard Isomorphism. %
The principle enables the reading of a type as a logical proposition, and therefore, a proof is simply a term inhabiting that type. %
This principle is powerful: we can use one algorithm for checking the well-typedness of programs and the validity of proofs. %

Expressing programs and proofs in one language, however, does come at a cost. %
Unlike classical higher-order logic, where quantifiers can quantify over any higher-order predicate, in type theory, we are forced to organize types in a stratified universe hierarchy due to its computational nature, so that a $\Pi$ type (the type-theoretic counterpart for a universal quantifier) can only quantify types at lower universe levels. %
Typically, universes are organized in two ways, cumulative and non-cumulative. %
In a cumulative universe hierarchy, e.g. in Rocq and Idris, a type, say $\Nat : \Ty{0}$, automatically resides at all higher universe levels, i.e. $\Nat : \Ty{n}$ for any $n$. %
Meanwhile, in a non-cumulative setting, e.g. in Agda and Lean, every type is locked at one universe level, so the flexibility in universe levels is exchanged for the type uniqueness property, stating that a well-typed term has one unique type up to syntactic equivalence. %
Nevertheless, cumulativity is convenient and matches the intuition for many in a set-theoretic reading, where universes are organized as an infinite tower of inclusions. 

In reflection of a series of inclusions of universes, the syntactic theory of a type theory often incorporates \emph{subtyping}~\citep{mctt,DBLP:journals/corr/abs-1710-03912,DBLP:phd/hal/Laurent25,luo_coercive_1999,laurent_functoriality_2024}. %
The subtyping judgment relaxes universe levels to arbitrarily higher ones and allows this relaxation to propagate under $\Pi$ types. %
Nevertheless, in an actual implementation of a dependent type theory with universe subtyping, e.g. in Rocq, subtyping is defined rather conservatively. %
The $\Pi$ subtyping rule is typically defined \emph{covariantly}, only propagating the subtyping relation in the output types:
\begingroup
\setlength{\abovedisplayskip}{2pt plus 1pt minus 1pt}
\setlength{\belowdisplayskip}{2pt plus 1pt minus 1pt}
\[
    \inferrule
    {\tyequiv{A}{A'}{\Ty i} \\ \subtyp[\Gamma, x : A']{B}{B'}}
    {\subtyp{\usynPi{x}{A}{B}}{\usynPi{x}{A'}{B'}}}
\]
\endgroup
In this rule, the input types $A$ and $A'$ are equivalent, and only $B$ is required to be a subtype of $B'$ in an extended context with $A'$. %
Semantically, this rule is coherent with the set-theoretic interpretation~\citep{DBLP:journals/corr/abs-1710-03912} of $\Pi$ types as sets of tuples of inputs and outputs. %
This subtyping rule effectively relaxes the second component types of the tuples, and therefore is naturally correct. %
Whereas in the contravariant case, the rule requires increasing care:
\begingroup
\setlength{\abovedisplayskip}{2pt plus 1pt minus 1pt}
\setlength{\belowdisplayskip}{2pt plus 1pt minus 1pt}
\[
    \inferrule
    {\khlmath{\subtyp{A'}{A}} \\ \typing[\Gamma, x: A]{B}{\Ty i} \\ \subtyp[\Gamma, x : A']{B}{B'}}
    {\subtyp{\usynPi{x}{A}{B}}{\usynPi{x}{A'}{B'}}}
\]
\endgroup
The subtyping judgment is now contravariant, because we also compare the input types. %
The input types $A'$ is a subtype of $A$, even though it comes from the right-hand side in the conclusion. %
We are no longer sure about the well-formedness of $B$ judgmentally, so we must add that as the second premise. %
Finally, $B$ and $B'$ are still compared in the same way, but since $A' \subseteq A$ now holds, this comparison in fact is conducted in a \emph{narrower} or sub-context for $B$. %
The shrink of the context of $B$ and $B'$ causes a major challenge in our technical development. %
Contravariant subtyping also has unexpected effects on the syntactic equivalence relation and $\beta\eta$ normal forms. %
Consider a function $f : \usynPi{x}{\Ty 1}{\Ty 1}$. %
Normalizing this function gives us an $\eta$-expanded form: $\lambda (x: \Ty 1). f~x$. %
This normal form stays the same with covariant subtyping: it has type $\usynPi{x}{\Ty 1}{\Ty n}$ for any $n$. %
With contravariant subtyping, the normal form must take input types into account. %
If we explicitly annotate $f : \usynPi{x}{\Ty{\khlmath{0}}}{\Ty 2}$, then the correct normal form becomes $\lambda (x: {\Ty{\khlmath{0}}}). f~x$. %
In addition, the $\eta$ expansion rule allows us to show that the equivalence between two normal forms: $\lambda (x: \Ty 1). f~x \approx f \approx \lambda (x: \Ty 0). f~x:  \usynPi{x}{\Ty 0}{\Ty 2}$.  %
This newly introduced phenomenon shows that a term can be equivalent to two unequal ($\eta$-long) normal forms; contravariant subtyping breaks the unique representation of normal forms. %

In this paper, we give a complete and mechanized account for contravariant universe subtyping. %
We develop \mclttsub, based on \mcltt~\citep{mctt}, a full verification pipeline for a correct implementation of a proof assistant in Rocq. %
One immediate benefit of extending \mcltt is a verified and efficient executable for our extension. %
The changes to the syntactic system, the normalization-by-evaluation (NbE) algorithm, and the semantic definitions are moderate. %
We adapt the semantic subtyping to also consider subtyping of input types. %
The challenge arrives when we have to fix the $\Pi$ case in the subtyping algorithm. %
In \mcltt, the algorithm proceeds by first normalizing both types, and then structurally comparing the results based on their shapes.  %
In the $\Pi$ case, since the input types are required to be equivalent, the output types can be safely compared recursively. %
This algorithm has been proven sound and complete w.r.t. to the covariant subtyping judgment. %
However, in \mclttsub, the completeness of this algorithm is no longer clear. %
Consider two types in normal forms, ${\usynPi{x}{A}{B}}$ and ${\usynPi{x}{A'}{B'}}$. %
Contravariance requires to first compare $A'$ with $A$. %
Then the algorithm would compare $B$ and $B'$, but $B$ comes from a normalization in an extended context with $A$, while the judgment requires a comparison in a context with the subtype $A'$. %
It is therefore not by an immediate argument why it suffices to compare $B$'s normal type obtained in a super-context. %
This observation thus motivates the first subtyping algorithm: we just re-normalize $B$ again in a context extended with $A'$. %
This algorithm effectively forces completeness onto the $\Pi$ case, but iterative renormalizations induce an increased computational cost. %

Revising the syntactic judgments and the normalization algorithm, we notice that (1) universe subtyping only alters the universe levels; in particular, the shapes of types are unchanged, and (2) the normalization process is almost oblivious to universe levels. %
Hence, we postulate that normalizing a term in a sub-context always produces a normal form that only differs in $\lambda$'s type annotation. %
Nonetheless, proving this property using a semantic method is of great difficulty, because the subtyping algorithm effectively requires to compare two results of normalization in a sub-context and in a super-context.  %
Whereas, the existing semantic approaches are only capable of reasoning about normalizations in equivalent ones. %
Luckily, we further realize that different normalization passes in sub-contexts share \emph{the same execution traces}. %
Therefore, we can \emph{bisimulate} the program traces, and prove the desired property via an elementary program equivalence proof in a pleasantly unexpected place. %
This approach leads us to simple completeness and soundness proofs of the subtyping algorithm, eliminating the needs to do renormalization. %

In summary, this paper makes the following contributions:
\begin{enumerate}[nosep,leftmargin=*]
    \item We develop \mclttsub by extending \mcltt to support contravariant subtyping for $\Pi$ types. %
    We adapt the syntactic theory, the normalization-by-evaluation (NbE) algorithm, and its completeness and soundness proof to this setting. %
    Leveraging \mcltt's verification pipeline, we extract a verified type-checker against the declarative syntactic judgments, which allows us to inspect contravariant subtyping in real action. 
    \item We develop a subtyping algorithm with \emph{renormalization}. %
    Renormalization occurs to the output types in the $\Pi$ case, and ensures the completeness of the algorithm. %
    However, since the NbE algorithm already produces $\beta\eta$ normal forms,  we observe that this algorithm is not optimal in computational cost. %
    \item We develop another subtyping algorithm which removes the renormalization step. %
    The algorithm simply recurses on the structure of normal forms. %
    It is a moderate extension of the original subtyping algorithm from \mcltt, hence computationally efficient. %
    We prove that this algorithm is complete via a bisimulation proof by keeping track of program traces of NbE, which successfully circumvents the difficulties of designing a semantic model. %
    We believe that program equivalence is a technique rarely used in type theory, hence providing a fresh perspective in this area. 
\end{enumerate}

The rest of the paper proceeds as follows. %
\Cref{sec:overview} gives a high-level overview of the ideas in this paper. %
\Cref{sec:syntax} presents \mclttsub, our extension of MLTT with contravariant universe subtyping, and its syntactic properties. %
\Cref{sec:nbe-per} recalls the NbE algorithm together with the PER model, and adapts the semantic subtyping to the contravariant setting. %
\Cref{sec:algo} explains how NbE yields a type-checking procedure modulo a correct subtyping algorithm. %
\Cref{sec:renorm,sec:renorm-free} develop two subtyping algorithms, one with renormalization and one without, and prove their correctness. %
\Cref{sec:discussion,sec:related} discuss our findings and related work, and \Cref{sec:conclusion} concludes.
Our mechanization in Rocq, developed with the assistance of Claude Opus, can be found in the accompanying artifact. %

%% file: sections/overview.tex
\section{Bird's-eye View of \mclttsub}\label{sec:overview}

Before diving into the technical details, this section introduces \mclttsub and the key ideas behind.

\subsection{\mclttsub by Examples}

\begingroup

One immediate advantage that we can take by extending \mcltt is a verified executable. %
We can feed our previous example in \Cref{sec:intro} to \mclttsub to see contravariant subtyping in real action: 
\begin{codeexample}
Parsed:
fun (f : forall (x : Type@1) -> Type@1) -> f
  : forall (f : forall (x : Type@1) -> Type@1) (x : Type@0) -> Type@2
Normalized Result:
fun (x1 : forall (A1 : Type@1) -> Type@1) \colorbox{yellow!30}{(A2 : Type@0)} -> x1 A2
  : forall (x1 : forall (A1 : Type@1) -> Type@1) (A2 : Type@0) -> Type@2
\end{codeexample}
In spite of the name shuffling by the pretty printer, we can indeed see that the function \texttt{f} (renamed to \texttt{x1}) is $\eta$-expanded with an input argument \texttt{A2 : Type@0}, aligned with our previous analysis. %
The following example shows that contravariant subtyping is successfully triggered in a function application \texttt{x2 x1}. %
The argument $x1$ is $\eta$-expanded with an argument \texttt{A3 : Type@0}: 
\begin{codeexample}
Normalized Result:
  fun (x1 : forall (A1 : Type@1) -> Type@1)
      (x2 : forall (x3 : forall (A2 : Type@0) -> Type@2) -> Type@3)
    -> x2 \colorbox{yellow!30}{(fun (A3 : Type@0) -> x1 A3)}
  : forall (x1 : forall (A1 : Type@1) -> Type@1)
      (x2 : forall (x3 : forall (A2 : Type@0) -> Type@2) -> Type@3) -> Type@3
\end{codeexample}
As the final example, we inspect the interactions between the recursor and $\eta$ expansion. %
The example program introduces a bogus \texttt{zero} test to interfere with the expansion process. %
Morally, it encodes \texttt{if x == zero then fun z -> z else fun z -> z}. %
A legal \mclttsub expression requires an explicit motive and the use of the recursor:
\begin{codeexample}
Parsed:
  fun (x : Nat)
    -> rec x return y . forall (z : Type@1) -> Type@2
       | zero => fun (z : Type@2) -> z
       | succ n, r => fun (z : Type@1) -> z end
  : forall (x : Nat) (z : Type@0) -> Type@3
\end{codeexample}
We invoke the recursor \texttt{rec} on \texttt{x} to simulate the \texttt{if} expression, and specify the motive to be \texttt{y . forall (z : Type@1) -> Type@2} where the scrutinee \texttt{y} is unused.  %
In the \texttt{zero} branch, we provide an identity function where the input type is at level \texttt{2} to trigger contravariance. %
In the \texttt{succ} branch, the identity function has the input type at level \texttt{1}. %
Finally, we ask the result of the elimination to have type \texttt{forall (z : Type@0) -> Type@3}. %
\mclttsub normalizes this term to
\begin{codeexample}
Normalized Result:
  fun (x1 : Nat) \colorbox{yellow!30}{(A1 : Type@0)}
    -> (rec x1 return x2 . forall (A2 : Type@1) -> Type@2
        | zero => fun \colorbox{yellow!30}{(A3 : Type@1)} -> A3
        | succ x3, x4 => fun (A4 : Type@1) -> A4
        end) A1
  : forall (x1 : Nat) (A1 : Type@0) -> Type@3
\end{codeexample}
The normalizer rearranges both identity functions to take \texttt{Type@1}, and the overall recursion is $\eta$-expanded to take an input of \texttt{Type@0} as required by the final type. %
Note that with contravariance, the motive \emph{is not} altered; expressions in recursor arms are checked against the motive rather than the overall type. 

\endgroup

\subsection{A Journey to Contravariance in $\Pi$ types}\label{sec:overview:journey}

The original \mcltt only support \emph{covariant} subtyping\footnote{To be precise, the subtyping is invariant on input types and covariant on output types.} for $\Pi$ types, where the input types are compared \emph{equivalent}. %
This exact formalization is also employed in Rocq. %
Covariant subtyping leads to a subtyping algorithm that is easy to justify: a simple recursive comparison after normalization. %
The comparison between the output types is clearly complete, because they both live in equivalent contexts. %
A question naturally arises: why stop at covariance? %
According to the set-theoretic semantics of the Calculus of Inductive Constructions (CIC) outlined by~\citet{DBLP:journals/corr/abs-1710-03912},   subtyping is interpreted as set inclusions. %
Therefore, assigning the semantics of $\Pi$ types to sets of input and output pairs is coherent with this interpretation. %
Unfortunately, contravariance is incompatible with this semantics, and thus does not have a justifying semantic model. %

In this paper, we give a fully mechanized account for contravariance. %
Extending \mcltt to \mclttsub first requires changes in the syntactic theory, the semantic partial equivalence relation (PER) model, and the correctness proof of the NbE algorithm. %
These changes are moderate. %
The change to the syntactic theory only involves updating the subtyping rule for $\Pi$ types to compare the input types contravariantly. %
This update is also reflected in the semantic subtyping predicate. %
This change complicates several properties related to semantic subtyping, e.g. transitivity. %
Since input and output types are now compared in different directions, transitivity requires a proof technique akin to that of transitivity in System F$_{<:}$~\citep{DBLP:conf/tacs/CardelliMMS91}, where we must do induction on a measure predicate on the middle type. %
Fixing semantic judgment paves a smooth way to establishing the completeness and soundness proof of the NbE algorithm. %

Nevertheless, it is not straightforward to adjust the type-checker, central to which is the subtyping algorithm. %
In \mcltt, the subtyping algorithm is organized as follows:
\begingroup
\setlength{\abovedisplayskip}{2pt plus 1pt minus 1pt}
\setlength{\belowdisplayskip}{2pt plus 1pt minus 1pt}
\[
    \inferrule
    {\vdash^\Nf_A{\nbetyf \Gamma A} \subseteq {\nbetyf \Gamma B}}
    {\agsub{A}{B}}
    \qquad
    \inferrule
    {W_1 = W_1' \\ \vdash^\Nf_A{W_2} \subseteq {W_2'}}
    {\vdash^\Nf_A{\usynPi{x}{W_1}{W_2}} \subseteq {\usynPi{x}{W_1'}{W_2'}}}
\]
\endgroup
where $\nbetyf \Gamma A$ normalizes $A$ as a type in context $\Gamma$. %
Our initial attempt is to replace the premise $W_1 = W_1'$ with $\vdash^\Nf_A{W_1'} \subseteq {W_1}$. %
However, we immediately bump into a problem with the completeness theorem,  which states that all syntactic subtypes are admitted by the algorithm. %
In the $\Pi$ case, we are not able to apply the induction hypothesis, because $W_2$ is a normal type obtained in $\Gamma, x : W_1$, while the IH requires a type in $\Gamma, x : W_1'$, which is now a sub-context. %
This mismatch means that with contravariance, we need to reason about normal forms obtained from \emph{unequivalent} contexts, which is beyond what the current semantic toolkit is capable of. 

\subsection{Algorithm 1: Renormalization}

The key insight comes from analyzing why the IH got stuck: if the IH requires a normal type from a sub-context, we simply give it by running normalization again. %
Thus the first working algorithm that we propose is to renormalize the output types in the subtyping algorithm. %
Specifically,
\begingroup
\setlength{\abovedisplayskip}{2pt plus 1pt minus 1pt}
\setlength{\belowdisplayskip}{2pt plus 1pt minus 1pt}
\[
    \inferrule
    {\anfsub{W_1'}{W_1} \\ \anfsub[\Gamma, x : W_1']{\nbetyf{\Gamma, x : W_1'}{W_2}}{W_2'}}
    {\anfsub{\usynPi{x}{W_1}{W_2}}{\usynPi{x}{W_1'}{W_2'}}}
\]
\endgroup
Instead of comparing $W_2$ with $W_2'$ directly, we renormalize $W_2$ in the desired context $\Gamma, x : W_1'$ before the comparison. The IH then applies, and completeness follows. %
This renormalization strategy, however, is better suited for \emph{weak head} normalization: in that setting, the output types are not yet fully normalized, so it is natural to resume normalization in the right context and compare the results, exactly as done by \citet{abel_normalization_2017}. %
In our setting, by contrast, the NbE algorithm already returns $\beta\eta$ normal forms, so renormalizing an already normalized term incurs unwanted computational cost. %
Nonetheless, this algorithm still serves as a baseline verified checker that handles contravariant subtyping. Running it 
on several examples also provides us with insights to develop a renormalization-free algorithm discussed in the next subsection.

\subsection{Algorithm 2: Renormalization-free}

Our initial attempt in \Cref{sec:overview:journey} is a more suitable algorithm in many ways: it has a simpler implementation, and is more efficient. %
We would like to show that normalization in a sub-context results in a normal form only differs in universe levels. %
Yet we need to reason about normalization passes in different contexts. %
A semantic model for this property must simultaneously subsume semantic subtyping and generalize the PER model to pair semantic types at different universe levels. %
Giving a definition that satisfies both requirements is already challenging; a fully semantic proof in addition requires the whole proof arch to define semantic judgments, and to prove fundamental theorems, after which finally can we conclude our desired property. %
The effort is quite laborious. 

Light is shed when we think out of the box. %
We would like to take advantage of (1) that ultimately subtyping only compares universe levels, and (2) that cumulativity allows normalization to neglect universe levels most of the time. %
Therefore, instead of insisting on a typical semantic approach, we take an unexpected route: a program equivalence proof via \emph{bisimulation} relations. %
The bisimulation relations are defined in two layers: in the first layer, we relate type annotations that may differ in universe levels, and in the second, we relate semantic values that must have equal universe levels. %
Then we prove that bisimulated, type-annotated values have equal normal forms up to $\lambda$'s type annotations. %
Essentially, the bisimulation relations capture the program traces of the NbE algorithm. %
We are left to show that normalization in sub-contexts is captured by the bisimulation, using which we conclude that the desired subtyping algorithm is indeed complete. %

%% file: sections/syntax.tex
\section{MLTT with Contravariant Universe Subtyping}\label{sec:syntax}

In this section, we define the syntax and the declarative judgments of our
type theory: a core MLTT with natural numbers, $\Pi$ types, explicit
substitutions, a full cumulative universe hierarchy, and (the focus of this
paper) a \emph{contravariant} subtyping judgment for $\Pi$ types. The system
extends \mcltt~\cite{mctt}, and the syntax mostly follows it.

\subsection{Syntax}

\begingroup
\setlength{\abovedisplayskip}{2pt plus 1pt minus 1pt}
\setlength{\belowdisplayskip}{2pt plus 1pt minus 1pt}
\[
  \begin{array}{l@{\quad}l@{~}c@{~}l}
    \multicolumn{4}{l}{\text{Variables}\quad x, y, r \qquad \text{De Bruijn indices}~ \N \ni d \qquad \text{Universe levels}~ \N \ni i, j, k} \\
    \text{Contexts}         & \Gamma , \Delta  & ::= &
    \quad \cdot \mid \Gamma , x : A
    \\
    \text{Expressions}      & A, B, C,      & ::= & \quad
    x_d \mid \Nat \mid \Ty{i} \mid \usynPi{x}{A}{B}
    \\ & M, N && \quad \mid \synze \mid \synsu M \mid \recn{x.A}{M_{\mathsf{z}}}{y,r.M_{\mathsf{s}}}{N}
    \\ &&& \quad \mid \usynlam{x}{A}{M} \mid M~N \mid M[\sigma]
    \\
    \text{Substitutions}    & \sigma , \delta & ::= & \quad \synid \mid \synwk \mid \sigma\circ \delta \mid \sigma , M/x_0
    \\
    \text{Normal forms}     & V, W             & ::= & \quad  U \mid \Nat \mid \Ty{i} \mid \usynPi{x}{W_1}{W_2} \mid \synze \mid \synsu W \mid \usynlam{x}{W_1}{W_2}
    \\
    \text{Neutral forms}    & U                      & ::= & \quad x_d \mid U~W \mid \recn{x.W}{W_{\mathsf{z}}}{y,r.W_{\mathsf{s}}}{U}
  \end{array}
\]
\endgroup

In the mechanization, variables are de Bruijn indices; in the text we attach
abstract names for readability and write the index as a subscript $x_d$ only
when it matters. We use $M$ and $N$ for expressions that are terms, and $A$,
$B$, and $C$ for those that are types; as usual in dependent type theories,
both inhabit one syntactic sort. The universe hierarchy $\Ty{i}$ is indexed by
meta-level natural numbers. The recursor
$\recn{x.A}{M_{\mathsf{z}}}{y,r.M_{\mathsf{s}}}{N}$ eliminates the natural
number $N$ with motive $x.A$, base case $M_{\mathsf{z}}$, and step case
$y,r.M_{\mathsf{s}}$ where $y$ binds the predecessor and $r$ the recursive
result. Functions are Church-style: $\lambda$-abstractions carry their input-type
annotation. Substitutions are explicit, following
\citet{DBLP:journals/jfp/AbadiCCL91}: identity $\synid$, weakening $\synwk$,
composition $\sigma \circ \delta$, and extension $\sigma, M/x_0$, which sends
the topmost variable to $M$ and behaves as $\sigma$ on the rest.

Normal and neutral forms are mutually defined as the usual restriction of
expressions: neutral forms are variables and eliminations blocked by a neutral
scrutinee; normal forms are types and introduction forms with normal
components, together with neutral forms.

\subsection{Judgments}

Our system is given by seven \emph{mutually defined} judgments: context
well-formedness $\wfctx{\Gamma}$ and context subtyping
$\wfctxsub{\Gamma}{\Delta}$ ($\Gamma$ is a sub-context of $\Delta$); typing
$\typing{M}{A}$ and substitution typing $\substyping{\sigma}{\Delta}$ ($\sigma$
is a well-formed substitution from $\Delta$ to $\Gamma$); the two corresponding
equivalences $\tyequiv{M}{M'}{A}$ and $\subeq{\sigma}{\sigma'}{\Delta}$; and
subtyping $\subtyp A B$ ($A$ is a subtype of $B$ in $\Gamma$). All are taken in
an ambient context $\Gamma$.
The typing and equivalence rules for terms and substitutions are standard and
identical to \mcltt~\cite{mctt}; \Cref{app:declarative} collects them in full,
so we do not repeat them here. \Cref{fig:subtyping}
presents the rules that involve subtyping, with two highlighting
conventions used throughout the paper. Premises in
\phlmath{\text{light gray}} are redundant in the presence of the
presupposition lemma (\Cref{thm:syntax:presup}); they are included so that the
metatheory developed by induction over syntactic judgments can have additional induction hypotheses given by them. 
Premises on a \khlmath{\text{gray background}} (\emph{highlighted}) mark the substance
of this work: the contravariant changes relative to \mcltt.

\begin{figure}
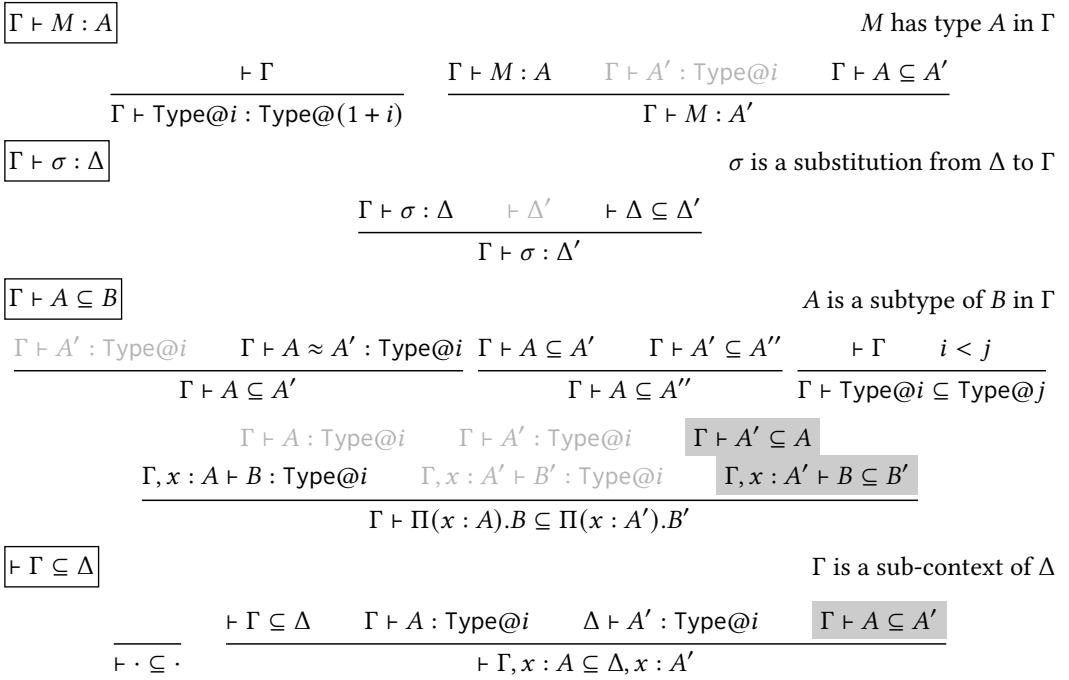

  \makeatletter\setlength{\mpr@andskip}{0.25em plus 0.3fil minus 0.2em}\makeatother
  \raggedright
  \mysbox{$\typing{M}{A}$} \hfill $M$ has type $A$ in $\Gamma$
  \begin{mathparshrink}
    \inferrule
    {\wfctx \Gamma}
    {\typing{\Ty i}{\Ty{(1 + i)}}}

    \inferrule
    {\typing M A \\
     \phlmath{\typing {A'} {\Ty i}} \\
     \subtyp A {A'}}
    {\typing M {A'}}

  \end{mathparshrink}
  \figrulesep

  \mysbox{$\substyping{\sigma}{\Delta}$} \hfill $\sigma$ is a substitution from $\Delta$ to $\Gamma$
  \begin{mathparshrink}
    \inferrule
    {\substyping \sigma \Delta \\
     \phlmath{\wfctx{\Delta'}} \\
     \wfctxsub{\Delta}{\Delta'}}
    {\substyping{\sigma}{\Delta'}}
  \end{mathparshrink}
  \figrulesep

  \mysbox{$\subtyp A B$} \hfill $A$ is a subtype of $B$ in $\Gamma$
  \begin{mathparshrink}
    \inferrule
    {\phlmath{\typing{A'}{\Ty i}} \\
     \tyequiv{A}{A'}{\Ty i}}
    {\subtyp{A}{A'}}

    \inferrule
    {\subtyp{A}{A'} \\ \subtyp{A'}{A''}}
    {\subtyp{A}{A''}}

    \inferrule
    {\wfctx \Gamma \\ i < j}
    {\subtyp{\Ty i}{\Ty j}}

    \inferrule
    {\phlmath{\typing{A}{\Ty i}} \\
     \phlmath{\typing{A'}{\Ty i}} \\
     \khlmath{\subtyp{A'}{A}} \\\\
     \typing[\Gamma, x : A]{B}{\Ty i} \\
     \phlmath{\typing[\Gamma, x : A']{B'}{\Ty i}} \\
     \khlmath{\subtyp[\Gamma, x : A']{B}{B'}}}
    {\subtyp{\usynPi{x}{A}{B}}{\usynPi{x}{A'}{B'}}}
  \end{mathparshrink}
  \figrulesep

  \mysbox{$\wfctxsub{\Gamma}{\Delta}$} \hfill $\Gamma$ is a sub-context of $\Delta$
  \begin{mathparshrink}
    \inferrule
    {~}
    {\wfctxsub{\cdot}{\cdot}}

    \inferrule
    {\wfctxsub{\Gamma}{\Delta} \\
     \typing{A}{\Ty i} \\
     \typing[\Delta]{A'}{\Ty i} \\
     \khlmath{\subtyp{A}{A'}}}
    {\wfctxsub{\Gamma, x : A}{\Delta, x : A'}}
  \end{mathparshrink}
  \vspace{-9pt}%
  \caption{Subsumption, subtyping, and context subtyping rules (selected rules).}
  \vspace{-5pt}%
  \figtailspace
  \label{fig:subtyping}
\end{figure}

\vspace{-5pt}
\paragraph{Subtyping}
The reflexivity, transitivity, and universe rules are standard: reflexivity
embeds definitional equivalence into subtyping (which
also subsumes the usual type-conversion rule), transitivity composes, and
$\subtyp{\Ty i}{\Ty j}$ for $i < j$ delivers cumulativity. The sole
departure is the $\Pi$ rule, which replaces the equivalent-input-type premise by
the contravariant comparison $\subtyp{A'}{A}$ motivated in
\Cref{sec:intro,sec:overview}, and then compares output types in the
narrowed context $\Gamma, x : A'$. Thus contravariance makes context subtyping
central: the premise $\subtyp{A'}{A}$ induces the genuine sub-context
$\wfctxsub{(\Gamma, x : A')}{(\Gamma, x : A)}$. Once narrowing is established
(\Cref{thm:syntax:narrowing}), the version that compares output types in the
wider context, $\subtyp[\Gamma, x : A]{B}{B'}$, is admissible.

\vspace{-0.5em}
\subsection{Syntactic Properties}\label{sec:syntax:prop}

We now collect the syntactic properties of the declarative system used
throughout the paper. They are proved by structural induction on derivations.
The main structural property is narrowing: all judgments are preserved when the
context shrinks along context subtyping.

\Needspace{3\baselineskip}
\vspace{-3pt}%
\begin{theorem}[\mechx{Narrowing}{ctxsub\_*}]\label{thm:syntax:narrowing}
  Suppose $\wfctxsub{\Delta}{\Gamma}$.
  \setlength{\multicolsep}{2pt}
  \begin{multicols}{2}
  \begin{itemize}[nosep,leftmargin=*]
  \item If $\typing M A$, then $\typing[\Delta] M A$;
  \item If $\tyequiv{M}{M'}{A}$, then $\tyequiv[\Delta]{M}{M'}{A}$;
  \item If $\substyping{\sigma}{\Gamma'}$, then $\substyping[\Delta]{\sigma}{\Gamma'}$;
  \item If $\subeq{\sigma}{\sigma'}{\Gamma'}$, then
    $\subeq[\Delta]{\sigma}{\sigma'}{\Gamma'}$;
  \item If $\subtyp A B$, then $\subtyp[\Delta] A B$.
  \end{itemize}
  \end{multicols}
\end{theorem}
\vspace{-3pt}%

Context conversion is the special case of \Cref{thm:syntax:narrowing} for
equivalent contexts. Cumulativity is subtyping; subsumption recovers level
lifting:

\vspace{-3pt}%
\begin{lemma}[\mechx{Cumulativity}{wf\_cumu, lift\_exp\_ge}]\label{lem:syntax:cumu}
  If $\typing{A}{\Ty i}$ and $i \le j$, then $\typing{A}{\Ty j}$, and
  similarly for the term equivalence judgment.
\end{lemma}
\vspace{-3pt}%

Presupposition holds in the usual form; for subtyping it places both sides at a
common level.

\vspace{-3pt}%
\begin{theorem}[\mechx{Presupposition}{presup\_*}]\label{thm:syntax:presup}
  \phantom{x}
  \begin{itemize}[nosep,leftmargin=*]
  \item If $\typing M A$, then $\wfctx \Gamma$ and $\typing{A}{\Ty i}$ for some $i$;
  \item If $\tyequiv{M}{M'}{A}$, then $\wfctx \Gamma$, $\typing M A$,
    $\typing{M'}{A}$, and $\typing{A}{\Ty i}$ for some $i$;
  \item If $\subtyp A B$, then $\wfctx \Gamma$ and $\typing{A}{\Ty i}$ and
    $\typing{B}{\Ty i}$ for some common $i$;
  \item If $\wfctxsub{\Gamma}{\Delta}$, then $\wfctx{\Gamma}$ and $\wfctx{\Delta}$.
  \end{itemize}
\end{theorem}
\vspace{-3pt}%

%% file: sections/nbe_per.tex
\section{NbE and its Soundness and Completeness}\label{sec:nbe-per}

This section establishes the NbE infrastructure used by the algorithms: the
untyped domain, evaluation, and readback (\Cref{sec:nbe-per:nbe}), the PER model
and semantic subtyping relation (\Cref{sec:nbe-per:per}), and the completeness
and soundness theorems for NbE
(\Cref{sec:nbe-per:completeness,sec:nbe-per:soundness}). %
The NbE algorithm and the surrounding PER and gluing models follow the
covariant case. %
Contravariance mainly changes the semantic subtyping relation: its $\Pi$ rule
uses a contravariant premise for the input type and interprets the output type
over the narrower input relation. %
In the completeness proof, the influence is mostly local to the $\Pi$ subtyping
case of the fundamental theorem.
In the soundness proof, the same asymmetry forces one induction to show both
that semantic subtyping reflects to syntactic subtyping and that glued terms can
be transported along semantic subtyping.

\subsection{Untyped Domain and NbE}\label{sec:nbe-per:nbe}

\begingroup
\setlength{\abovedisplayskip}{2pt plus 1pt minus 1pt}
\setlength{\belowdisplayskip}{2pt plus 1pt minus 1pt}
\[
  \begin{array}{l@{\quad}l@{~}c@{~}l}
    \text{De Bruijn levels} & \N \ni z \\
    \text{Neutral domain values} & \DomNe \ni v & ::= & \quad \dvar{z} \mid \dnatrec{A}{b}{M_{\mathsf{s}}}{\rho}{v} \mid v~w \\
    \text{Domain values} & \Dom \ni a, b, m, n & ::= & \quad \dneut{a}{v} \mid \dN \mid \dU{i} \mid \dpi{a}{B}{\rho} \mid \dze \mid \dsu{a} \mid \dlam{M}{\rho} \\
    \text{Normal domain values} & \DomNf \ni w & ::= & \quad \dreif{a}{m} \\
    \text{Evaluation environments} & \Env \ni \rho
  \end{array}
\]
\endgroup

Well-typed expressions are evaluated into an untyped domain whose variables
are de Bruijn \emph{levels} $\dvar z$. Unlike de Bruijn indices, which count
inward from the top of ambient context, levels name variables by their absolute position
in the context; extending a context therefore leaves existing
levels unchanged. An environment is a map from de Bruijn indices to domain
values, $\rho : \N \to \Dom$, equipped with two operations: extension and
dropping. The extension $\denvext{\rho}{n}$ maps $0$ to $n$ and $k+1$ to
$\rho(k)$, while dropping removes the first entry, $\denvdrop{\rho}(k)=\rho(k+1)$.
Functions and $\Pi$ types are
defunctionalized into closures $\dlam{M}{\rho}$ and $\dpi{a}{B}{\rho}$
capturing their respective ambient environments. Reflected neutral values $\dneut{a}{v}$ and
reified values $\dreif{a}{m}$ carry a domain type $a$ as a type \emph{tag};
the tag later drives the $\eta$-expansion performed by readback (discussed shortly). We reserve the word
\emph{annotation} for the syntactic input-type annotations on
$\lambda$-abstractions. We omit the routine NbE rules and proof cases for the
natural-number constructors $\dN$, $\dze$, and $\dsu{a}$.

Evaluation $\evald{M}{\rho}{m}$ and its companions, substitution evaluation
and domain application $\dapp{m}{n} \shortsearrow b$, are defined as functional
relations in the standard way. %
The goal of evaluation is to eliminate $\beta$-redexes.
Explicit substitutions are evaluated to environments, which are then used to
evaluate the substituted term.
\Cref{fig:eval} shows a selection of the rules; the full definition is collected in \Cref{app:nbe}. Two rules are worth noting.
First, evaluation of a $\lambda$-abstraction $\intp{\usynlam{x}{A}{M}}(\rho)$ \emph{discards} its input-type
annotation: the closure $\dlam{M}{\rho}$ records only the body. Type
annotations on $\lambda$s are thus never consulted during evaluation; they
reappear only when readback re-derives them from classifying $\Pi$ types.
Second, applying a reflected \emph{neutral} function $\dapp{\dneut{\dpi{a}{B}{\rho}}{v}}{n}$ computes a fresh type
tag for the result by evaluating the output type $B$ in an environment extended with the argument value $n$. Thus tags are not just inert labels: they can trigger computation during
tag-directed readback.

\begin{figure}
  \raggedright
  \mysbox{$\evald{M}{\rho}{m}$\quad$\dapp{m}{n} \shortsearrow b$\quad$\evalsd{\sigma}{\rho}{\rho'}$} \hfill
  \parbox[t]{0.42\linewidth}{\raggedleft evaluation, domain application, and substitution evaluation (selected rules)}
  \begin{mathparshrink}
    \inferrule
    {~}
    {\evald{x_d}{\rho}{\rho(d)}}

    \inferrule
    {\evald{A}{\rho}{a}}
    {\evald{\usynPi{x}{A}{B}}{\rho}{\dpi{a}{B}{\rho}}}

    \inferrule
    {~}
    {\evald{\usynlam{x}{A}{M}}{\rho}{\dlam{M}{\rho}}}

    \inferrule
    {\evald{M}{\rho}{m} \\ \evald{N}{\rho}{n} \\ \dapp{m}{n} \shortsearrow b}
    {\evald{M~N}{\rho}{b}}

    \inferrule
    {\evalsd{\sigma}{\rho}{\rho'} \\ \evald{M}{\rho'}{m}}
    {\evald{M[\sigma]}{\rho}{m}}

    \inferrule
    {\evald{M}{\denvext{\rho}{n}}{m}}
    {\dapp{\dlam{M}{\rho}}{n} \shortsearrow m}

    \inferrule
    {\evald{B}{\denvext{\rho}{n}}{b}}
    {\dapp{\dneut{\dpi{a}{B}{\rho}}{v}}{n} \shortsearrow
     \dneut{b}{v~\dreif{a}{n}}}

    \inferrule{~}{\evalsd{\synid}{\rho}{\rho}}

    \inferrule{~}{\evalsd{\synwk}{\rho}{\denvdrop{\rho}}}

    \inferrule
    {\evalsd{\sigma}{\rho}{\rho'} \\ \evald{M}{\rho}{m}}
    {\evalsd{(\sigma, M/x)}{\rho}{\denvext{\rho'}{m}}}

    \inferrule
    {\evalsd{\tau}{\rho}{\rho_\tau} \\ \evalsd{\sigma}{\rho_\tau}{\rho'}}
    {\evalsd{\sigma \circ \tau}{\rho}{\rho'}}
  \end{mathparshrink}
  \figrulesep

  \mysbox{$\rnf{z}{w}{W}$\quad$\rne{z}{v}{U}$\quad$\rtyd{z}{a}{V}$} \hfill normal, neutral, and type readback (selected rules)
  \begin{mathparshrink}
    \inferrule
    {\rtyd{z}{a}{W_1} \\
     \dapp{m}{\dneut{a}{\dvar z}} \shortsearrow m' \\
     \evald{B}{\denvext{\rho}{\dneut{a}{\dvar z}}}{b} \\
     \rnf{(1+z)}{\dreif{b}{m'}}{W_2}}
    {\rnf{z}{\dreif{\dpi{a}{B}{\rho}}{m}}{\usynlam{x}{W_1}{W_2}}}

    \inferrule
    {\rtyd{z}{a}{W_1} \\
     \evald{B}{\denvext{\rho}{\dneut{a}{\dvar z}}}{b} \\
     \rtyd{(1+z)}{b}{W_2}}
    {\rtyd{z}{\dpi{a}{B}{\rho}}{\usynPi{x}{W_1}{W_2}}}

    \inferrule{\rtyd{z}{a}{A}}{\rnf{z}{\dreif{\dU i}{a}}{A}}

    \inferrule{~}{\rtyd{z}{\dU i}{\Ty i}}
    \quad
    \inferrule{\rne{z}{v}{U}}{\rtyd{z}{\dneut{a}{v}}{U}}
    \quad
    \inferrule{~}{\rne{z}{\dvar{x}}{x_{(z - x - 1)}}}
    \quad
    \inferrule{\rne{z}{v}{U} \\ \rnf{z}{w}{W}}{\rne{z}{v~w}{U~W}}
  \end{mathparshrink}
  \vspace{-9pt}%
  \caption{Evaluation, domain application, substitution evaluation, and
    readback (selected rules).}
  \vspace{-18pt}%
  \figtailspace
  \label{fig:eval}
\end{figure}

Readback maps domain values back to syntactic normal forms and performs type-directed
$\eta$-expansion, threading type tags as it goes. It is three mutually
defined functional relations (shown in \Cref{fig:eval}): $\rnf{z}{w}{W}$ on normal
values, $\rne{z}{v}{U}$ on neutrals, and $\rtyd{z}{a}{V}$ on types, where $z$
counts the enclosing binders so that a level $z'$ reads back as the index
$x_{z-z'-1}$. A universe-tagged value is read back by immediately switching to
type readback, i.e. $\rnf{z}{\dreif{\dU i}{a}}{A}$ if $\rtyd{z}{a}{A}$. %
Thus
readback of types is not a separate post-processing phase but the type case of
normal readback. 
Also note that the universe level $i$ is entirely neglected; due to cumulativity, universe levels are only significant when a universe is read back. %
This observation is critical in our subsequent development of subtyping algorithm, c.f. \Cref{sec:renorm-free}. %
When reading back a value tagged with $\Pi(a, B, \rho)$, the function value $m$ is applied to $\uparrow^a(l_z)$, a fresh de Bruijn level tagged with the input domain type $a$. %
The result of this application is read back together with a tag given by the evaluation of $B$ in an extended environment. %
Note that the $\lambda$-annotation in the result is determined by $a$, the input domain type from the overall tag, not by the
function value itself.

Composing the pieces yields the NbE algorithm. The initial environment
$\initenv{\Gamma}$ evaluates a context by reflecting, for each entry, the
corresponding de Bruijn level at the value of its type.

\vspace{-3pt}%
\begin{definition}[NbE]\label{def:nbe}
  For $\typing M A$, \;
  $\nbed{\Gamma}{A}{M}{W} \coloneqq
  \rnf{\nbesize{\Gamma}}{\dreif{\evalf{A}{\initenv\Gamma}}{\evalf{M}{\initenv\Gamma}}}{W}$.
  For $\typing A {\Ty i}$, \;
  $\nbetyd{\Gamma}{A}{V} \coloneqq \rtyd{\nbesize{\Gamma}}{\evalf{A}{\initenv\Gamma}}{V}$.
\end{definition}
\vspace{-3pt}%

Evaluation, readback, and the composed NbE relations are functional:
outputs are uniquely determined by the input arguments.
Hence, for convenience, we sometimes also use
functional notation, e.g. $\nbetyf{\Gamma}{A}$, to mean $\exists V, \nbetyd{\Gamma}{A}{V}$.

\subsection{PER Model and Model of Subtyping}\label{sec:nbe-per:per}

To prove completeness of NbE, we need a semantic account of when two evaluated
domain values represent the same syntax. The partial equivalence relation (PER) model provides this account: related semantic types
determine the PER relation for their elements, and
related values are read back to the same normal form. We first define three
basic PERs directly from readback: two values are related iff they read back to a
common normal form (neutral form, resp.) at every depth.
\begingroup
\setlength{\abovedisplayskip}{2pt plus 1pt minus 1pt}
\setlength{\belowdisplayskip}{2pt plus 1pt minus 1pt}
\begin{mathparshrink}
  \inferrule
  {\forall z.\ \Rnf_z(w) = \Rnf_z(w')}
  {\inper{w}{w'}{\Nfd}}

  \inferrule
  {\forall z.\ \Rne_z(v) = \Rne_z(v')}
  {\inper{v}{v'}{\Ned}}

  \inferrule
  {\forall z.\ \Rty_z(a) = \Rty_z(a')}
  {\inper{a}{a'}{\Tyd}}
\end{mathparshrink}
\endgroup
These basic readback PERs are viewed as the extrema of the PER model that we are about to define. %
Given their definitions, the central relation is the universe-level-indexed
PER $\perm{i}{a}{a'}{R}$, which says that domain types $a$ and $a'$ are related
at universe level $i$ and that their elements are related by $R$. %
In a typical type-theoretic setup, semantics of types would have been given by an inductive-recursive relation~\citep{DBLP:journals/apal/DybjerS03},
but since \coq does not support induction-recursion, 
following \mcltt~\cite{mctt}, the PER model is encoded in
\coq's impredicative \texttt{Prop}. %
The relation $R$ should be viewed as a result of $a$ and $a'$. %
The PER model is weakly functional, in the sense that the relation $R$ is unique up to logical equivalence, written as
$\releq$. %
All use of $R$ respects logical equivalence.
The neutral, universe, and $\Pi$ rules are shown in \Cref{fig:per-model}; other
cases are deferred to \Cref{app:semantic}. %
The definition proceeds by first well-founded recursion on universe level $i$, and then an inductive definition.
The well-founded recursion is to handle 
the $\dU j$ case at level
$i$, which refers to the PER at a smaller
universe level $j < i$.  The model is unchanged from
the one used in the covariant setting in \mcltt: it is symmetric, transitive, and
cumulative. Its realizability property, stated in
\Cref{thm:per:realizability}, says that the model is ``sandwiched'' by the
readback PERs.

\begin{figure}
  \begin{mathparshrink}
    \inferrule
    {\inper{v}{v'}{\Ned}}
    {\inper{\dneut{a}{v}}{\dneut{a'}{v'}}{\Neud}}
    \qquad

    \inferrule
    {\inper{b}{b'}{\Ned} \\ R \releq \Neud}
    {\perm{i}{\dneut{a}{b}}{\dneut{a'}{b'}}{R}}
    \qquad
    \inferrule
    {j < i \\ R \releq \Uc_j}
    {\perm{i}{\dU j}{\dU j}{R}}

    \inferrule
    {\perm{i}{a}{a'}{R_I} \\
     \phlmath{R_I \text{ is a PER}} \\\\
     \forall \D \uname \inper{n}{n'}{R_I}.\
       \perm{i}{\evalf{B}{\denvext{\rho}{n}}}{\evalf{B'}{\denvext{\rho'}{n'}}}{R_O(\D)} \\\\
     R \releq \{(m, m') \mid
       \forall \D \uname \inper{n}{n'}{R_I}.\
       \inper{\dapp{m}{n}}{\dapp{m'}{n'}}{R_O(\D)}\}}
    {\perm{i}{\dpi{a}{B}{\rho}}{\dpi{a'}{B'}{\rho'}}{R}}
  \end{mathparshrink}
  \vspace{-9pt}%
  \caption{The neutral-value PER $\Neud$ and the
    universe-indexed PER model $\perm{i}{a}{a'}{R}$ (selected rules).}
  \vspace{-10pt}%
  \figtailspace
  \label{fig:per-model}
\end{figure}

\vspace{-3pt}%
\begin{theorem}[\mechx{Realizability}{realize\_per\_univ\_elem\_gen}]\label{thm:per:realizability}
  Suppose $\perm{i}{a}{a'}{R}$. Then
  \begin{itemize}[nosep,leftmargin=*]
  \item $\inper{a}{a'}{\Tyd}$;
  \item if $\inper{v}{v'}{\Ned}$, then $\inper{\dneut{a}{v}}{\dneut{a'}{v'}}{R}$;
  \item if $\inper{m}{m'}{R}$, then $\inper{\dreif{a}{m}}{\dreif{a'}{m'}}{\Nfd}$.
  \end{itemize}
\end{theorem}
\vspace{-3pt}%

\begin{figure}
  \begingroup
  \makeatletter\def\mpr@sep{0.75em}\makeatother
  \begin{mathparshrink}
    \inferrule
    {\inper{v}{v'}{\Ned}}
    {\semsubtyp{\dneut{a}{v}}{\dneut{a'}{v'}}{i}}

    \inferrule
    {i \le j \\ j < k}
    {\semsubtyp{\dU i}{\dU j}{k}}

    \inferrule
    {\khlmath{\semsubtyp{a'}{a}{i}} \\
     \perm{i}{a'}{a'}{R_I} \\
     \peruni{i}{\dpi{a}{B}{\rho}}{\dpi{a}{B}{\rho}} \\
     \peruni{i}{\dpi{a'}{B'}{\rho'}}{\dpi{a'}{B'}{\rho'}} \\
     \forall \inper{n}{n'}{\khlmath{R_I}}.\
        \semsubtyp{\intp{B}(\denvext{\rho}{n})}{\intp{B'}(\denvext{\rho'}{n'})}{i}}
    {\semsubtyp{\dpi{a}{B}{\rho}}{\dpi{a'}{B'}{\rho'}}{i}}
  \end{mathparshrink}
  \endgroup
  \vspace{-9pt}%
  \caption{The model of subtyping $\semsubtyp{a}{b}{i}$ (selected rules).}
  \vspace{-18pt}%
  \figtailspace
  \label{fig:per-subtyp}
\end{figure}

\paragraph{Semantic subtyping}
\Cref{fig:per-subtyp} defines semantic subtyping $\semsubtyp{a}{b}{i}$
between domain types $a$ and $b$ at universe level $i$, mirroring the syntactic subtyping rules.
Unlike the PER model, semantic subtyping does not return another relation to relate domain values; it relates
types only. Neutral types are related when their underlying neutrals are
related in $\Ned$, and $\dU i$ is below $\dU j$ when $i \le j$. The $\Pi$ rule
is the clause that changes for contravariance. Its input types are related
\emph{contravariantly} by the recursive premise
$\semsubtyp{a'}{a}{i}$, whereas in the covariant case this premise is just
PER-relatedness, $\peruni{i}{a}{a'}$. %
Next, the premise $\perm{i}{a'}{a'}{R_I}$ extracts a relation $R_I$ from $a'$, which is used latter for subtyping $B$ and $B'$. %
The
following two premises record that both $\Pi$ domain types are self-related in the PER
model. 
In other words, they are semantically well-formed types. %
The last output-type premise then quantifies
over values related in $R_I$, the relation returned by the PER of the narrower input type $a'$. %
Semantic subtyping has several expected structural properties like transitivity, cumulativity and subsumption, as listed in the following lemma. 

\vspace{-3pt}%
\begin{lemma}[Properties of semantic subtyping]\label{lem:per:subtyp-props}\label{lem:per:elem-subtyping}
  \phantom{x}
  \begin{itemize}[nosep,leftmargin=*]
  \item \mechx{(reflexivity)}{per\_subtyp\_refl} If $\perm{i}{a}{b}{R}$, then
    $\semsubtyp{a}{b}{i}$ and $\semsubtyp{b}{a}{i}$;
  \item \mechx{(transitivity)}{per\_subtyp\_trans} If $\semsubtyp{a_1}{a_2}{i}$ and
    $\semsubtyp{a_2}{a_3}{i}$, then $\semsubtyp{a_1}{a_3}{i}$;
  \item \mechx{(cumulativity)}{per\_subtyp\_cumu} If $\semsubtyp{a}{b}{i}$ and $i \le j$,
    then $\semsubtyp{a}{b}{j}$.
  \item \mechx{(subsumption)}{per\_elem\_subtyping} If
    $\semsubtyp{a}{b}{i}$, $\perm{i}{a}{a'}{R}$, $\perm{i}{b}{b'}{R'}$, and
    $\inper{m}{m'}{R}$, then $\inper{m}{m'}{R'}$.
  \end{itemize}
\end{lemma}
\vspace{-3pt}%

Among these properties, transitivity is the proof most affected by the
contravariant $\Pi$ rule.
In the $\Pi$ case, chaining
$\semsubtyp{\dpi{a_1}{B_1}{\rho_1}}{\dpi{a_2}{B_2}{\rho_2}}{i}$ with
$\semsubtyp{\dpi{a_2}{B_2}{\rho_2}}{\dpi{a_3}{B_3}{\rho_3}}{i}$ is given the
input-type obligations in the flipped direction ($\semsubtyp{a_2}{a_1}{i}$ and $\semsubtyp{a_3}{a_2}{i}$ ), and also the output-type induction hypotheses are stated
over different input PERs ($a_3$ and $a_2$), so a plain induction on either derivation fails.
The mechanization proves
transitivity by induction on an inductively defined well-foundedness
measure on the middle type. %
The measure captures $a_2$ and also all possible evaluations of $B_2$. %
This is a standard technique inspired by the transitivity proof of $F_{<:}$~\citep{DBLP:conf/tacs/CardelliMMS91}.

Finally, the PER model and semantic subtyping are extended to contexts and
environments pointwise: $\perctxd{\Gamma}{\Gamma'}{R}$ relates two contexts and returns
a PER $R$ on environments, and semantic context subtyping
$\perctxsub{\Gamma}{\Gamma'}$ generalizes semantic subtyping similarly. We
show the extension rule of the latter; the
context PER, and the remaining rules are collected in
\Cref{app:semantic}:
\begingroup
\setlength{\abovedisplayskip}{2pt plus 1pt minus 1pt}
\setlength{\belowdisplayskip}{2pt plus 1pt minus 1pt}
\begin{mathparshrink}
  \inferrule
  {\perctxsub{\Gamma}{\Gamma'} \\
   \perctxd{\Gamma}{\Gamma}{R_{T}} \\\\
   \forall \inper{\rho}{\rho'}{R_{T}}.\
     \semsubtyp{\evalf{A}{\rho}}{\evalf{A'}{\rho'}}{i} \\
   \perctxd{\Gamma, x{:}A}{\Gamma, x{:}A}{R_1} \\
   \perctxd{\Gamma', x{:}A'}{\Gamma', x{:}A'}{R_2}}
  {\perctxsub{\Gamma, x{:}A}{\Gamma', x{:}A'}}
\end{mathparshrink}
\endgroup

As in the $\Pi$ case of $\semsubtyp{}{}{}$, the rule packages
self-relatedness witnesses for both extended contexts, and the
quantification ranges over environments related at the \emph{smaller}
context. Semantic context subtyping enjoys the analogue of
\Cref{lem:per:elem-subtyping}: it includes the context PER of the
sub-context into that of the super-context
as well as reflexivity and transitivity.

\subsection{Completeness}\label{sec:nbe-per:completeness}

Completeness goes through the standard PER-based logical relation: every
syntactic judgment has a \emph{semantic} counterpart asserting the
corresponding PER-relatedness in all related environments, a fundamental
theorem sends each derivation to its counterpart, and realizability of the
PER model then extracts equal normal forms. For the present discussion, the
important clauses are semantic equivalence and semantic subtyping. Semantic
equivalence says that evaluating the two terms in any pair of related
environments gives related domain values at the evaluated type:
\begingroup
\setlength{\abovedisplayskip}{2pt plus 1pt minus 1pt}
\setlength{\belowdisplayskip}{2pt plus 1pt minus 1pt}
\[
  \begin{array}{r@{\ }l}
  \semtyeq{M}{M'}{A} \coloneqq{}
  & \exists R.\ \perctxd{\Gamma}{\Gamma}{R} \text{ and } \exists i.\
    \forall \inper{\rho}{\rho'}{R}.\ \exists R'. \\
  & \quad \perm{i}{\evalf{A}{\rho}}{\evalf{A}{\rho'}}{R'}
    \text{ and } \inper{\evalf{M}{\rho}}{\evalf{M'}{\rho'}}{R'}
  \end{array}
\]
\endgroup
The new semantic subtyping
judgment requires both endpoint types to be self-related, and then relates
their evaluated values by semantic typing in every pair of related
environments:
\begingroup
\setlength{\abovedisplayskip}{2pt plus 1pt minus 1pt}
\setlength{\belowdisplayskip}{2pt plus 1pt minus 1pt}
\[
  \begin{array}{r@{\ }l}
  \semsub{A}{A'} \coloneqq{}
  & \exists R.\ \perctxd{\Gamma}{\Gamma}{R} \text{ and } \exists i.\
    \forall \inper{\rho}{\rho'}{R}. \\
  & \quad \peruni{i}{\evalf{A}{\rho}}{\evalf{A}{\rho'}}
    \text{ and } \peruni{i}{\evalf{A'}{\rho}}{\evalf{A'}{\rho'}}
    \text{ and } \semsubtyp{\evalf{A}{\rho}}{\evalf{A'}{\rho'}}{i}.
  \end{array}
\]
\endgroup
Semantic typing $\semtyp{M}{A}$ and the substitution judgments follow the
same template, while the semantic context judgments $\semctx{\Gamma}$ and
$\semctxeq{\Gamma}{\Gamma'}$ are thin wrappers asserting that the context PER
is defined. The fundamental theorem now carries one clause per
syntactic judgment, including the subtyping judgments:
\vspace{-3pt}%
\begin{theorem}[\mechx{Fundamental theorem for completeness}{completeness\_fundamental}]\label{thm:completeness:fundamental}
$ $
  \setlength{\multicolsep}{2pt}
  \begin{multicols}{2}
  \begin{itemize}[nosep,leftmargin=*]
  \item If $\wfctx\Gamma$, then $\semctx{\Gamma}$;
  \item If $\wfctxsub{\Gamma}{\Gamma'}$, then $\semctxsub{\Gamma}{\Gamma'}$;
  \item If $\typing M A$, then $\semtyp M A$;
  \item If $\substyping{\sigma}{\Delta}$, then $\semjudge{\sigma : \Delta}$;
  \item If $\tyequiv{M}{M'}{A}$, then $\semtyeq{M}{M'}{A}$;
  \item If $\subeq{\sigma}{\sigma'}{\Delta}$, then $\semjudge{\sigma \approx \sigma' : \Delta}$;
  \item If $\subtyp A {A'}$, then $\semsub A {A'}$.
  \end{itemize}
  \end{multicols}
\end{theorem}
\vspace{-3pt}%

Completeness is extracted by instantiating the fundamental theorem at the initial environment and realizability of the PER model. 
\vspace{-3pt}%
\begin{theorem}[\mechx{Completeness of NbE}{completeness}]\label{thm:completeness:nbe}
  If $\tyequiv{M}{M'}{A}$, then there is a $W$ such that
  $\nbed{\Gamma}{A}{M}{W}$ and $\nbed{\Gamma}{A}{M'}{W}$.
\end{theorem}
\vspace{-3pt}%

\subsection{Soundness}\label{sec:nbe-per:soundness}

Establishing soundness requires a gluing model that connects semantic domain
values back to syntax. We use the conventional Kripke gluing
construction. We defer the full definitions in
\Cref{app:gluing}; here we only present its signature and a high-level discussion. For a domain type $a$,
the Kripke gluing model $\glumodels{i}{P}{El}{a}$ assigns two predicates to
$a$: the type predicate $P$ classifies syntactic types representing $a$, and the term predicate $El$ relates syntactic terms and
terms and domain values as elements of $a$. We write $\gluty{A}[i]{P}$ when the
syntactic type $A$ is in the type predicate, and
$\glutm{A}{M}{m}[i]{El}$ when the term $M$ is related to the domain value $m$.
At a $\Pi$ type, this term predicate relate functions extensionally by applying
them to glued arguments. The gluing predicates quantify over weakening substitutions, so glued types
and terms is stable under any extended contexts, hence the Kripkeness of the model. The
weakening relation incorporates context subtyping, so the Kripke model is automatically compatible with subtyping.  Context
gluing then generalizes the gluing model for domain types pointwise and relates substitutions to environments. Thus
$\sndtyping{M}{A}$ means that, for every glued
substitution--environment pair $(\sigma,\rho)$, the syntactic term
$M[\sigma]$ is glued to the evaluated value $\evalf{M}{\rho}$ at the type
$A[\sigma]$.

The standard gluing facts are also unchanged. Realizability asserts that readbacks of domain values are definitionally equivalent to their related syntactic terms. Cumulativity raises glued types and terms
to higher universes; for terms, the $\Pi$ case also requires a less intuitive lowering lemma,
because the element predicate occurs contravariantly. These facts let the
subtyping transport lemma below align the universe level of semantic subtyping
with the universe levels of the two gluing derivations.

Thus contravariance does not change the gluing model itself; it only changes how
semantic subtyping is transported through gluing. Subtyping enters soundness
through the typing and substitution subsumption rules, whose semantic cases
transport glued terms along semantic subtyping:

\vspace{-3pt}%
\begin{lemma}[\mechx{Subsumption for gluing}{glu\_rel\_exp\_subtyp}]\label{lem:soundness:subtyp}
  If $\sndtyping M A$, $\sndtyping {A'} {\Ty i}$, and $\subtyp A {A'}$, then
  $\sndtyping M {A'}$; similarly for substitutions along
  $\wfctxsub{\Delta}{\Delta'}$.
\end{lemma}
\vspace{-3pt}%

The key point behind \Cref{lem:soundness:subtyp} is the interaction between
semantic subtyping and gluing. Two facts are needed. The first, escape, turns
semantic subtyping between glued domain types into syntactic subtyping. The
second, transfer, moves a glued term across semantic subtyping. In the
covariant case these facts can be proved sequentially: escape first, then
transfer. This relies on the covariant $\Pi$ rule, where input types are related by PER.  When escape descends under a $\Pi$ to compare output types,
the bound variable is already glued at the input relation expected by the
output predicate.

Contravariance breaks this sequence. The bound variable is glued on the
wrong side of the input subtyping and must be transferred before escape can
continue under the binder. But transfer for $\Pi$ types uses escape, so the
two statements are proved by one induction:

\vspace{-3pt}%
\begin{lemma}[\mechx{Escape and transfer}{glu\_univ\_elem\_per\_subtyp\_typ\_escape\_trm\_if}]\label{lem:soundness:escape-transfer}
  If $\semsubtyp{a}{a'}{i}$, $\glumodels{i}{P}{El}{a}$, and
  $\glumodels{i}{P'}{El'}{a'}$, then
  \begin{itemize}[nosep,leftmargin=*]
  \item (escape) if $\gluty{A}[i]{P}$ and $\gluty{A'}[i]{P'}$, then $\subtyp A {A'}$;
  \item (transfer) if $\gluty{A'}[i]{P'}$, $\subtyp{A}{A'}$, and
    $\glutm{A}{M}{m}[i]{El}$, then $\glutm{A'}{M}{m}[i]{El'}$.
  \end{itemize}
\end{lemma}
\vspace{-3pt}%

\vspace{-3pt}%
\begin{proof}
  By a single induction on the derivation of $\semsubtyp{a}{a'}{i}$, proving
  both statements simultaneously. The mutual calls described above fall
  entirely in the $\Pi$ case, where escape compares the output types of
  $\subtyp{\usynPi{x}{A_1}{B_1}}{\usynPi{x}{A_1'}{B_1'}}$ in the narrower
  context $\Gamma, x : A_1'$; the remaining cases follow componentwise.
\end{proof}
\vspace{-3pt}%

\vspace{-3pt}%
\begin{theorem}[\mechx{Fundamental theorem and soundness of NbE}{soundness}]\label{thm:soundness:nbe}
  If $\typing M A$, then $\sndtyping M A$, and consequently there is a $W$
  with $\nbed{\Gamma}{A}{M}{W}$ and $\tyequiv{M}{W}{A}$.
\end{theorem}
\vspace{-3pt}%

Soundness and completeness yield the following properties: NbE is
idempotent on its own outputs, and subtyping admits an inversion principle
characterizing derivable subtypings.

\vspace{-3pt}%
\begin{lemma}[\mechx{Idempotence of NbE}{nbe\_ty\_idempotent\_on\_nbe\_results}]\label{lem:nbe:idempotent}
  If $\typing{A}{\Ty i}$ and $\nbetyd{\Gamma}{A}{W}$, then $\nbetyd{\Gamma}{W}{W}$.
\end{lemma}
\vspace{-3pt}%

\vspace{-3pt}%
\begin{lemma}[\mechx{Subtyping inversion}{subtyp\_spec}]\label{lem:subtyp-spec}
  If $\subtyp A B$, then one of the following holds, up to $\approx$:
  \begin{itemize}
  \item $\tyequiv{A}{B}{\Ty i}$ for some $i$; or
  \item $A \approx \Ty i$ and $B \approx \Ty j$ with $i \le j$; or
  \item $A \approx \usynPi{x}{A_1}{B_1}$ and $B \approx \usynPi{x}{A_2}{B_2}$
    with $\subtyp{A_2}{A_1}$ and $\subtyp[\Gamma, x : A_2]{B_1}{B_2}$.
  \end{itemize}
\end{lemma}
\vspace{-3pt}%

Other usual metatheoretic properties also come as corollaries, by
the conventional approach: equivalence is decided by normalize-and-compare;
closed terms of type $\Nat$ are definitionally equal to numerals (canonicity);
and the theory is consistent, stated as the
non-existence of $M$ satisfying $\typing[\cdot]{M}{\usynPi{x}{\Ty i}{x}}$. The contravariant
rule changes none of these statements, but it does change their proofs'
shape: with subsumption available, canonical-form analyses must invert typing
modulo subtyping. 

%% file: sections/algo_typing.tex
\section{From NbE to Type Checking}\label{sec:algo}

After re-establishing a complete and sound NbE algorithm, the final step is to verify the type-checker and the subtyping algorithm, a critical dependency of the type checker. %
In this section, we introduce the type checker, which is identical to that of \mcltt. %
Thus, its correctness reduces to the subtyping algorithm it invokes. %
We first give the bidirectional
type-checking rules, which invokes the subtyping algorithm; we then prove that the
type-checking algorithm is complete and sound w.r.t. the syntactic judgments. %
The whole challenge of the verification falls on
\emph{subtyping}, which we describe two algorithms of in the two following sections, one doing
\emph{renormalization} 
(\Cref{sec:renorm}) and the other \emph{renormalization-free} (\Cref{sec:renorm-free}). %
\Cref{sec:semantic-hard} later
analyzes why the typical semantic route to subtyping is too complex and challenging.

\begin{figure}
  \raggedright
  \mysbox{$\acheck{M}{A}$} \hfill $M$ checks against the type $A$ in $\Gamma$
  \begin{mathparshrink}
    \inferrule
    {\ainfer{M}{W} \\ \agsub{W}{A}}
    {\acheck{M}{A}}
  \end{mathparshrink}
  \vspace{-3pt}
  \raggedright
  \mysbox{$\ainfer{M}{W}$} \hfill $M$ synthesizes the normal-form type $W$ in $\Gamma$
  \begin{mathparshrink}
    \inferrule
    {~}
    {\ainfer{\Ty i}{\Ty{(1+i)}}}

    \inferrule
    {x : A \in \Gamma \\ \nbetyd{\Gamma}{A}{W}}
    {\ainfer{x}{W}}

    \inferrule
    {\ainfer{A}{\Ty i} \\\\ \ainfer[\Gamma, x : A]{B}{\Ty j}}
    {\ainfer{\usynPi{x}{A}{B}}{\Ty{(\max(i,j))}}}

    \inferrule
    {\ainfer{A}{\Ty i} \\\\ \ainfer[\Gamma, x : A]{M}{W} \\
     \nbetyd{\Gamma}{A}{V}}
    {\ainfer{\usynlam{x}{A}{M}}{\usynPi{x}{V}{W}}}

    \inferrule
    {\ainfer{M}{\usynPi{x}{V}{W}} \\ \acheck{N}{V} \\
     \nbetyd{\Gamma}{W[\synid, N/x]}{W'}}
    {\ainfer{M~N}{W'}}
  \end{mathparshrink}
  \vspace{-12pt}%
  \caption{Bidirectional type checking, parameterized by $\asubtyp{A}{B}$.}
  \label{fig:bidir}
  \vspace{-12pt}%
  \figtailspace
\end{figure}

The bidirectional type checking algorithm is directly inherited from \mcltt: the inference judgment
$\ainfer{M}{W}$ synthesizes a \emph{normal} type $W$ for $M$, and
the checking judgment $\acheck{M}{A}$ compares a given type
against an inferred one from $M$ (\Cref{fig:bidir}). NbE keeps synthesized types normal; the check rule is the only place to invoke subtyping. %
To support contravariance, we only change the subtyping algorithm. %
The subtyping algorithm $\agsub{A}{B}$ has two concrete instantiations: the renormalizing
algorithm $\asubtyp{A}{B}$ in \Cref{sec:renorm}, the 
renormalization-free algorithm $\arfsub{A}{B}$ in \Cref{sec:renorm-free}. %
For now, let us admit that these two subtyping algorithms are sound and complete, whose proof are discussed in the next two sections, and move on to proving the completeness and soundness of typing.

\vspace{-3pt}%
\begin{theorem}[\mechx{Soundness of algorithmic
    typing}{alg\_type\_\{check,infer\}\_sound,
    alg\_type\_\{check,infer\}\_rf\_sound}]\label{thm:bidir:sound}
  Suppose $\wfctx{\Gamma}$,
  \begin{itemize}[nosep,leftmargin=*]
  \item if $\ainfer{M}{W}$, then $\typing{M}{W}$;
  \item if $\acheck{M}{A}$ and $\typing{A}{\Ty i}$, then $\typing{M}{A}$.
  \end{itemize}
\end{theorem}
\vspace{-3pt}%

The well-formedness premise for type $A$ in the check direction mirrors how a kernel runs: the type checker is
invoked on a context and an expected type that has already been checked well-formed. %
This invariance is
re-established at every recursive call. %
The inference direction needs no such premise as presupposition tells us W is a well-formed type. 

\vspace{-3pt}%
\begin{theorem}[\mechx{Completeness of algorithmic
    typing}{alg\_type\_\{check,infer\}\_complete,
    alg\_type\_\{check,infer\}\_rf\_complete}]\label{thm:bidir:complete}
  Suppose $\typing{M}{A}$ with $M$ substitution-free. Then
  $\acheck{M}{A}$, and there is a $W$ with $\ainfer{M}{W}$ and
  $\agsub{W}{A}$.
\end{theorem}
\vspace{-3pt}%

Completeness is stated for the substitution-free fragment, the fragment received as user inputs. %
It ensures that (1) we always succeed on checking a term against a type it has, and (2) for a fixed term in a fixed context, we always infer the \emph{principal} type, a subtype of every type that the term may have in that context. %
Combining both theorems, we are sure that the extracted type-checker to OCaml is also correct w.r.t. our syntactic judgments. 

%% file: sections/algo_renorm.tex
\section{A Renormalizing Algorithm and its Correctness Proof}\label{sec:renorm}

Our first decision procedure for contravariant
subtyping compares normal forms and \emph{renormalizes}
the output types of $\Pi$. %
Its correctness proof is \emph{purely syntactic}: it only inspects
the syntactic judgments and uses properties of NbE
established in \Cref{sec:nbe-per}. %
Renormalization conceptually simplifies the algorithm, 
but it does introduce unexpected tedium to the transitivity proof, 
where we must maintain the original normal forms as a handle to do induction on. %
We prove soundness and completeness in
\Cref{sec:renorm:algo,sec:renorm:complete} respectively.

\subsection{The Algorithm and Its Soundness}\label{sec:renorm:algo}

\begin{figure}
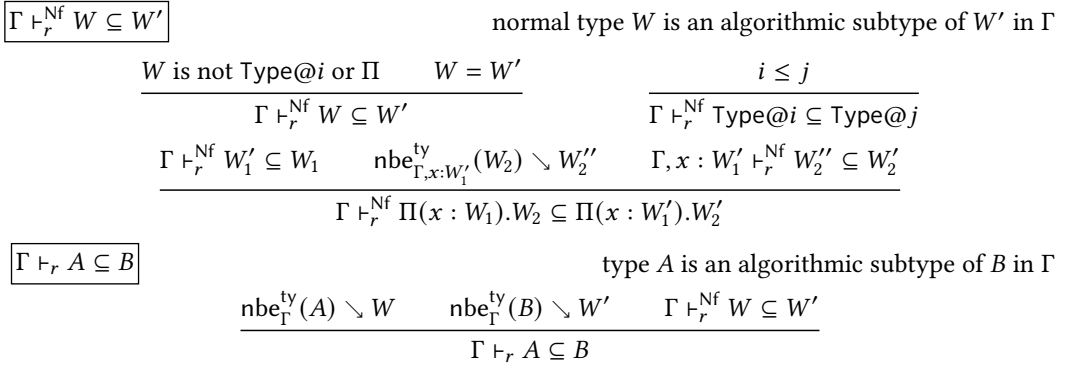

  \raggedright
  \mysbox{$\anfsub{W}{W'}$} \hfill normal type $W$ is an algorithmic subtype of $W'$ in $\Gamma$
  \begin{mathparshrink}
    \inferrule
    {W \text{ is not } \Ty{i} \text{ or } \Pi \\ W = W'}
    {\anfsub{W}{W'}}

    \inferrule
    {i \le j}
    {\anfsub{\Ty i}{\Ty j}}

    \inferrule
    {\anfsub{W_1'}{W_1} \\
     \nbetyd{\Gamma, x : W_1'}{W_2}{W_2''} \\
     \anfsub[\Gamma, x : W_1']{W_2''}{W_2'}}
    {\anfsub{\usynPi{x}{W_1}{W_2}}{\usynPi{x}{W_1'}{W_2'}}}
  \end{mathparshrink}
  \figrulesep
  \raggedright
  \mysbox{$\asubtyp{A}{B}$} \hfill type $A$ is an algorithmic subtype of $B$ in $\Gamma$
  \begin{mathparshrink}
    \inferrule
    {\nbetyd{\Gamma}{A}{W} \\
     \nbetyd{\Gamma}{B}{W'} \\
     \anfsub{W}{W'}}
    {\asubtyp{A}{B}}
  \end{mathparshrink}
  \vspace{-1em}
  \caption{Renormalizing algorithmic subtyping}
  \label{fig:asnf}
  \vspace{-10pt}%
  \figtailspace
\end{figure}

The algorithm $\asubtyp{A}{B}$ presented in \Cref{fig:asnf} first normalizes two types and then compares results structurally by $\anfsub{W}{W'}$. %
Two base cases are straightforward: the reflexivity
compares the equality of normal types and the universe rule compares universe by cumulativity. %
The $\Pi$ case first
recurses contravariantly on the input types. %
However, $W_2$ and $W_2'$ are results of normalization in \emph{different} contexts,
$\Gamma, x : W_1$ and $\Gamma, x : W_1'$ respectively. %
To eliminate the difference, the algorithm \emph{renormalizes} $W_2$ in $\Gamma, x : W_1'$, so that the comparison between the output type $W_2''$ and $W_2'$ is complete and sound.

The soundness theorem of the algorithm is easy to establish.  %
It says that subtypes admitted by the algorithm are also admitted by  declarative subtyping. %
Its proof is a straightforward induction.

\vspace{-3pt}%
\begin{theorem}[\mechx{Soundness}{alg\_subtyping\_nf\_sound, alg\_subtyping\_sound}]\label{thm:renorm:sound}
  If $\asubtyp{A}{B}$ with $A$ and $B$ well-formed at $\Ty i$, then $\subtyp A B$.
\end{theorem}
\vspace{-3pt}%

\subsection{Completeness}\label{sec:renorm:complete}

Completeness is the converse: every declarative subtyping is accepted by the
algorithm. %
Renormalization in the $\Pi$ case shifts the proof complexity to the proof of \emph{transitivity} of $\anfsub{W}{W'}$. %
Central to transitivity is \emph{narrowing}, which states that algorithmic subtyping is stable under renormalization in sub-contexts.

Both lemmas share one difficulty: \emph{a measure to induct
on}. Renormalization breaks structural recursion. %
Instead of sub-terms, recursive calls run on the
\emph{renormalized} output types the $\Pi$ rule produces, and thus we do not have a guaranteed relation between sizes of normal forms before and after renormalization. %
The way out is to keep track of the \emph{original} normal forms computed by $\asubtyp{A}{B}$ and received by $\anfsub{W}{W'}$. %
We call $W$ and $W'$ the \emph{pre-nf}s of the normal forms after renormalization. %
More formally,
\vspace{-3pt}%
\begin{definition}[\mechx{Pre-renormalization normal form}{pre\_renorm\_nf}]\label{def:renorm:source}
  A normal form $V$ is a \emph{pre-nf} of $W$ at $\Gamma$, if
  $\typing{V}{\Ty j}$ and $\nbetyd{\Gamma}{V}{W}$. %
\end{definition}
\vspace{-3pt}%
\noindent That is, $V$ is a well-typed normal form that (re)normalizes to $W$. %
Since the subtyping algorithm proceeds by the shapes of given normal types, pre-nfs shrink for induction hypotheses in the proof. %
For example, in the $\Pi$ case, the algorithm renormalizes the left output type $W_2$ to $W_2''$ and recurses. %
Instead of applying the induction hypothesis on $W_2''$, the size of which has no immediate comparison with $W_2$, we keep track of its pre-nf $W_2$ directly, and all subsequent renormalizations of deeper recursions of $W_2''$ can be tracked by sub-terms of $W_2$ as their pre-nfs. %
The following lemma is the critical lemma which enables this proof method: renormalization of a normal type in a sub-context is identical to direct normalization in that sub-context. 
\vspace{-3pt}%
\begin{lemma}[\mechx{Composition across sub-contexts}{nbe\_ty\_compose\_ctxsub}]\label{lem:renorm:compose}
  If $\wfctxsub{\Gamma'}{\Gamma}$, $\typing{M}{\Ty i}$, $\nbetyd{\Gamma}{M}{W}$, and
  $\nbetyd{\Gamma'}{W}{W'}$, then $\nbetyd{\Gamma'}{M}{W'}$.
\end{lemma}
\vspace{-3pt}%
This lemma directly connects results of renormalization and their pre-nfs. %
Therefore, recursive calls in the subtyping algorithm
operate on strictly smaller pre-nfs, even though the normal forms it compares do not have immediately provable smaller sizes. %
By induction on the sizes of pre-nfs, we prove narrowing:
\vspace{-3pt}%
\begin{lemma}[\mechx{Narrowing}{asnf\_narrow\_gen}]\label{lem:renorm:narrow}
  Let $\subtyp{A'}{A}$. If $V,
  V'$ are pre-nfs of $W, W'$ at $\Gamma, x : A$ respectively, and $\anfsub[\Gamma, x : A]{W}{W'}$,
  then for some $W_0, W_0'$, $\nbetyd{\Gamma, x : A'}{W}{W_0}$,
  $\nbetyd{\Gamma, x : A'}{W'}{W_0'}$, and
  $\anfsub[\Gamma, x : A']{W_0}{W_0'}$.
\end{lemma}
\vspace{-3pt}%
This lemma is at the core of proving transitivity of the subtyping algorithm. %
We use this lemma to show that the subtyping relation is still admitted by the algorithm after renormalization. %
To prove narrowing, we must prove a generalized version instead, because in the $\Pi$ case, the output types are renormalized in extended contexts. %
\Cref{lem:renorm:compose} is applied to realign the pre-nfs and the results of renormalization under sub-contexts.

\vspace{-3pt}%
\begin{lemma}[\mechx{Transitivity}{asnf\_trans\_gen, asnf\_trans}]\label{lem:renorm:trans}
  Let $V, V', V''$ be pre-nfs of $W, W', W''$ at
  $\Gamma$ respectively. If
  $\anfsub{W}{W'}$, and $\anfsub{W'}{W''}$, then $\anfsub{W}{W''}$.
\end{lemma}
\vspace{-3pt}%

\vspace{-3pt}%
\begin{proof}[Proof sketch]
  Induction on $\nfsize{V'}$. The only substantive
  case is the $\Pi$ one. %
  First, the input types recurse contravariantly,
  giving $\anfsub{W_1''}{W_1}$. %
  Then the output types live in
  \emph{different} contexts, $\Gamma, x{:}W_1'$ and $\Gamma, x{:}W_1''$, so
  narrowing carries the first to
  $\Gamma, x{:}W_1''$, and \Cref{lem:renorm:compose} identifies its renormalized types for comparison. %
\end{proof}
\vspace{-3pt}%

With transitivity and \Cref{lem:renorm:compose} in hand, completeness is a
direct induction.

\vspace{-3pt}%
\begin{theorem}[\mechx{Completeness}{completeness\_subtyp}]\label{thm:renorm:complete}
  If $\subtyp{A}{A'}$, then $\asubtyp{A}{A'}$.
\end{theorem}
\vspace{-3pt}%

%% file: sections/algo_renorm_free.tex
\section{A Renormalization-Free Algorithm and its Correctness Proof}\label{sec:renorm-free}

The algorithm of \Cref{sec:renorm} is correct but unsatisfying in two ways: it
re-runs NbE at every $\Pi$ level, and its proof, though conceptually
straightforward, is technically tedious. Analyzing what renormalization actually
does to an already-normal form lets us do better. We present a second decision
procedure that performs \emph{no} renormalization, comparing two given types
purely structurally once they are normalized. Its correctness rests on a
bisimulation argument that runs the original normalization against a
hypothetical renormalization under a narrower context and shows that two
compute the same results. The argument is \emph{semantic} but not in a traditional sense: 
it is a program equivalence proof, which is not a frequently used method in studying type theory.

\subsection{The Algorithm}\label{sec:renorm-free:algo}

Previously, renormalization is introduced to ease the handling of the $\Pi$ case. %
Avoiding the renormalization first requires
understanding what it does to an already-normal form. The guiding claim of this
section, established rigorously below as \Cref{thm:rf:renorm-codomain-ann}, is
that the effect is \emph{almost nothing}: renormalization only changes $\lambda$ annotations. %
This claim is based on two observations. %
(1) The first is that
renormalization preserves a normal form's \emph{shape}. Normalization already yields an $\beta\eta$ normal form, so a second
pass reduces no new $\beta$-redex and can only redo the $\eta$-expansions. The
expansions are directed by shapes of types, which are preserved by the subtyping relation, in the readback functions. %
(2) Empirically, we only find examples where renormalization changes $\lambda$ annotations, but not other components, c.f. our examples in \Cref{sec:overview}. %
Therefore, we conjecture that the only changing components after renormalization are the universe levels in annotations. %
If this conjecture holds, then we can safely proceed in the $\Pi$ case and ignore all $\lambda$ annotations when comparing equality between normal forms. %

\begin{figure}
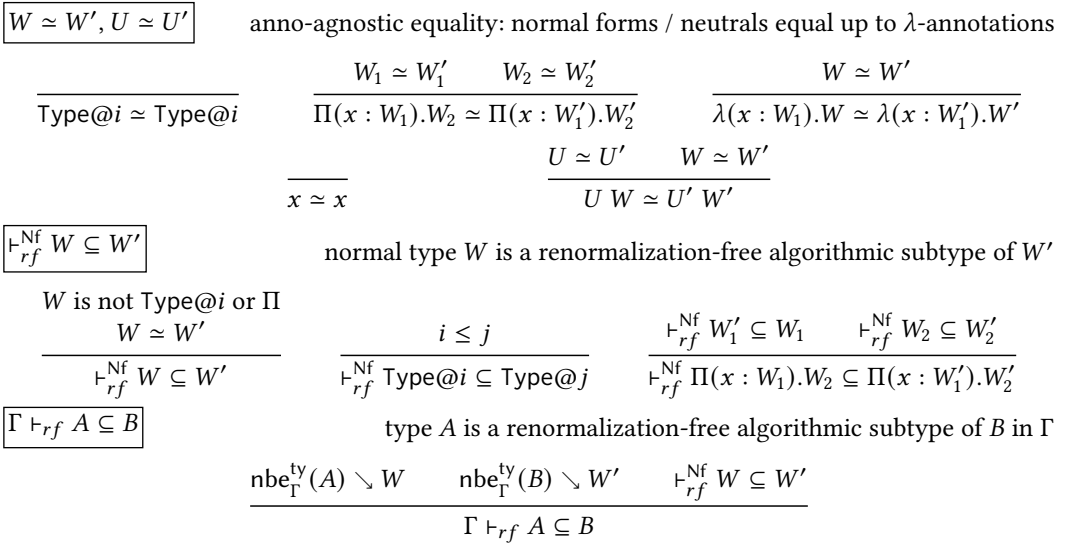

  \raggedright
  \mysbox{$\nfaeq{W}{W'}$,\ $\nfaeq{U}{U'}$} \hfill anno-agnostic equality: normal forms / neutrals equal up to $\lambda$-annotations
  \begin{mathparshrink}
    \inferrule{~}{\nfaeq{\Ty i}{\Ty i}}

    \inferrule{\nfaeq{W_1}{W_1'} \\ \nfaeq{W_2}{W_2'}}
    {\nfaeq{\usynPi{x}{W_1}{W_2}}{\usynPi{x}{W_1'}{W_2'}}}

    \inferrule{\nfaeq{W}{W'}}
    {\nfaeq{\usynlam{x}{W_1}{W}}{\usynlam{x}{W_1'}{W'}}}

    \inferrule{~}{\nfaeq{x}{x}}

    \inferrule{\nfaeq{U}{U'} \\ \nfaeq{W}{W'}}
    {\nfaeq{U~W}{U'~W'}}
  \end{mathparshrink}
  \vspace{-3pt}

  \raggedright
  \mysbox{$\rfsub{W}{W'}$} \hfill normal type $W$ is a renormalization-free algorithmic subtype of $W'$
  \begin{mathparshrink}
    \inferrule
    {W \text{ is not } \Ty{i} \text{ or } \Pi \\\\ \nfaeq{W}{W'}}
    {\rfsub{W}{W'}}
    \qquad
    \inferrule
    {i \le j}
    {\rfsub{\Ty i}{\Ty j}}
    \qquad
    \inferrule
    {\rfsub{W_1'}{W_1} \\ \rfsub{W_2}{W_2'}}
    {\rfsub{\usynPi{x}{W_1}{W_2}}{\usynPi{x}{W_1'}{W_2'}}}
  \end{mathparshrink}
   \vspace{-3pt}
   
  \raggedright
  \mysbox{$\arfsub{A}{B}$} \hfill type $A$ is a renormalization-free algorithmic subtype of $B$ in $\Gamma$
  \begin{mathparshrink}
    \inferrule
    {\nbetyd{\Gamma}{A}{W} \\
     \nbetyd{\Gamma}{B}{W'} \\
     \rfsub{W}{W'}}
    {\arfsub{A}{B}}
  \end{mathparshrink}
  \vspace{-9pt}%
  \caption{Annotation-agnostic equality and renormalization-free subtyping.}
  \label{fig:srf}
  \vspace{-10pt}%
  \figtailspace
\end{figure}

This picture turns directly into an algorithm (\Cref{fig:srf}; full rules for $\nfaeq{}{}$ are in \Cref{app:algo}). To check
$\subtyp A B$, it normalizes both types and compares the results purely structurally: in the $\Pi$ case, it recurses contravariantly on the
input types, and covariantly \emph{directly} on the output types.  %
For universes, it compares
levels ($\Ty i \le \Ty j$). %
In all other cases, it compares the two normal
forms up to $\nfaeq{W}{W'}$ instead of syntactic equality.  %
The \emph{annotation-agnostic equality} $\nfaeq{W}{W'}$
(\emph{anno-agnostic} for short), is almost syntactic equality, except that the
annotations of $\lambda$s are dropped. %
In particular, universe levels are compared equal. %
Then being able to recurse directly on the
output types is exactly what the guiding claim warrants; the mere differences in $\lambda$ annotations
are ignored by the $W \simeq W'$ comparison in the base case.  %
This algorithm has a similar shape to \mcltt and is more efficient than that in \Cref{sec:renorm}. %

\subsection{Soundness}\label{sec:renorm-free:sound}

Since the algorithm only checks equality up to erased $\lambda$-annotations (by $\nfaeq{W}{W'}$), 
we must show this relaxation is compatible with syntactic equivalence. %
The following lemma proves that well-typed $\lambda$ expressions can be compared just by their bodies, regardless of their annotations:
\vspace{-3pt}%
\begin{lemma}[\mechx{Annotation irrelevance}{ann\_irrel\_fn\_cong}]\label{lem:rf:ann-irrel}
  If $\typing{\usynlam{x}{A}{M}}{\usynPi{x}{B}{C}}$,
  $\typing{\usynlam{x}{A'}{M'}}{\usynPi{x}{B}{C}}$, and
  $\tyequiv[\Gamma, x:B]{M}{M'}{C}$, then
  $\tyequiv{\usynlam{x}{A}{M}}{\usynlam{x}{A'}{M'}}{\usynPi{x}{B}{C}}$.
\end{lemma}
\vspace{-3pt}%
\vspace{-3pt}%
\begin{proof}[Proof Sketch]
    Do $\eta$ expansion on both $\lambda$ terms and then use $\beta$ reduction to get rid of $A$ and $A'$. %
    The conclusion holds by congruence. 
\end{proof}
\vspace{-3pt}%

That $\nfaeq{W}{W'}$ and $\nfaeq{U}{U'}$ preserve syntactic equivalence is proved by mutual induction:

\vspace{-3pt}%
\begin{lemma}[\mechx{Annotation erasure is definitional}{aeq\_defeq}]\label{lem:rf:aeq-defeq}
  \phantom{x}
  \begin{itemize}[nosep,leftmargin=*]
  \item If $\nfaeq{W}{W'}$, $\typing{W}{C}$, and $\typing{W'}{C}$, then
    $\tyequiv{W}{W'}{C}$;
  \item If $\nfaeq{U}{U'}$, $\typing{U}{C}$, and $\typing{U'}{C'}$,
    then there is a \emph{common subtype} $D$ with
    $\tyequiv{U}{U'}{D}$, $\subtyp{D}{C}$, and $\subtyp{D}{C'}$.
  \end{itemize}
\end{lemma}
\vspace{-3pt}%

The spine cases are congruences,
and the $\Pi$ case goes through the subtyping inversion principle
(\Cref{lem:subtyp-spec}). Note that the two cases are stated asymmetrically:
the neutral case carries two possibly different types $C, C'$ and concludes only their common
subtype $D$, while the normal-form case uses a single shared $C$. %
This difference is dictated by their usages. %
The normal-form result is used where the overall types are known. %
Meanwhile the neutral result is applied in the elimination case, so we must focus on a common subtype. %
This can be viewed as a stronger invariant supported by the inversion lemmas modulo subtyping of \Cref{sec:syntax:prop}.

\vspace{-5pt}%
\begin{theorem}[\mechx{Soundness}{subtyp\_rf\_sound, alg\_subtyping\_rf\_sound}]\label{thm:rf:sound}
  If $\rfsub{W}{W'}$, $\typing{W}{\Ty i}$, and $\typing{W'}{\Ty i}$, then
  $\subtyp{W}{W'}$. Consequently $\arfsub{A}{B}$ with both sides typed at
  $\Ty i$ implies $\subtyp A B$.
\end{theorem}
\vspace{-5pt}%

\vspace{-5pt}%
\begin{proof}
  Induction on $\rfsub{W}{W'}$. Apply
  \Cref{lem:rf:aeq-defeq} in the base case. %
  In the $\Pi$ case, we use \Cref{thm:syntax:narrowing} to narrow the context for the output types. 
\end{proof}
\vspace{-5pt}%

\vspace{-0.3em}
\subsection{Completeness}\label{sec:renorm-free:engine}\label{sec:renorm-free:complete}

\begin{figure}
  \raggedright
  \mysbox{$\compat{M}{M'}$} \hfill tag equivalence (syntax): like $\nfaeq{}{}$ but levels free
  \begin{mathparshrink}
    \inferrule
    {~}
    {\khlmath{\compat{\Ty i}{\Ty j}}}

    \inferrule
    {\compat{A}{A'} \\ \compat{B}{B'}}
    {\compat{\usynPi{x}{A}{B}}{\usynPi{x}{A'}{B'}}}

    \inferrule
    {\compat{M}{M'}}
    {\compat{\usynlam{x}{A}{M}}{\usynlam{x}{A'}{M'}}}
  \end{mathparshrink}
  \figrulesep
 
  \raggedright
  \mysbox{$\dcompat{a}{a'}$} \hfill tag equivalence (domain): levels free, annotations dropped
  \begin{mathparshrink}
    \inferrule
    {~}
    {\khlmath{\dcompat{\dU i}{\dU j}}}

    \inferrule
    {\dcompat{a}{a'}}
    {\dcompat{\dsu{a}}{\dsu{a'}}}

    \inferrule
    {\dcompat{\rho}{\rho'} \\ \compat{M}{M'}}
    {\dcompat{\dlam{M}{\rho}}{\dlam{M'}{\rho'}}}

    \inferrule
    {\dcompat{a}{a'} \\ \compat{B}{B'} \\ \dcompat{\rho}{\rho'}}
    {\dcompat{\dpi{a}{B}{\rho}}{\dpi{a'}{B'}{\rho'}}}

    \inferrule
    {\dcompat{a}{a'} \\ \dcompat{v}{v'}}
    {\dcompat{\dneut{a}{v}}{\dneut{a'}{v'}}}

    \inferrule
    {\dcompat{a}{a'} \\ \dcompat{m}{m'}}
    {\dcompat{\dreif{a}{m}}{\dreif{a'}{m'}}}
  \end{mathparshrink}
  \figrulesep
  
  \raggedright
  \mysbox{$\daeq{a}{a'}$} \hfill tag-agnostic equivalence (domain only): levels exact, bodies equal, agnostic at tags
  \begin{mathparshrink}
    \inferrule
    {~}
    {\daeq{\dU i}{\dU i}}

    \inferrule
    {\daeq{a}{a'}}
    {\daeq{\dsu{a}}{\dsu{a'}}}

    \inferrule
    {\daeq{\rho}{\rho'}}
    {\daeq{\dlam{M}{\rho}}{\dlam{M}{\rho'}}}

    \inferrule
    {\daeq{a}{a'} \\ \daeq{\rho}{\rho'}}
    {\daeq{\dpi{a}{B}{\rho}}{\dpi{a'}{B}{\rho'}}}

    \inferrule
    {\khlmath{\dcompat{a}{a'}} \\ \daeq{v}{v'}}
    {\daeq{\dneut{a}{v}}{\dneut{a'}{v'}}}

    \inferrule
    {\khlmath{\dcompat{a}{a'}} \\ \daeq{v}{v'}}
    {\daeq{\dreif{a}{v}}{\dreif{a'}{v'}}}
  \end{mathparshrink}
  \vspace{-13pt}%
  \caption{Tag and tag-agnostic equivalences (selected rules).}
  \label{fig:dom-relations}
  \vspace{-10pt}%
  \figtailspace
\end{figure}

Due to the structural simplicity of the algorithm, transitivity is comparatively simple to prove. %
The proof proceeds by induction on the middle type $W'$.
\vspace{-5pt}%
\begin{lemma}[\mechx{Transitivity}{subtyp\_rf\_trans}]\label{lem:rf:trans}
  If $\rfsub{W}{W'}$ and $\rfsub{W'}{W''}$, then $\rfsub{W}{W''}$.
\end{lemma}
\vspace{-5pt}%
Following the analysis of \Cref{sec:renorm-free:algo}, the real challenge for completeness is the
$\Pi$ case: we must show that a direct comparison suffices and renormalization is unnecessary. %
The rest
of this subsection builds the machinery to close this gap by making the guiding claim formal and precise: renormalization 
in a sub-context changes only the $\lambda$ annotations.
This statement corresponds to the following conjecture:

\vspace{-0.3em}
\begin{conjecture}\label{conj:renorm-uni}
    If $\rfsub{V'}{V}$, $\nbetyd{\Gamma, x:V}{B}{W}$, and
  $\nbetyd{\Gamma, x:V'}{B}{W'}$, then $\nfaeq{W}{W'}$.
\end{conjecture}
\vspace{-0.3em}

In this conjecture, we compare $W$ and $W'$, the results of normalizing $B$ in a super- and a sub-context. %
Unfolding the definition of $\nbe$, the conclusion requires a comparison between evaluations in two \emph{different} environments, due to the subtyping relation. %
This capability, unfortunately, is unequipped in the current semantic setup. %
Typically, a type theorist would attempt to construct a stronger, more general semantic model that captures the equivalence under evaluations in environments arising from a context and a narrowed context. %
However, this approach is deemed laborious and was difficult, as we will discuss further in \Cref{sec:discussion}. %

Instead, we approach it with a fresh mind. %
Observe that the readback function $\Rnf$ (\Cref{fig:eval}) simply disregards universe levels in the tags and redirects readback of domain types to $\Rty$: $\Rnf_z(\downarrow^{\dU i}(a)) := \Rty_z(a)$. %
Meanwhile, universe levels matter when they are actually read back to some universes: $\Rty_z(\dU i) := \Ty i$. 
Thus, the universe levels in tags are irrelevant; they are only significant as the tagged. %
This observation is captured by two levels of relations in \Cref{fig:dom-relations}. %
On the tag level, we use the tag equivalence relation, $\mathbin{\dot{\sim}}$, to capture terms and domain objects that are equal up to universe levels. %
The base cases (two topmost highlighted rules in \Cref{fig:dom-relations}) relate all universes regardless of their levels, because levels are irrelevant in tags. %
Omitted cases are trivially congruences. %
We lift tag equivalences point-wise to relate also pairs of substitutions and of environments. %
When the same domain value $m$ is tagged by equivalent $a$ and $a'$ i.e. $\dcompat{a}{a'}$, the readbacks of $m$ are \emph{equal} up to anno-agnostic equality $\nfaeq{}{}$. %

On the tagged value level, tag-agnostic equivalence, $\mathbin{\dot{\simeq}}$, relates domain objects that have same readbacks up to anno-agnostic equality. %
For that reason, it only relates universes with the same levels. %
Note that tag-agnostic equivalence does not relate terms; it suffices to require identical captured terms for proving completeness. %
This design choice is reflected in the $\Lambda$ and the $\Pi$ cases. %
Finally, when we encounter tags, we relate them by tag equivalence. %
This is reflected by the highlighted premise $\khlmath{\dcompat{a}{a'}}$ in the last two rules in \Cref{fig:dom-relations}. %
Both rules say that if the tags lead to the same path during readback, we can conclude that the tagged values are read back the same. 

With tag and tag-agnostic equivalences, we proceed with \Cref{conj:renorm-uni} in the following way. %
Note that $B$ is normalized in two contexts, $\Gamma, x : V$ and $\Gamma, x : V'$. %
Since $V'$ and $V$ only appear in the contexts, they are only used as tags during normalization. %
If they are further related by tag equivalence, then they should drive readback down the same execution path. %
If in addition, evaluation preserves both tag and tag-agnostic equivalences, then evaluating $B$ in two contexts really share an identical execution trace, modulo the minor differences of $\lambda$-annotations, which are safely ignored by anno-agnostic equality. %
In this proof plan, we consider the two traces of program executions, leading to an elementary \emph{bisimulation} argument of our desired property. %
This is a technique rarely seen in type theory; we believe that it can be effective and elegant in many other scenarios that are yet to be discovered. %

Formally, we first aim at the following theorem. 
\vspace{-6pt}%
\begin{theorem}[\mechx{Renormalization up to tags}{renorm\_codomain\_ann}]\label{thm:rf:renorm-codomain-ann}
  If $\compat{A}{A'}$, $\nbetyd{\Gamma, x:A}{B}{W}$, and
  $\nbetyd{\Gamma, x:A'}{B}{W'}$, then $\nfaeq{W}{W'}$.
\end{theorem}
\vspace{-6pt}%
The theorem says that if $A$ and $A'$ are tag-equivalent, then normalizing the same term in their respective extended contexts yields the same result up to anno-agnostic equality. %
Based on our plan, its proof begins with the following series of preservation lemmas. %
The first says that evaluation preserves tag equivalence; the second says that evaluating the same term in tag-agnostically equivalent environments preserves tag-agnostic equivalence; the third says that readback transports tag-agnostic equivalence to anno-agnostic equality. %
The evaluation preservation results generalize and hold simultaneously for substitution
evaluation and domain application/recursion, while the readback result is proved
simultaneously for normal- and neutral-form readback.
\vspace{-6pt}%
\begin{lemma}[\mechx{Preservation of tag and tag-agnostic equivalences}{eval\_preserves\_subtyp\_compat, eval\_preserves\_dom\_aeq, readback\_erases\_dom\_aeq}]\label{lem:rf:preservation}
  \phantom{x}
  \begin{enumerate}[nosep,leftmargin=*]
  \item If $\compat{M}{M'}$, $\dcompat{\rho}{\rho'}$, $\evald{M}{\rho}{m}$, and
    $\evald{M'}{\rho'}{m'}$, then $\dcompat{m}{m'}$.
  \item If
    $\daeq{\rho}{\rho'}$, $\evald{M}{\rho}{m}$, and $\evald{M}{\rho'}{m'}$,
    then $\daeq{m}{m'}$.
  \item If
    $\daeq{a}{a'}$, $\rtyd{z}{a}{V}$, and $\rtyd{z}{a'}{V'}$, then
    $\nfaeq{V}{V'}$.
  \end{enumerate}
\end{lemma}
\vspace{-6pt}%
The proof of the preservation lemmas is essentially a (trivial) bisimulation when environments are viewed as program states. %
The bisimulation is elementary: with tag and tag-agnostic equivalent inputs, the evaluation and the readback functions always take the same step of execution, with the mild exception of $\lambda$-annotations in the $\eta$ expansion step. %
The readback preservation result therefore relates outputs by anno-agnostic equality, which ignores this exceptional annotation difference. 

The tag-equivalence preservation result is slightly more subtle because it compares the evaluation of two related terms, $M$ and $M'$, while evaluation may contract $\beta$-redexes. %
This still follows from the shape information in $\compat M{M'}$: the two terms reduce in lockstep, including the $\beta$ steps. %
Quantifying over arbitrary related terms is necessary because tags may already differ in \Cref{thm:rf:renorm-codomain-ann}, and because $\Lambda$ and $\Pi$ closures may contain unevaluated bodies and ambient environments that are known only up to tag equivalence. %

By contrast, tag-agnostic preservation evaluates the \emph{same} term on both sides, as
\Cref{thm:rf:renorm-codomain-ann} only considers normalization of the same term. %
\Cref{thm:rf:renorm-codomain-ann} can be proved by combining these preservation results:
first, contexts $\Gamma, x:A$ and $\Gamma, x:A'$ are tag-equivalent, so their initial environments are tag-agnostically equivalent. %
Therefore, evaluating $B$ in tag-agnostically equivalent environments yields tag-agnostically equivalent domain values. %
Reading them back yields two anno-agnostically equal normal forms. From this point, we only need to reason about syntactic terms, using only tag equivalence $\compat{A}{A'}$ and anno-agnostic equality $\nfaeq{W}{W'}$. %
The part about domain values is entirely shielded by \Cref{thm:rf:renorm-codomain-ann}.

The following two lemmas further connect both anno-agnostic equality and tag equivalence with the algorithmic subtyping algorithm. The first is a relaxation lemma, which captures the intuition that normal forms admitted by the subtyping algorithm only differ in universe levels and $\lambda$-annotations. %
Equivalently, two algo-subtypes are tag equivalent.  %
This lemma, along with \Cref{thm:rf:renorm-codomain-ann}, closes \Cref{conj:renorm-uni}.
The second is a \emph{stability} lemma : the subtyping algorithm is stable under anno-agnostically equal inputs. %

\vspace{-5pt}%
\begin{lemma}[\mechx{Normal subtypes are tag equivalent}{subtyp\_rf\_exp\_subtyp\_compat}]\label{lem:rf:erase-compat}
  If $\rfsub{W}{W'}$, then $\compat{W}{W'}$.
\end{lemma}
\vspace{-5pt}%

\vspace{-0.5em}

\vspace{-5pt}%
\begin{lemma}[\mechx{Stability under anno-agnostic equality}{subtyp\_rf\_respects\_nf\_aeq}]\label{lem:rf:aeq-stable}
  If $\rfsub{W}{W'}$, $\nfaeq{W}{V}$, and $\nfaeq{W'}{V'}$, then
  $\rfsub{V}{V'}$.
\end{lemma}
\vspace{-5pt}%

Both are proved by structural induction on derivations.  %
Combining all theorems and lemmas above, we are ready to prove \Cref{conj:renorm-uni}, which is really the core of the completeness theorem. 
\vspace{-3pt}%
\begin{theorem}[\mechx{Completeness}{completeness\_subtyp\_rf}]\label{thm:rf:complete}
  If $\subtyp{A}{A'}$, then $\arfsub{A}{A'}$.
\end{theorem}
\vspace{-5pt}%

\vspace{-10pt}%
\begin{proof}
  Induction on $\subtyp{A}{A'}$.
  \emph{Reflexivity}: completeness of NbE (\Cref{thm:completeness:nbe})
  yields a common normal form.
  \emph{Universes}: trivial.
  \emph{Transitivity}: determinism of NbE aligns the middle normal form, and
  \Cref{lem:rf:trans} closes the case.

  The $\Pi$ case is the main case. %
  We compare $\usynPi{x}{A}{B}$ and $\usynPi{x}{A'}{B'}$. %
  The induction hypotheses give
  $\rfsub{\nbetyf{\Gamma}{A'}}{\nbetyf{\Gamma}{A}}$ for the input types and
  $\rfsub{\nbetyf{\Gamma, x : A'}{B}}{\nbetyf{\Gamma, x : A'}{B'}}$ for the output
  types. By the Soundness of NbE and the stability of NbE under equivalent contexts, we know $\nbetyf{\Gamma, x : \nbetyf{\Gamma}{A}}{B} = \nbetyf{\Gamma, x : A}{B}$ (similarly for $\nbetyf{\Gamma, x : A'}{B}$). To show $\rfsub{\nbetyf{\Gamma, x : A}{B}}{\nbetyf{\Gamma, x : A'}{B'}}$, we first apply \Cref{conj:renorm-uni} to get $\nfaeq{\nbetyf{\Gamma, x : A'}{B}}{\nbetyf{\Gamma, x : A}{B}}$, then apply \Cref{lem:rf:aeq-stable} to close the proof. 
\end{proof}
\vspace{-3pt}%

%% file: sections/discussion.tex
\vspace{-0.5em}
\section{Discussions}\label{sec:discussion}

The previous sections develop two complete and sound algorithms, but neither proof comes
from a single fully semantic model. We discuss our previous attempts at
such a semantic proof. We then discuss how the current framework might
extend to other subtyping features.

\vspace{-5pt}
\subsection{Why a fully semantic proof resists}\label{sec:semantic-hard}

In \Cref{sec:renorm-free:complete}, we describe a completeness proof for the renormalization-free subtyping algorithm. %
The primary challenge of the completeness proof is to reason about the evaluation of the same term in two different contexts related by the sub-context relation. %
Unfortunately, this requirement has escaped what the current PER model, as given by \mcltt and in \Cref{sec:nbe-per}, can offer. %
Therefore, the focus of the proof has been to find a way to reason about evaluation in unbalanced environments. %
Our proof is somewhat unexpected: it does not involve constructing a model, in a way similar to \Cref{sec:nbe-per}, but reasons about the program behaviors of the evaluation and readback functions instead. %
This proof method has the evident advantage of simplicity, while it is also not very typical. %
Usually, a type-theoretical approach is fully semantic: one finds a model that is general enough, so that once instantiated, it shows that evaluations of the same term in a sub- and a super-context can be related by anno-agnostic equality (c.f. \Cref{conj:renorm-uni}). 

One closely related work is the development of sized types by \citet{abel_normalization_2017}, where they introduce sized types to MLTT so that the termination checker can admit non-primitive recursions more easily. %
Orthogonal to sized types is their contravariant setup of $\Pi$ types, exactly identical to our setting. %
They have also provided an NbE and a subtyping algorithm, and proved both complete and sound. %
After looking into their development, we see that \citeauthor{abel_normalization_2017}'s NbE algorithm only normalizes to \emph{weak-head} normal forms. %
Their subtyping algorithm first weak-head normalizes both given types and then make the comparison based on their head shapes. %
It is therefore natural to compare output types of $\Pi$ in an extended context with the narrower input type, and rely on the recursion to perform deeper weak-head normalization. %
This setting is analogous to our renormalization algorithm in \Cref{sec:renorm}, in its continuous weak-head normalization in recursive calls. %
In fact, their setting makes the completeness of subtyping even simpler to prove than our renormalization algorithm, because they do not need to reason about the interactions between renormalization and the subtyping algorithm at all. %
Whereas in our case, since our NbE gives full $\beta\eta$ normal forms, this complication is unavoidable and we rely on pre-nfs in \Cref{def:renorm:source} to prove transitivity. %
Furthermore, it is also unclear how to adapt \citeauthor{abel_normalization_2017}'s constructions to prove completeness of the renormalization-free algorithm in \Cref{sec:renorm-free}, where the core comparison algorithm does not (need to) keep track of the ambient context at all. %
In retrospect, we think that full $\beta\eta$ normalization does introduce unwanted complexity in developing an efficient contravariant subtyping algorithm. %

Taking one step back from \citeauthor{abel_normalization_2017}'s work, to construct a fully semantic completeness proof for the algorithm in \Cref{sec:renorm-free}, the semantic relation of two domain types must satisfy two aspects. %
(1) Two related domain types should be read back equal up to anno-agnostic equality. %
This aspect corresponds to the base case of the algorithm, which takes $\eta$ expansion into account and ignores $\lambda$-annotations during equality comparison. %
(2) The relation itself should subsume the subtyping relation and also generalize the PER model. %
The second aspect is the primary source of complication due to its conflicting requirements. %
First, since this relation is used to reason about evaluation in sub-contexts, it must capture the subtyping relation between types that are extended in the context, so it must subsume the semantic subtyping, an asymmetric relation. %
Moreover, we are concerned about the anno-agnostic equality of the readbacks of domain values tagged by related domain types. %
Therefore, this semantic relation must also generalize the PER model, a symmetric relation, so that it has an output relation to relate domain values. %
In summary, the relation should be of the form $\dcompatrel a b R$, so that if $m$ and $m'$ are related by $R$, then the readbacks of $\dreif a m$ and $\dreif b{m'}$ are anno-agnostically equal. %
This desired property is effectively the realizability of the $\dcompatrel a b R$ relation. %
The structure of this relation also has two layers and is similar to tag and tag-agnostic equivalences in \Cref{sec:renorm-free:complete} retrospectively. %
The domain types $a$ and $b$ are only used in tags. %
The main difference in the semantic relation $\dcompatrel a b R$ is its asymmetry, caused by its subsumption of subtyping, which has led to several technical difficulties and conflicts with PER's symmetry. %

In one previous failed attempt, we try to eliminate the asymmetry and reintroduce symmetry by permissively relating universes regardless of their levels in the base case, as in tag equivalence:

\vspace{-0.5em}
\[
  \setlength{\arraycolsep}{2pt}
    \inferrule
    {R \releq \Uc_{\min(i, j)}}
    {\dcompatrel {\dU i}{\dU j} R}
    \quad
    \inferrule
    {\dcompatrel{a}{a'}{R_I} \quad
     \forall \D \uname {n} \sim {n'} \in {R_I}.\
       \dcompatrel{\evalf{B}{\denvext{\rho}{n}}}{\evalf{B'}{\denvext{\rho'}{n'}}}{R_O(\D)} \\\\
     R \releq \{(m, m') \mid
       \forall \D \uname {n} \sim {n'} \in {R_I}.\
       {\dapp{m}{n}} \sim {\dapp{m'}{n'}} \in {R_O(\D)}\}}
    {\dcompatrel{\dpi{a}{B}{\rho}}{\dpi{a'}{B'}{\rho'}}{R}}
\]
\vspace{-0.5em}

At the first glance, this relation resembles the tag equivalence in many ways. %
In the universe case, the returned relation is just the PER model at the smaller level $\min(i, j)$, as we want the readbacks of values related by $R$ to be equal and types at level $\min(i, j)$ can be inhabited at both levels $i$ and $j$. %
For the $\Pi$ case, $\dcompatrel{a}{a'}{R_I}$ returns $R_I$ that relates domain values of types $a$ and $a'$, respectively. %
Then the final returned relation $R$ requires that for all related $n$ and $n'$, applying $m$ and $m'$ to them respectively yields related results. %
This definition does allow us to prove the realizability property specified above; but it falls apart in the transitivity proof in the $\Pi$ case. %
The problem is that this model is overly generalized; transitivity of $\dcompatrel{a}{a'}{R}$ and $\dcompatrel{a'}{a''}{R'}$ derives 
$\dcompatrel{a}{a''}{R''}$ for some $R''$, which
does not maintain any information about $a'$. %
In addition, this relation does \emph{not} subsume subtyping, as $B$ in the $\Pi$ rule is not evaluated with a value of the subtype. %

The semantic relation above is defined with the same underlying idea as the successful proof in \Cref{sec:renorm-free:complete} in mind, but why is it a failed attempt? %
Our understanding is that, in a semantic construction, a model must simultaneously satisfy various properties, including closure under $\Pi$ types, transitivity, and subsumption of semantic subtyping and the PER model, in this particular case.  %
It is therefore challenging and time-consuming to tune the definition to meet these requirements \emph{at once}. %
Whereas in the proof in \Cref{sec:renorm-free:complete}, we bypass the challenge wholesale; instead of meanings of types, we focus on the program behaviors, which is indeed the right insight to finalize the completeness proof. %

Finally, we can imagine that a semantic proof remains possible. %
In fact, we have a promising candidate in mind. %
Consider the second aspect above, which requires the semantic relation to extend semantic subtyping and the PER model. %
It seems sufficient to simply embrace the asymmetry and combine semantic subtyping and the aforementioned failed attempt in one inductive relation:
\begingroup
\setlength{\abovedisplayskip}{2pt plus 1pt minus 1pt}
\setlength{\belowdisplayskip}{2pt plus 1pt minus 1pt}
\[
    \inferrule
    {i \le j \\\\ R \releq \Uc_{\min(i, j)}}
    {\dcompatrel {\dU i}{\dU j} R}
    \qquad
    \inferrule
    {\dcompatrel{a'}{a}{R_I} \\ \semsubtyp{\dpi{a}{B}{\rho}}{\dpi{a'}{B'}{\rho'}}{i} \\\\
     \forall \D \uname {n'} \sim {n} \in {R_I}.\
       \dcompatrel{\evalf{B}{\denvext{\rho}{n}}}{\evalf{B'}{\denvext{\rho'}{n'}}}{R_O(\D)} \\\\
     R \releq \{(m, m') \mid
       \forall \D \uname {n'} \sim {n} \in {R_I}.\
       {\dapp{m}{n}} \sim {\dapp{m'}{n'}} \in {R_O(\D)}\}}
    {\dcompatrel{\dpi{a}{B}{\rho}}{\dpi{a'}{B'}{\rho'}}{R}}
\]
\endgroup

This new definition is more plausible. %
First, it is asymmetric; it is easy to prove the subsumption of semantic subtyping: if $\dcompatrel a b R$, then $\semsubtyp a b i$ for some $i$. %
The asymmetry comes from two places. %
In the universe case, we now require $i \le j$. %
In the $\Pi$ case, to target the subsumption of semantic subtyping, the first premise now compares $a$ and $a'$ contravariantly, i.e. $\dcompatrel{a'}{a}{R_I}$. %
We further force the subtyping of both $\Pi$ types by requiring the second premise. %
The subsumption of semantic subtyping further enables transitivity, and realizability is maintained. %

Nevertheless, we still think that the proof in \Cref{sec:renorm-free:complete} is more elegant and concise. %
Even if the candidate model above is strong enough to establish completeness, a whole semantic setup is still unavoidable: we must prove necessary ambient properties, realizability, fundamental theorems, and finally extract the desired property of \Cref{conj:renorm-uni}. %
This entire roadmap is a detour compared to our present approach. %
We simply use tag and tag-agnostic equivalences to capture two layers of behavioral differences, and prove the preservation lemma (\Cref{lem:rf:preservation}) w.r.t. them. %
The proof complexity is reduced drastically by taking the right approach. %
We hope our approach to simplify other semantic proofs.

\vspace{-0.3em}
\subsection{Scaling beyond contravariant universes}

For a new subtyping feature in a subsumptive subtyping discipline, subsumption
must preserve the term: using a subtype where a supertype is expected inserts no
term-level coercion. The semantic model
therefore has to validate an inclusion like \Cref{lem:per:elem-subtyping}: the
same domain value related at the subtype must also be related at the supertype.
Once this transport is available, the two algorithmic developments pose
different proof obligations. The renormalizing algorithm needs the output-type
comparison to be realignable by another NbE run in the narrowed context, with a
well-founded measure like the pre-renormalization normal forms of
\Cref{def:renorm:source}. The renormalization-free algorithm needs the two
possible normalizations to differ only in observations that the comparison
relation may ignore, as $\nfaeq{}{}$ ignores $\lambda$-annotations in
\Cref{thm:rf:renorm-codomain-ann}.

Record and $\Sigma$ types are the first candidate. We assume the standard syntax for them and $\eta$-rules are considered. Concretely, for record, it is
$M \approx \{\,\ell_i = M.\ell_i\,\}_{i \in I}$ at record type
$\{\ell_i{:}A_i\}_{i \in I}$; for dependent pairs, it is $M \approx (\mathsf{fst}\,M,\mathsf{snd}\,M)$ at $\Sigma(x{:}A).\,B$. To handle these rules, $\Sigma$ and record types need to guide the $\eta$ expansions when they appear in tag positions. The prospective
subtyping rules can be sketched as follows:
\begingroup
\setlength{\abovedisplayskip}{2pt plus 1pt minus 1pt}
\setlength{\belowdisplayskip}{2pt plus 1pt minus 1pt}
\begin{mathparshrink}
  \small
  \inferrule
  {\subtyp{A}{A'} \\ \subtyp{B}{B'}}
  {\subtyp{\{\ell_1{:}A,\,\ell_2{:}B\}}{\{\ell_1{:}A',\,\ell_2{:}B'\}}}
  \and
  \inferrule
  {\subtyp{A}{A'} \\ \subtyp[\Gamma, x : A]{B}{B'}}
  {\subtyp{\Sigma(x{:}A).\,B}{\Sigma(x{:}A').\,B'}}
  \and
  \inferrule{~}
  {\subtyp{\{\ell_1{:}A,\,\ell_2{:}B\}}{\{\ell_1{:}A\}}}
\end{mathparshrink}
\endgroup
The first rule, record depth, is the easy case: subtyping happens
component-wise, and the same record value can be viewed at the supertype. The
algorithmic proof is also straightforward, with no sub-context discrepancy. Dependent $\Sigma$-depth subtyping is less immediate. The displayed rule is the
natural covariant candidate: the first component moves from $A$ to $A'$, and the
second component is compared under $x : A$ because an actual pair of the subtype
contains a first component of type $A$.  Unlike the $\Pi$ rule, however, it is less clear that this is exactly the right
semantic principle for $\Sigma$: the second component of a $\Sigma$ value is
typed at its actual first component, not uniformly at every element of the
domain. The semantic question is whether this rule is adequate as stated, or whether the interpretation of $B'$ at such an
$x$ requires an additional compatibility condition. Algorithmically, if this
rule is validated, it should still be milder than the current $\Pi$ subtyping
case, since pair $\eta$ preserves the pair shape and the rule does not reverse
the first-component subtyping direction.

Record width subtyping is different. Semantically, the displayed rule is reasonable for projection-based records: a value belongs to the narrow record
type when its retained projections have the required PERs, so the same wide
value can be used at the narrow type.
Algorithmically, however, it is more complicated, since subtype-related record types of different
widths can $\eta$-expand the same neutral record to normal forms with different
sets of projected fields. Such a more significant perturbation to ``shapes'' further requires non-trivial adaption of both our current proof techniques.

For a base inclusion such as $\subtyp{\Nat}{\Int}$, the algorithmic side is
benign: the new rule lives at a leaf in the subtyping derivation, does not
trigger any $\eta$-expansion, and is therefore comparable to a simplified
universe subtyping case. The challenge is semantic. In dependent type theory,
values are not untyped objects later classified by types: new data structures
are introduced together with their types by their formation and introduction
rules. Thus, if $\Nat$ and $\Int$ are defined as separate inductive types, a
constructor-built natural number is not already an integer value. Relating the
two then inevitably requires a coercion from $\Nat$ to $\Int$. A subsumptive account would therefore need a different
design in which the two base types are interpreted over a shared representation,
so that every natural-number value is already available as an integer value. 

%% file: sections/related.tex
\vspace{-0.5em}
\section{Related Work}\label{sec:related}

This section reviews related work in four strands: subsumptive and coercive
subtyping, contravariant $\Pi$ subtyping and the procedures that decide it,
universe structure and universe subtyping, and mechanizations of dependent type
theory. 

\vspace{-3pt}
\paragraph{Subsumptive and coercive subtyping}
Subtyping is commonly presented in two styles: \emph{subsumptive} subtyping,
where a term of a subtype may be used at a supertype unchanged, and
\emph{coercive} subtyping, where elaboration inserts an explicit map. Both
styles occur beyond dependent type theory~\cite{cardelli_multiple_1988}. This paper
uses the subsumptive style: semantically, subtyping is an inclusion of elements,
so a value related at the subtype is already related at the supertype. Universe
cumulativity in systems such as \coq{} and \mcltt{} is subsumptive in this
sense, and our semantic coercion is the identity
(\Cref{lem:per:elem-subtyping}). A large body of dependent-type work studies the coercive style.
In Luo's coercive subtyping~\cite{luo_coercive_1999,luo_coercive_2013}, a declared map
$c : A \to B$ lets a term of $A$ stand where a $B$ is expected, but the
elaborated term contains the coercion. This is more flexible than subsumption
but requires coherence of coercion paths; for example, the first projection can
make $\Sigma(x{:}A).\,B$ a subtype of $A$. This style underlies \coq's
implicit-coercion mechanism~\cite{saibi_inheritance_1997}. More recent
dependent-type work makes coercions compute through type constructors, as in
Laurent's functorial coercions~\cite{laurent_functoriality_2024} and
AdapTT~\cite{adjedj_adaptt_2026}; these systems address a different problem
from our transparent, identity-coercion subtyping.

\vspace{-3pt}
\paragraph{Contravariant $\Pi$ subtyping}
Contravariant input-type subtyping has been studied across various systems in
programming language research~\cite{cardelli_multiple_1988,castagna_contravariance_1995,frisch_semantic_2008}; for dependent types, \citet{aspinall_subtyping_2001} decide
subtyping for $\lambda P$ with bounded type constants against a
reduction back end, and \citet{zwanenburg_pts_1999} adds contravariant
subtyping to pure type systems with normalization and subject-reduction
results. The procedures
that do decide the full contravariant rule in a dependent setting are the
closest to this paper. \citet{DBLP:conf/tlca/FridlenderP13} decide an MLTT
with cumulative universes and contravariant $\Pi$ subtyping by
normalize-and-compare over NbE, but their theory omits $\eta$ and its
$\lambda$s are input-annotation-free, so readback is untyped and depends
only on the length of the context, free from the sub-context complexity in the first place.
\citet{scherer_universe_2011} treat contravariant universe subtyping with
$\eta$ and typed equality on paper, and report attempting, but not completing, a \coq{}
mechanization of the core argument. \citet{abel_normalization_2017} extend
NbE to sized dependent types with the full contravariant rule, decided
incrementally on weak-head normal forms. We have already discussed it extensively in \Cref{sec:discussion}.
\citet{liu_conversion_2026} mechanize in \coq{} a sound, complete, and
terminating untyped reduce-and-compare procedure with cumulative universes,
contravariant $\Pi$ input types, and $\eta$ for functions and pairs. Their syntactic rules differ from ours in two aspects. %
First, their subtyping rule for $\Pi$ types, if $\subtyp{A'}{A}$ and $\subtyp[\Gamma, x : A]{B}{B'}$, then $\subtyp{\usynPi{x}{A}{B}}{\usynPi{x}{A'}{B'}}$, compares $B$ and $B'$ in a context extended with the super-type, and hence is weaker than ours.
Secondly, their system includes
injectivities of $\Pi$ types as built-in rules as a technical necessity, which also eases the meta-theory. %
Meanwhile, injectivities of type constructors are provable via a semantic argument in our setting. %
Moreover, the absence of annotations on
their Curry-style $\lambda$s makes the development harder to scale to a full
type-checking procedure. %
Their verified artifact is a conversion checker. Our work instead provides a mechanized, sound, complete, and terminating type-checking procedure for contravariant universe subtyping over an
NbE backend.

\vspace{-3pt}
\paragraph{Universe structures}
Universes and their structures have long been a central topic in type theory. %
Introduced by \citet{DBLP:books/daglib/0000395} and further refined by \citet{palmgren_universes_1998} are two common universe structures: the Russell style, where types and terms have a unified syntactic category and share the same well-formedness judgment, and the Tarski style, where types and terms are distinguished on both syntactic and judgmental levels. %
The Russell style is the most common in practice and is what the syntactic theory of this paper presents. %
The Tarski style is semantically more fundamental~\cite{luo_universe_notes_2012}. %
In this style, elements of universes are \emph{codes} of types, and a decoder, often written as $\El$, is used to translate codes into actual types. %
Our semantic models are given in this style. %
\citet{coquand_presheaf} recently defines the Coquand style, which is similar to the Tarski style in its distinction between types and terms. %
It differs, however, in its additional encoding function, allowing to obtain codes back from types. %
The Coquand style in addition requires the encoder and the decoder to form an isomorphism. %
This style is convenient in certain constructions of modal type theory~\citep{DBLP:conf/lics/GratzerKNB20,DBLP:journals/lmcs/GratzerKNB21}.
Besides style, universes are also organized in different hierarchies:
\emph{cumulativity}, as in Luo's ECC~\cite{luo_ecc_1994} and \coq{} with its PCUIC
formalization~\cite{DBLP:journals/corr/abs-1710-03912}, or \emph{non-cumulativity}, as in  \agda~\cite{agda} and Lean~\citep{DBLP:conf/cade/Moura021}. %

\paragraph{Mechanizations of dependent type theory}
Our development extends \mcltt~\cite{mctt}, which mechanizes NbE \`a la
\citet{abel_normalization_2013} for MLTT with a cumulative hierarchy and
a \emph{covariant} $\Pi$ rule with universe subtyping in \coq, using an
impredicative encoding of the inductive-recursive PER model and
extracting a verified checker; type-directed conversion checking and its
completeness via PERs and logical relations follow the NbE line of
\citet{abel_normalization_2013}. The present work
keeps that architecture and upgrades the subtyping discipline to
contravariant $\Pi$ subtyping, which modifies its model moderately but
reshapes the algorithmic layer and its proof entirely.
\citet{abel_decidability_2017} mechanize in Agda a different,
reduction-based proof of decidability of conversion for MLTT, using typed
weak-head reduction and a Kripke logical relation rather than NbE.
\citet{DBLP:journals/pacmpl/PujetT22,DBLP:journals/pacmpl/PujetT23} extend
that line to observational equality and impredicative observational equality,
including decidability results and a refinement that avoids
induction-recursion.
\citet{wieczorek_coq_2018} mechanize NbE for an MLTT with one universe in
\coq. %
\citet{hu_jang_pientka_2023} mechanize NbE for a modal extension of
MLTT with full cumulativity in Agda; \citet{jiang_noncumu_2025} mechanize
NbE for a non-cumulative hierarchy. They also adopted an operational view on NbE to 
prove their additional universe levels annotations (to establish NbE soundness) do not change the execution trace and can be erased. 
Most recently, \citet{danielsson_firstclass_2026} internalize
universe levels as a first-class type and prove normalization and
decidability of type-checking for level-polymorphic Agda, a different
axis of universe flexibility from the cumulative slack we confine.

Beyond MLTT cores, verified kernels are now being built for the full
theories of working proof assistants: \citet{DBLP:conf/aplas/LiesnikovC24}
give a correct-by-construction, intrinsically typed checker for a core of
Agda in Agda, shifting the trusted base from kernel code to the stated
rules, and \citet{carneiro_lean4lean_2024} reimplement and verify Lean's
kernel in Lean, uncovering and fixing a soundness bug in the process. In
the reduction-based tradition, \citet{coqcoqcorrect} verify a type
checker for PCUIC against its formalized typing rules (assuming
normalization), \citet{lennonbertrand2025doescertifyconversion} isolates
what proof obligations are needed to certify conversion checkers, and
\citet{lennonbertrand_bidir} equips PCUIC with a complete bidirectional typing
discipline, all under the invariant-input-type cumulativity this paper relaxes.

%% file: sections/conclusion.tex
\vspace{-0.6em}
\section{Conclusion}\label{sec:conclusion}

We have presented \mclttsub, a complete, mechanized account of contravariant
universe subtyping for an NbE-based MLTT kernel. The syntactic theory, the NbE
algorithm, and the PER model adapt to contravariance with only moderate changes:
semantic subtyping is extended to also compare input types, and the soundness and
completeness of NbE are re-established in the new setting. The difficulty
concentrates in the subtyping algorithm, for which we give two verified
solutions. The renormalizing algorithm re-runs NbE on output types at every $\Pi$
and is proved correct,
at the price of extra computation. The renormalization-free algorithm instead
compares the already-computed normal forms structurally, recovering efficiency with a subtler completeness proof.

The renormalization-free proof reaches this result by an unconventional route.
The standard type-theoretic approach would build a semantic model and read
completeness off a realizability theorem. Contravariance, which inverts the
direction of subtyping under $\Pi$, makes this delicate: the property is
inherently directional, yet establishing it forces one to relate both directions
of the subtyping at once, and the symmetric reasoning the PER model rests on
breaks down, leaving each induction hypothesis valid in only one direction.
Rather than re-engineer the model, we step outside the well-typed domain and treat
NbE as a running algorithm: on raw terms, with no typing assumptions, we recover a
\emph{symmetric} invariant that forgets which side is the subtype, and a
bisimulation over the normalizer's shared execution traces closes the otherwise
directional inductions, from which completeness follows. Program equivalence is a
standard tool in operational semantics yet rarely brought to bear on dependent
type theory, and we believe this perspective reaches beyond the present result:
it realizes, outside any model, the universe-level irrelevance that
\citet{abel_normalization_2017} proposed to build into one; it should accommodate
further subtyping features, such as  $\Sigma$ or record  subtyping, which we leave to
future work; and, more broadly, it offers a fresh, algorithmic perspective for
type-theoretic meta-theory.

%% file: sections/appendix/rules.tex
\pagebreak
\section{The Full Declarative System}\label{app:declarative}

This appendix collects the rules and definitions that the body elides or
shows only in part. They are transcribed from the \coq{} mechanization and
use the notation of the main text. Subtyping, term subsumption, and context
subtyping are already given in full in \Cref{fig:subtyping}; the universe
PER model, semantic subtyping, and the completeness semantic judgments in
\Cref{fig:per-model,fig:per-subtyp} and \Cref{sec:nbe-per:completeness}. We
do not repeat them.

Throughout, $\qsub{\sigma} \coloneqq (\sigma \circ \synwk, x_0)$ denotes the
\emph{lifted} substitution that protects the topmost variable, and
$\qsub{\qsub{\sigma}}$ its double lift. Four subsumption-style rules
($\typing{M}{A'}$, $\substyping{\sigma}{\Delta'}$, and the two equivalence
analogues) each carry an extra well-formedness premise on the \emph{target}
that is redundant given presupposition (\Cref{thm:syntax:presup}); we mark
such premises \phlmath{\text{light gray}}, exactly as in \Cref{fig:subtyping}.

\paragraph{Contexts and variable lookup}
Context well-formedness $\wfctx{\Gamma}$ and the lookup judgment
$x_d : A \in \Gamma$:
\begin{rulepar}
  \inferrule{~}{\wfctx{\cdot}}

  \inferrule{\wfctx{\Gamma} \\ \typing{A}{\Ty i}}{\wfctx{\Gamma, x : A}}

  \inferrule{~}{x_0 : A[\synwk] \in \Gamma, x : A}

  \inferrule{x_d : A \in \Gamma}{x_{1+d} : A[\synwk] \in \Gamma, y : B}
\end{rulepar}
Context equivalence $\wfctxeq{\Gamma}{\Delta}$ is the evident pointwise
congruence:
\begin{rulepar}
  \inferrule{~}{\wfctxeq{\cdot}{\cdot}}

  \inferrule
  {\wfctxeq{\Gamma}{\Delta} \\\\
   \typing{A}{\Ty i} \\ \typing{A'}{\Ty i} \\\\
   \typing[\Delta]{A}{\Ty i} \\ \typing[\Delta]{A'}{\Ty i} \\\\
   \tyequiv{A}{A'}{\Ty i} \\ \tyequiv[\Delta]{A}{A'}{\Ty i}}
  {\wfctxeq{\Gamma, x : A}{\Delta, x : A'}}
\end{rulepar}

\paragraph{Typing of terms}
The full set of rules for $\typing{M}{A}$ (the subsumption rule is in
\Cref{fig:subtyping}):
\begin{rulepar}
  \inferrule{\wfctx{\Gamma}}{\typing{\Ty i}{\Ty{(1+i)}}}

  \inferrule{\wfctx{\Gamma}}{\typing{\Nat}{\Ty 0}}

  \inferrule{\wfctx{\Gamma}}{\typing{\synze}{\Nat}}

  \inferrule{\typing{M}{\Nat}}{\typing{\synsu M}{\Nat}}

  \inferrule
  {\wfctx{\Gamma} \\ x_d : A \in \Gamma}
  {\typing{x_d}{A}}

  \inferrule
  {\typing{A}{\Ty i} \\ \typing[\Gamma, x : A]{B}{\Ty i}}
  {\typing{\usynPi{x}{A}{B}}{\Ty i}}

  \inferrule
  {\typing{A}{\Ty i} \\ \typing[\Gamma, x : A]{M}{B}}
  {\typing{\usynlam{x}{A}{M}}{\usynPi{x}{A}{B}}}

  \inferrule
  {\typing{A}{\Ty i} \\ \typing[\Gamma, x : A]{B}{\Ty i} \\\\
   \typing{M}{\usynPi{x}{A}{B}} \\ \typing{N}{A}}
  {\typing{M~N}{B[\synid, N/x]}}

  \inferrule
  {\typing[\Gamma, x : \Nat]{A}{\Ty i} \\
   \typing{M_{\mathsf z}}{A[\synid, \synze/x]} \\\\
   \typing[\Gamma, y : \Nat, r : A]{M_{\mathsf s}}{A[\synwk\circ\synwk, \synsu y/x]} \\
   \typing{N}{\Nat}}
  {\typing{\recn{x.A}{M_{\mathsf z}}{y,r.M_{\mathsf s}}{N}}{A[\synid, N/x]}}

  \inferrule
  {\substyping{\sigma}{\Delta} \\ \typing[\Delta]{M}{A}}
  {\typing{M[\sigma]}{A[\sigma]}}
\end{rulepar}

\paragraph{Typing of substitutions}
For $\substyping{\sigma}{\Delta}$ (the subsumption rule
$\substyping{\sigma}{\Delta'}$ is in \Cref{fig:subtyping}):
\begin{rulepar}
  \inferrule{\wfctx{\Gamma}}{\substyping{\synid}{\Gamma}}

  \inferrule{\wfctx{\Gamma, x : A}}{\substyping[\Gamma, x : A]{\synwk}{\Gamma}}

  \inferrule
  {\substyping[\Gamma_1]{\delta}{\Gamma_2} \\ \substyping[\Gamma_2]{\sigma}{\Gamma_3}}
  {\substyping[\Gamma_1]{\sigma \circ \delta}{\Gamma_3}}

  \inferrule
  {\substyping{\sigma}{\Delta} \\ \typing[\Delta]{A}{\Ty i} \\ \typing{M}{A[\sigma]}}
  {\substyping{(\sigma, M/x)}{\Delta, x : A}}
\end{rulepar}

\paragraph{Term equivalence: substitution on the formers}
Each former commutes with substitution; lifts $\qsub{\cdot}$ appear under
binders.
\begin{rulepar}
  \inferrule{\substyping{\sigma}{\Delta}}{\tyequiv{\Ty i[\sigma]}{\Ty i}{\Ty{(1+i)}}}

  \inferrule{\substyping{\sigma}{\Delta}}{\tyequiv{\Nat[\sigma]}{\Nat}{\Ty 0}}

  \inferrule{\substyping{\sigma}{\Delta}}{\tyequiv{\synze[\sigma]}{\synze}{\Nat}}

  \inferrule{\substyping{\sigma}{\Delta} \\ \typing[\Delta]{M}{\Nat}}
  {\tyequiv{(\synsu M)[\sigma]}{\synsu{(M[\sigma])}}{\Nat}}

  \inferrule
  {\substyping{\sigma}{\Delta} \\ \typing[\Delta]{A}{\Ty i} \\ \typing[\Delta, x : A]{B}{\Ty i}}
  {\tyequiv{(\usynPi{x}{A}{B})[\sigma]}{\usynPi{x}{A[\sigma]}{B[\qsub{\sigma}]}}{\Ty i}}

  \inferrule
  {\substyping{\sigma}{\Delta} \\ \typing[\Delta]{A}{\Ty i} \\ \typing[\Delta, x : A]{M}{B}}
  {\tyequiv{(\usynlam{x}{A}{M})[\sigma]}{\usynlam{x}{A[\sigma]}{M[\qsub{\sigma}]}}{(\usynPi{x}{A}{B})[\sigma]}}

  \inferrule
  {\substyping{\sigma}{\Delta} \\ \typing[\Delta]{A}{\Ty i} \\ \typing[\Delta, x : A]{B}{\Ty i} \\\\
   \typing[\Delta]{M}{\usynPi{x}{A}{B}} \\ \typing[\Delta]{N}{A}}
  {\tyequiv{(M~N)[\sigma]}{(M[\sigma])~(N[\sigma])}{B[\sigma, N[\sigma]/x]}}

  \inferrule
  {\substyping{\sigma}{\Delta} \\
   \typing[\Delta, x : \Nat]{A}{\Ty i} \\
   \typing[\Delta]{M_{\mathsf z}}{A[\synid, \synze/x]} \\\\
   \typing[\Delta, y : \Nat, r : A]{M_{\mathsf s}}{A[\synwk\circ\synwk, \synsu y/x]} \\
   \typing[\Delta]{N}{\Nat}}
  {\tyequiv{(\recn{x.A}{M_{\mathsf z}}{y,r.M_{\mathsf s}}{N})[\sigma]}
           {\recn{x.A[\qsub{\sigma}]}{M_{\mathsf z}[\sigma]}{y,r.M_{\mathsf s}[\qsub{\qsub{\sigma}}]}{N[\sigma]}}
           {A[\sigma, N[\sigma]/x]}}
\end{rulepar}

\paragraph{Term equivalence: congruence}
\begin{rulepar}
  \inferrule{\tyequiv{M}{M'}{\Nat}}{\tyequiv{\synsu M}{\synsu{M'}}{\Nat}}

  \inferrule
  {\typing{A}{\Ty i} \\ \tyequiv{A}{A'}{\Ty i} \\ \tyequiv[\Gamma, x : A]{B}{B'}{\Ty i}}
  {\tyequiv{\usynPi{x}{A}{B}}{\usynPi{x}{A'}{B'}}{\Ty i}}

  \inferrule
  {\typing{A}{\Ty i} \\ \tyequiv{A}{A'}{\Ty i} \\ \tyequiv[\Gamma, x : A]{M}{M'}{B}}
  {\tyequiv{\usynlam{x}{A}{M}}{\usynlam{x}{A'}{M'}}{\usynPi{x}{A}{B}}}

  \inferrule
  {\typing{A}{\Ty i} \\ \typing[\Gamma, x : A]{B}{\Ty i} \\\\
   \tyequiv{M}{M'}{\usynPi{x}{A}{B}} \\ \tyequiv{N}{N'}{A}}
  {\tyequiv{M~N}{M'~N'}{B[\synid, N/x]}}

  \inferrule
  {\typing[\Gamma, x : \Nat]{A}{\Ty i} \\ \tyequiv[\Gamma, x : \Nat]{A}{A'}{\Ty i} \\
   \tyequiv{M_{\mathsf z}}{M_{\mathsf z}'}{A[\synid, \synze/x]} \\\\
   \tyequiv[\Gamma, y : \Nat, r : A]{M_{\mathsf s}}{M_{\mathsf s}'}{A[\synwk\circ\synwk, \synsu y/x]} \\
   \tyequiv{N}{N'}{\Nat}}
  {\tyequiv{\recn{x.A}{M_{\mathsf z}}{y,r.M_{\mathsf s}}{N}}
           {\recn{x.A'}{M_{\mathsf z}'}{y,r.M_{\mathsf s}'}{N'}}{A[\synid, N/x]}}
\end{rulepar}

\paragraph{Term equivalence: $\beta$ and $\eta$}
\begin{rulepar}
  \inferrule
  {\typing[\Gamma, x : \Nat]{A}{\Ty i} \\
   \typing{M_{\mathsf z}}{A[\synid, \synze/x]} \\
   \typing[\Gamma, y : \Nat, r : A]{M_{\mathsf s}}{A[\synwk\circ\synwk, \synsu y/x]}}
  {\tyequiv{\recn{x.A}{M_{\mathsf z}}{y,r.M_{\mathsf s}}{\synze}}{M_{\mathsf z}}{A[\synid, \synze/x]}}

  \inferrule
  {\typing[\Gamma, x : \Nat]{A}{\Ty i} \\
   \typing{M_{\mathsf z}}{A[\synid, \synze/x]} \\\\
   \typing[\Gamma, y : \Nat, r : A]{M_{\mathsf s}}{A[\synwk\circ\synwk, \synsu y/x]} \\
   \typing{M}{\Nat}}
  {\tyequiv{\recn{x.A}{M_{\mathsf z}}{y,r.M_{\mathsf s}}{\synsu M}}
           {M_{\mathsf s}[\synid, M/y, (\recn{x.A}{M_{\mathsf z}}{y,r.M_{\mathsf s}}{M})/r]}
           {A[\synid, \synsu M/x]}}

  \inferrule
  {\typing{A}{\Ty i} \\ \typing[\Gamma, x : A]{B}{\Ty i} \\ \typing[\Gamma, x : A]{M}{B} \\ \typing{N}{A}}
  {\tyequiv{(\usynlam{x}{A}{M})~N}{M[\synid, N/x]}{B[\synid, N/x]}}

  \inferrule
  {\typing{A}{\Ty i} \\ \typing[\Gamma, x : A]{B}{\Ty i} \\ \typing{M}{\usynPi{x}{A}{B}}}
  {\tyequiv{M}{\usynlam{x}{A}{(M[\synwk])~x_0}}{\usynPi{x}{A}{B}}}
\end{rulepar}

\paragraph{Term equivalence: variables, substitution laws, conversion}
\begin{rulepar}
  \inferrule{\wfctx{\Gamma} \\ x_d : A \in \Gamma}{\tyequiv{x_d}{x_d}{A}}

  \inferrule
  {\substyping{\sigma}{\Delta} \\ \typing[\Delta]{A}{\Ty i} \\ \typing{M}{A[\sigma]}}
  {\tyequiv{x_0[(\sigma, M/x)]}{M}{A[\sigma]}}

  \inferrule
  {\substyping{\sigma}{\Delta} \\ \typing[\Delta]{A}{\Ty i} \\ \typing{M}{A[\sigma]} \\ x_d : B \in \Delta}
  {\tyequiv{x_{1+d}[(\sigma, M/x)]}{x_d[\sigma]}{B[\sigma]}}

  \inferrule
  {\wfctx{\Gamma, y : B} \\ x_d : A \in \Gamma}
  {\tyequiv[\Gamma, y : B]{x_d[\synwk]}{x_{1+d}}{A[\synwk]}}

  \inferrule
  {\tyequiv[\Delta]{M}{M'}{A} \\ \subeq{\sigma}{\sigma'}{\Delta}}
  {\tyequiv{M[\sigma]}{M'[\sigma']}{A[\sigma]}}

  \inferrule{\typing{M}{A}}{\tyequiv{M[\synid]}{M}{A}}

  \inferrule
  {\substyping[\Gamma]{\tau}{\Gamma'} \\ \substyping[\Gamma']{\sigma}{\Gamma''} \\ \typing[\Gamma'']{M}{A}}
  {\tyequiv{M[\sigma \circ \tau]}{M[\sigma][\tau]}{A[\sigma \circ \tau]}}

  \inferrule
  {\tyequiv{M}{M'}{A} \\ \phlmath{\typing{A'}{\Ty i}} \\ \subtyp{A}{A'}}
  {\tyequiv{M}{M'}{A'}}

  \inferrule{\tyequiv{M}{M'}{A}}{\tyequiv{M'}{M}{A}}

  \inferrule{\tyequiv{M}{M'}{A} \\ \tyequiv{M'}{M''}{A}}{\tyequiv{M}{M''}{A}}
\end{rulepar}

\paragraph{Substitution equivalence}
The category laws, the $\eta$-law for substitutions, congruence, and the
subsumption analogue:
\begin{rulepar}
  \inferrule{\wfctx{\Gamma}}{\subeq{\synid}{\synid}{\Gamma}}

  \inferrule{\wfctx{\Gamma, x : A}}{\subeq[\Gamma, x : A]{\synwk}{\synwk}{\Gamma}}

  \inferrule
  {\subeq[\Gamma]{\tau}{\tau'}{\Gamma'} \\ \subeq[\Gamma']{\sigma}{\sigma'}{\Gamma''}}
  {\subeq{\sigma \circ \tau}{\sigma' \circ \tau'}{\Gamma''}}

  \inferrule
  {\subeq{\sigma}{\sigma'}{\Delta} \\ \typing[\Delta]{A}{\Ty i} \\ \tyequiv{M}{M'}{A[\sigma]}}
  {\subeq{(\sigma, M/x)}{(\sigma', M'/x)}{\Delta, x : A}}

  \inferrule{\substyping{\sigma}{\Delta}}{\subeq{\synid \circ \sigma}{\sigma}{\Delta}}

  \inferrule{\substyping{\sigma}{\Delta}}{\subeq{\sigma \circ \synid}{\sigma}{\Delta}}

  \inferrule
  {\substyping[\Gamma']{\sigma}{\Gamma} \\ \substyping[\Gamma'']{\sigma'}{\Gamma'} \\ \substyping[\Gamma''']{\sigma''}{\Gamma''}}
  {\subeq[\Gamma''']{(\sigma \circ \sigma') \circ \sigma''}{\sigma \circ (\sigma' \circ \sigma'')}{\Gamma}}

  \inferrule
  {\substyping[\Gamma']{\sigma}{\Gamma''} \\ \typing[\Gamma'']{A}{\Ty i} \\ \typing[\Gamma']{M}{A[\sigma]} \\ \substyping[\Gamma]{\tau}{\Gamma'}}
  {\subeq{(\sigma, M/x) \circ \tau}{(\sigma \circ \tau, M[\tau]/x)}{\Gamma'', x : A}}

  \inferrule
  {\substyping[\Gamma']{\sigma}{\Gamma} \\ \typing{A}{\Ty i} \\ \typing[\Gamma']{M}{A[\sigma]}}
  {\subeq[\Gamma']{\synwk \circ (\sigma, M/x)}{\sigma}{\Gamma}}

  \inferrule
  {\substyping[\Gamma']{\sigma}{\Gamma, x : A}}
  {\subeq[\Gamma']{\sigma}{(\synwk \circ \sigma, x_0[\sigma]/x)}{\Gamma, x : A}}

  \inferrule{\subeq{\sigma}{\sigma'}{\Delta}}{\subeq{\sigma'}{\sigma}{\Delta}}

  \inferrule{\subeq{\sigma}{\sigma'}{\Delta} \\ \subeq{\sigma'}{\sigma''}{\Delta}}{\subeq{\sigma}{\sigma''}{\Delta}}

  \inferrule
  {\subeq{\sigma}{\sigma'}{\Delta} \\ \phlmath{\wfctx{\Delta'}} \\ \wfctxsub{\Delta}{\Delta'}}
  {\subeq{\sigma}{\sigma'}{\Delta'}}
\end{rulepar}

\section{Normalization by Evaluation: Full Operational Semantics}\label{app:nbe}

\Cref{fig:eval} in \Cref{sec:nbe-per:nbe}
show only a selection. The complete relations follow; the domain grammar and
notation ($\dneut{a}{v}$, $\dreif{a}{m}$, $\denvext{\rho}{n}$,
$\denvdrop{\rho}$, $\dvar z$, closures) are as in \Cref{sec:nbe-per:nbe}.

\paragraph{Evaluation, domain application, and recursion}
$\evald{M}{\rho}{m}$ for expressions, $\dapp{m}{n} \shortsearrow r$ for
domain application, $\drecd{m}{x.A}{M_{\mathsf z}}{y,r.M_{\mathsf s}}{\rho}{r'}$
for domain recursion, and $\evalsd{\sigma}{\rho}{\rho'}$ for substitutions:
\begin{rulepar}
  \inferrule{~}{\evald{\Ty i}{\rho}{\dU i}}

  \inferrule{~}{\evald{x_d}{\rho}{\rho(d)}}

  \inferrule{~}{\evald{\Nat}{\rho}{\dN}}

  \inferrule{~}{\evald{\synze}{\rho}{\dze}}

  \inferrule{\evald{M}{\rho}{m}}{\evald{\synsu M}{\rho}{\dsu m}}

  \inferrule{\evald{A}{\rho}{a}}{\evald{\usynPi{x}{A}{B}}{\rho}{\dpi{a}{B}{\rho}}}

  \inferrule{~}{\evald{\usynlam{x}{A}{M}}{\rho}{\dlam{M}{\rho}}}

  \inferrule
  {\evald{M}{\rho}{m} \\ \evald{N}{\rho}{n} \\ \dapp{m}{n} \shortsearrow r}
  {\evald{M~N}{\rho}{r}}

  \inferrule
  {\evald{N}{\rho}{m} \\ \drecd{m}{x.A}{M_{\mathsf z}}{y,r.M_{\mathsf s}}{\rho}{r'}}
  {\evald{\recn{x.A}{M_{\mathsf z}}{y,r.M_{\mathsf s}}{N}}{\rho}{r'}}

  \inferrule
  {\evalsd{\sigma}{\rho}{\rho'} \\ \evald{M}{\rho'}{m}}
  {\evald{M[\sigma]}{\rho}{m}}
\end{rulepar}
\begin{rulepar}
  \inferrule
  {\evald{M}{\denvext{\rho}{n}}{m}}
  {\dapp{\dlam{M}{\rho}}{n} \shortsearrow m}

  \inferrule
  {\evald{B}{\denvext{\rho}{n}}{b}}
  {\dapp{\dneut{\dpi{a}{B}{\rho}}{v}}{n} \shortsearrow \dneut{b}{v~\dreif{a}{n}}}
\end{rulepar}
\begin{rulepar}
  \inferrule
  {\evald{M_{\mathsf z}}{\rho}{m_{\mathsf z}}}
  {\drecd{\dze}{x.A}{M_{\mathsf z}}{y,r.M_{\mathsf s}}{\rho}{m_{\mathsf z}}}

  \inferrule
  {\drecd{b}{x.A}{M_{\mathsf z}}{y,r.M_{\mathsf s}}{\rho}{r'} \\
   \evald{M_{\mathsf s}}{\denvext{\denvext{\rho}{b}}{r'}}{m_{\mathsf s}}}
  {\drecd{\dsu b}{x.A}{M_{\mathsf z}}{y,r.M_{\mathsf s}}{\rho}{m_{\mathsf s}}}

  \inferrule
  {\evald{M_{\mathsf z}}{\rho}{m_{\mathsf z}} \\ \evald{A}{\denvext{\rho}{\dneut{b}{v}}}{a}}
  {\drecd{\dneut{b}{v}}{x.A}{M_{\mathsf z}}{y,r.M_{\mathsf s}}{\rho}
         {\dneut{a}{\dnatrec{A}{m_{\mathsf z}}{M_{\mathsf s}}{\rho}{v}}}}
\end{rulepar}
\begin{rulepar}
  \inferrule{~}{\evalsd{\synid}{\rho}{\rho}}

  \inferrule{~}{\evalsd{\synwk}{\rho}{\denvdrop{\rho}}}

  \inferrule
  {\evalsd{\sigma}{\rho}{\rho'} \\ \evald{M}{\rho}{m}}
  {\evalsd{(\sigma, M/x)}{\rho}{\denvext{\rho'}{m}}}

  \inferrule
  {\evalsd{\tau}{\rho}{\rho_\tau} \\ \evalsd{\sigma}{\rho_\tau}{\rho'}}
  {\evalsd{\sigma \circ \tau}{\rho}{\rho'}}
\end{rulepar}

\paragraph{Readback}
The three mutually defined relations $\rnf{z}{w}{W}$ (normal values),
$\rne{z}{v}{U}$ (neutral values), and $\rtyd{z}{a}{V}$ (types); $z$ counts the
enclosing binders. We write $\dvar{z}$ for the fresh variable at level $z$
and $\dneut{a}{\dvar z}$ for its reflection at type tag $a$. The
function/$\Pi$ cases drive $\eta$-expansion.
\begin{rulepar}
  \inferrule{\rtyd{z}{a}{A}}{\rnf{z}{\dreif{\dU i}{a}}{A}}

  \inferrule{~}{\rnf{z}{\dreif{\dN}{\dze}}{\synze}}

  \inferrule{\rnf{z}{\dreif{\dN}{m}}{W}}{\rnf{z}{\dreif{\dN}{\dsu m}}{\synsu W}}

  \inferrule{\rne{z}{v}{U}}{\rnf{z}{\dreif{\dN}{\dneut{a}{v}}}{U}}

  \inferrule
  {\rtyd{z}{a}{W_1} \\
   \dapp{m}{\dneut{a}{\dvar z}} \shortsearrow m' \\
   \evald{B}{\denvext{\rho}{\dneut{a}{\dvar z}}}{b} \\
   \rnf{(1+z)}{\dreif{b}{m'}}{W_2}}
  {\rnf{z}{\dreif{\dpi{a}{B}{\rho}}{m}}{\usynlam{x}{W_1}{W_2}}}

  \inferrule{\rne{z}{v}{U}}{\rnf{z}{\dreif{\dneut{a}{b}}{\dneut{c}{v}}}{U}}
\end{rulepar}
\begin{rulepar}
  \inferrule{~}{\rne{z}{\dvar{x}}{x_{(z - x - 1)}}}

  \inferrule{\rne{z}{v}{U} \\ \rnf{z}{w}{W}}{\rne{z}{v~w}{U~W}}

  \inferrule
  {\evald{A}{\denvext{\rho}{\dneut{\dN}{\dvar z}}}{a} \\
   \rtyd{(1+z)}{a}{A'} \\
   \evald{A}{\denvext{\rho}{\dze}}{a_{\mathsf z}} \\
   \rnf{z}{\dreif{a_{\mathsf z}}{m_{\mathsf z}}}{W_{\mathsf z}} \\\\
   \evald{A}{\denvext{\rho}{\dsu{\dneut{\dN}{\dvar z}}}}{a_{\mathsf s}} \\
   \evald{M_{\mathsf s}}{\denvext{\denvext{\rho}{\dneut{\dN}{\dvar z}}}{\dneut{a}{\dvar{(1+z)}}}}{m_{\mathsf s}} \\
   \rnf{(2+z)}{\dreif{a_{\mathsf s}}{m_{\mathsf s}}}{W_{\mathsf s}} \\\\
   \rne{z}{v}{U}}
  {\rne{z}{\dnatrec{A}{m_{\mathsf z}}{M_{\mathsf s}}{\rho}{v}}
       {\recn{x.A'}{W_{\mathsf z}}{y,r.W_{\mathsf s}}{U}}}
\end{rulepar}
\begin{rulepar}
  \inferrule{~}{\rtyd{z}{\dU i}{\Ty i}}

  \inferrule{~}{\rtyd{z}{\dN}{\Nat}}

  \inferrule
  {\rtyd{z}{a}{W_1} \\
   \evald{B}{\denvext{\rho}{\dneut{a}{\dvar z}}}{b} \\
   \rtyd{(1+z)}{b}{W_2}}
  {\rtyd{z}{\dpi{a}{B}{\rho}}{\usynPi{x}{W_1}{W_2}}}

  \inferrule{\rne{z}{v}{U}}{\rtyd{z}{\dneut{a}{v}}{U}}
\end{rulepar}

\section{Context-Level Semantic Models}\label{app:semantic}

The universe PER model $\perm{i}{a}{a'}{R}$ and semantic subtyping
$\semsubtyp{a}{a'}{i}$ are in \Cref{fig:per-model,fig:per-subtyp}; the
completeness semantic judgments in \Cref{sec:nbe-per:completeness}. We give
the auxiliary base PERs and the context-level relations, whose extension
rules the body only sketched, together with the standard $\Nat$ rules omitted
from \Cref{fig:per-model,fig:per-subtyp}:
\begin{rulepar}
  \inferrule{R \releq \Natd}{\perm{i}{\dN}{\dN}{R}}

  \inferrule{~}{\semsubtyp{\dN}{\dN}{i}}
\end{rulepar}

\paragraph{Base and inductive PERs}
$\inper{w}{w'}{\Nfd}$, $\inper{v}{v'}{\Ned}$, $\inper{a}{a'}{\Tyd}$ are in
\Cref{sec:nbe-per:per} (relate iff they read back to a common form at every
depth). On top of them, $\Natd$ and $\Neud$; both ignore type tags:
\begin{rulepar}
  \inferrule{~}{\inper{\dze}{\dze}{\Natd}}

  \inferrule{\inper{m}{n}{\Natd}}{\inper{\dsu m}{\dsu n}{\Natd}}

  \inferrule{\inper{v}{v'}{\Ned}}{\inper{\dneut{a}{v}}{\dneut{a'}{v'}}{\Natd}}

  \inferrule{\inper{v}{v'}{\Ned}}{\inper{\dneut{a}{v}}{\dneut{a'}{v'}}{\Neud}}
\end{rulepar}

\paragraph{Context PER and environments}
$\perctxd{\Gamma}{\Gamma'}{R}$ relates contexts and returns a PER $R$ on
environments. The empty context relates all environments; extension
quantifies the head type over environments related at the tail and pairs
environments by their head and (dropped) tail.
\begin{rulepar}
  \inferrule
  {R \releq \{(\rho, \rho') \mid \top\}}
  {\perctxd{\cdot}{\cdot}{R}}

  \inferrule
  {\perctxd{\Gamma}{\Gamma'}{R_T} \\
   \phlmath{R_T \text{ is a PER}} \\\\
   \forall \D \uname \inper{\rho}{\rho'}{R_T}.\
     \perm{i}{\evalf{A}{\rho}}{\evalf{A'}{\rho'}}{R_H(\D)} \\\\
   R \releq \{(\rho, \rho') \mid \exists\, \D \uname \inper{\denvdrop{\rho}}{\denvdrop{\rho'}}{R_T}.\
     \inper{\rho(0)}{\rho'(0)}{R_H(\D)}\}}
  {\perctxd{\Gamma, x : A}{\Gamma', x : A'}{R}}
\end{rulepar}
Here $R_H$ is a family indexed by a proof $\D$ of tail-relatedness, exactly as
the output relation $R_O$ in the $\Pi$ rule of \Cref{fig:per-model}.

\paragraph{Semantic context subtyping}
$\perctxsub{\Gamma}{\Gamma'}$ generalizes $\semsubtyp{}{}{}$ pointwise; both
extended contexts must be self-related, and the head types are compared
covariantly over environments related at the smaller tail:
\begin{rulepar}
  \inferrule{~}{\perctxsub{\cdot}{\cdot}}

  \inferrule
  {\perctxsub{\Gamma}{\Gamma'} \\
   \perctxd{\Gamma}{\Gamma}{R_T} \\\\
   \forall \inper{\rho}{\rho'}{R_T}.\
     \evald{A}{\rho}{a} \to \evald{A'}{\rho'}{a'} \to \semsubtyp{a}{a'}{i} \\\\
   \perctxd{\Gamma, x : A}{\Gamma, x : A}{R_1} \\\\
   \perctxd{\Gamma', x : A'}{\Gamma', x : A'}{R_2}}
  {\perctxsub{\Gamma, x : A}{\Gamma', x : A'}}
\end{rulepar}

\section{The Gluing (Logical-Relations) Model in Full}\label{app:gluing}

\Cref{sec:nbe-per:soundness} describes the Kripke gluing model only in prose.
We give its definitions. The model is inherited from \mcltt{} unchanged.
Contravariance surfaces not in these
definitions but in the metatheory built on them. We
write $\glumodels{i}{P}{El}{a}$ for ``the domain type $a$ glues, at level
$i$, to type-predicate $P$ and term-predicate $El$''; $\gluty{A}[i]{P}$ for
the glued type and $\glutm{A}{M}{m}[i]{El}$ for the glued term; and
$\sndtyping{M}{A}$, $\sndctx{\Gamma}$ for the soundness judgments. Predicate
equivalence ($\releq$, pointwise iff) ties each constructor's output
predicates to the canonical ones.

\paragraph{Weakening}
$\gluwk{\sigma}{\Delta}$ holds when $\sigma$ is definitionally a composite of
$\synwk$'s; its second rule carries the context-subtyping premise that
removes the need for a separate gluing judgment for subtyping:
\begin{rulepar}
  \inferrule
  {\subeq{\sigma}{\synid}{\Delta}}
  {\gluwk{\sigma}{\Delta}}

  \inferrule
  {\gluwk{\tau}{\Delta', x : A} \\
   \wfctxsub{\Delta'}{\Delta} \\
   \subeq{\sigma}{\synwk \circ \tau}{\Delta}}
  {\gluwk{\sigma}{\Delta}}
\end{rulepar}

\paragraph{The gluing relation}
$\glumodels{i}{P}{El}{a}$ is defined by the same case analysis on $a$ as the
PER model. (As there, universe levels $j < i$ recur through a well-founded
knot: at $\dU j$ the predicates are those of a glued type/term one level
down.)
\begin{rulepar}
  \inferrule
  {P \releq \{(\Gamma, A) \mid \tyequiv{A}{\Ty j}{\Ty i}\} \\
   j < i \\\\
   El \releq \{(\Gamma, A, M, m) \mid \typing{M}{A},\ \tyequiv{A}{\Ty j}{\Ty i},\ \gluty{M}[j]{P_m}\}}
  {\glumodels{i}{P}{El}{\dU j}}

  \inferrule
  {P \releq \{(\Gamma, A) \mid \tyequiv{A}{\Nat}{\Ty i}\} \\\\
   El \releq \{(\Gamma, A, M, m) \mid \gluty{A}[i]{P},\ \glunat[\Gamma]{M}{m}\}}
  {\glumodels{i}{P}{El}{\dN}}

  \inferrule
  {\inper{b}{b}{\Ned} \\
   P \releq P_{\dneut{a}{b}} \\
   El \releq El_{\dneut{a}{b}}}
  {\glumodels{i}{P}{El}{\dneut{a}{b}}}

  \inferrule
  {\glumodels{i}{P_I}{El_I}{a} \\
   \perm{i}{a}{a}{R_I} \\
   \perm{i}{\dpi{a}{B}{\rho}}{\dpi{a}{B}{\rho}}{R} \\\\
   \forall \inper{c}{c}{R_I}.\ \glumodels{i}{P_O(c)}{El_O(c)}{\evalf{B}{\denvext{\rho}{c}}} \\\\
   P \releq P_\Pi \\ El \releq El_\Pi}
  {\glumodels{i}{P}{El}{\dpi{a}{B}{\rho}}}
\end{rulepar}
In the universe rule, $P_m$ is the glued-type predicate determined by the
value $m$ at level $j$ (the knot tied by well-founded recursion on $i$). The
$\Pi$ predicates $P_\Pi$ and $El_\Pi$ are the Kripke clauses below; note that
$El_I$ occurs in a \emph{negative} position in both (a glued argument is
\emph{assumed}), which is the source of the cumulativity ``lowering'' lemma
(\Cref{sec:nbe-per:soundness}).

\paragraph{Glued types and terms by shape}
Unfolding $\gluty{A}[i]{P}$ ($\coloneqq P\,(\Gamma, A)$) and
$\glutm{A}{M}{m}[i]{El}$ ($\coloneqq El\,(\Gamma, A, M, m)$):
\begin{itemize}
\item \emph{Neutral} $a = \dneut{a_0}{b}$: $\gluty{A}[i]{P}$ holds iff
  $\typing{A}{\Ty i}$ and every readback of $b$ is, under every weakening
  $\gluwk{\sigma}{\Delta}$, definitionally equal to $A[\sigma]$; and
  $\glutm{A}{M}{m}[i]{El}$ additionally requires $\typing{M}{A}$,
  $m \approx m \in \Bot$, and that every readback of $m$ is, under every
  weakening, $\approx M[\sigma]$ at $A[\sigma]$.
\item \emph{$\Pi$} $a = \dpi{a_0}{B}{\rho}$ (with $A \approx \usynPi{x}{I}{O}$):
  $\gluty{A}[i]{P_\Pi}$ requires, for every $\gluwk{\sigma}{\Delta}$, that
  $\gluty[\Delta]{I[\sigma]}[i]{P_I}$, and for every glued argument
  $\glutm[\Delta]{I[\sigma]}{N}{n}[i]{El_I}$ that
  $\gluty[\Delta]{O[\qsub{\sigma}, N/x]}[i]{P_O(n)}$. The glued term
  $\glutm{A}{M}{m}[i]{El_\Pi}$ additionally requires that for every such glued
  argument there is $\dapp{m}{n} \shortsearrow r$ with
  $\glutm[\Delta]{O[\qsub{\sigma}, N/x]}{(M[\sigma])~N}{r}[i]{El_O(n)}$.
\end{itemize}
$\glunat[\Gamma]{M}{m}$ glues a $\Nat$-term to its value: $m = \dze$ with
$\tyequiv{M}{\synze}{\Nat}$; $m = \dsu{m'}$ with
$\tyequiv{M}{\synsu{M'}}{\Nat}$ and $\glunat[\Gamma]{M'}{m'}$; or $m$ neutral
with the same readback condition as the neutral term predicate.

\paragraph{Glued contexts and soundness judgments}
Contexts glue to predicates $\mathrm{Sb}$ on substitution--environment pairs,
written $\sndsub{\sigma}{\rho}{\mathrm{Sb}}$. The empty context glues
$\sigma$ to $\rho$ iff $\substyping{\sigma}{\cdot}$; extension pairs a glued
head variable with a glued tail:
\begin{rulepar}
  \inferrule
  {\mathrm{Sb} \releq \{(\Delta, \sigma, \rho) \mid \substyping[\Delta]{\sigma}{\cdot}\}}
  {\gluctx{\cdot}{\mathrm{Sb}}}

  \inferrule
  {\gluctx{\Gamma}{\mathrm{Sb}_T} \\
   \typing{A}{\Ty i} \\
   \forall \sndsub[\Delta]{\sigma}{\rho}{\mathrm{Sb}_T}.\
     \evald{A}{\rho}{a},\ \glumodels{i}{P}{El}{a},\ \gluty[\Delta]{A[\synwk \circ \sigma]}[i]{P} \\\\
   \mathrm{Sb} \releq \mathrm{Sb}_{\Gamma, A}}
  {\gluctx{\Gamma, x : A}{\mathrm{Sb}}}
\end{rulepar}
where, in the extension predicate $\mathrm{Sb}_{\Gamma, A}$, a pair
$(\Delta, \sigma, \rho)$ is related iff $\substyping[\Delta]{\sigma}{\Gamma, x : A}$,
the head $\glutm[\Delta]{A[\synwk \circ \sigma]}{x_0[\sigma]}{\rho(0)}[i]{El}$,
and the tail $\sndsub[\Delta]{\synwk \circ \sigma}{\denvdrop{\rho}}{\mathrm{Sb}_T}$.
The soundness judgments quantify over glued pairs:
\[
  \begin{array}{r@{\ }l}
    \sndctx{\Gamma} \coloneqq{} & \exists\, \mathrm{Sb}.\ \gluctx{\Gamma}{\mathrm{Sb}}; \\[2pt]
    \sndtyping{M}{A} \coloneqq{} & \exists\, \mathrm{Sb}, i.\ \gluctx{\Gamma}{\mathrm{Sb}}
      \text{ and } \forall \sndsub[\Delta]{\sigma}{\rho}{\mathrm{Sb}}.\\
    & \quad \evald{A}{\rho}{a},\ \evald{M}{\rho}{m},\ \glumodels{i}{P}{El}{a},\
      \glutm[\Delta]{A[\sigma]}{M[\sigma]}{m}[i]{El}.
  \end{array}
\]
The substitution soundness judgment $\sndsubtyping{\tau}{\Gamma'}$ is
analogous, gluing $\tau \circ \sigma$ to $\evalf{\tau}{\rho}$ at
$\mathrm{Sb}'$ for $\Gamma'$.

\section{Full Algorithmic Rules}\label{app:algo}

Two judgments of the algorithmic layer are shown only in part in the body:
\Cref{fig:bidir} elides the $\Nat$ cases of inference, and \Cref{fig:srf}
trims the congruences of annotation equality. We complete just those here;
\Cref{fig:dom-relations} already shows both domain-value relations, its only
omissions being the plain-congruence cases its caption notes.

\paragraph{The $\Nat$ cases of inference}
The rules of the $\ainfer{M}{W}$ judgment elided from \Cref{fig:bidir}:
\begin{rulepar}
  \inferrule{~}{\ainfer{\Nat}{\Ty 0}}

  \inferrule{~}{\ainfer{\synze}{\Nat}}

  \inferrule{\acheck{M}{\Nat}}{\ainfer{\synsu M}{\Nat}}

  \inferrule
  {\ainfer[\Gamma, x : \Nat]{A}{\Ty i} \\
   \acheck{M_{\mathsf z}}{A[\synid, \synze/x]} \\\\
   \acheck[\Gamma, y : \Nat, r : A]{M_{\mathsf s}}{A[\synwk\circ\synwk, \synsu y/x]} \\
   \acheck{N}{\Nat} \\
   \nbetyd{\Gamma}{A[\synid, N/x]}{W}}
  {\ainfer{\recn{x.A}{M_{\mathsf z}}{y,r.M_{\mathsf s}}{N}}{W}}
\end{rulepar}

\paragraph{The remaining congruences of annotation equality}
The congruence rules of $\nfaeq{W}{W'}$ omitted from \Cref{fig:srf} (which
shows the universe, $\Pi$, and $\lambda$ rules); all are plain congruences:
\begin{rulepar}
  \inferrule{~}{\nfaeq{\Nat}{\Nat}}

  \inferrule{~}{\nfaeq{\synze}{\synze}}

  \inferrule{\nfaeq{W}{W'}}{\nfaeq{\synsu W}{\synsu{W'}}}

  \inferrule{~}{\nfaeq{x_d}{x_d}}

  \inferrule{\nfaeq{U}{U'} \\ \nfaeq{W}{W'}}{\nfaeq{U~W}{U'~W'}}

  \inferrule{\nfaeq{W}{W'} \\ \nfaeq{W_{\mathsf z}}{W_{\mathsf z}'} \\
    \nfaeq{W_{\mathsf s}}{W_{\mathsf s}'} \\ \nfaeq{U}{U'}}
  {\nfaeq{\recn{x.W}{W_{\mathsf z}}{y,r.W_{\mathsf s}}{U}}
         {\recn{x.W'}{W_{\mathsf z}'}{y,r.W_{\mathsf s}'}{U'}}}
\end{rulepar}

%% file: main.bbl

\begin{thebibliography}{45}


\ifx \showCODEN    \undefined \def \showCODEN     #1{\unskip}     \fi
\ifx \showISBNx    \undefined \def \showISBNx     #1{\unskip}     \fi
\ifx \showISBNxiii \undefined \def \showISBNxiii  #1{\unskip}     \fi
\ifx \showISSN     \undefined \def \showISSN      #1{\unskip}     \fi
\ifx \showLCCN     \undefined \def \showLCCN      #1{\unskip}     \fi
\ifx \shownote     \undefined \def \shownote      #1{#1}          \fi
\ifx \showarticletitle \undefined \def \showarticletitle #1{#1}   \fi
\ifx \showURL      \undefined \def \showURL       {\relax}        \fi
\providecommand\bibfield[2]{#2}
\providecommand\bibinfo[2]{#2}
\providecommand\natexlab[1]{#1}
\providecommand\showeprint[2][]{arXiv:#2}

\bibitem[Abadi et~al\mbox{.}(1991)]%
        {DBLP:journals/jfp/AbadiCCL91}
\bibfield{author}{\bibinfo{person}{Mart{\'{\i}}n Abadi}, \bibinfo{person}{Luca
  Cardelli}, \bibinfo{person}{Pierre{-}Louis Curien}, {and}
  \bibinfo{person}{Jean{-}Jacques L{\'{e}}vy}.}
  \bibinfo{year}{1991}\natexlab{}.
\newblock \showarticletitle{Explicit Substitutions}.
\newblock \bibinfo{journal}{\emph{J. Funct. Program.}} \bibinfo{volume}{1},
  \bibinfo{number}{4} (\bibinfo{year}{1991}), \bibinfo{pages}{375--416}.
\newblock
\href{https://doi.org/10.1017/S0956796800000186}{doi:\nolinkurl{10.1017/S0956796800000186}}


\bibitem[Abel(2013)]%
        {abel_normalization_2013}
\bibfield{author}{\bibinfo{person}{Andreas Abel}.}
  \bibinfo{year}{2013}\natexlab{}.
\newblock \emph{\bibinfo{title}{Normalization by Evaluation: Dependent Types
  and Impredicativity}}.
\newblock Habilitation Thesis.
  \bibinfo{school}{Ludwig-Maximilians-Universit{\"a}t M{\"u}nchen},
  \bibinfo{address}{Munich, Germany}.
\newblock
\urldef\tempurl%
\url{https://www.cse.chalmers.se/~abela/habil.pdf}
\showURL{%
\tempurl}


\bibitem[Abel et~al\mbox{.}(2018)]%
        {abel_decidability_2017}
\bibfield{author}{\bibinfo{person}{Andreas Abel}, \bibinfo{person}{Joakim
  {\"{O}}hman}, {and} \bibinfo{person}{Andrea Vezzosi}.}
  \bibinfo{year}{2018}\natexlab{}.
\newblock \showarticletitle{Decidability of Conversion for Type Theory in Type
  Theory}.
\newblock \bibinfo{journal}{\emph{Proc. {ACM} Program. Lang.}}
  \bibinfo{volume}{2}, \bibinfo{number}{{POPL}} (\bibinfo{year}{2018}),
  \bibinfo{pages}{23:1--23:29}.
\newblock
\href{https://doi.org/10.1145/3158111}{doi:\nolinkurl{10.1145/3158111}}


\bibitem[Abel et~al\mbox{.}(2017)]%
        {abel_normalization_2017}
\bibfield{author}{\bibinfo{person}{Andreas Abel}, \bibinfo{person}{Andrea
  Vezzosi}, {and} \bibinfo{person}{Th{\'{e}}o Winterhalter}.}
  \bibinfo{year}{2017}\natexlab{}.
\newblock \showarticletitle{Normalization by Evaluation for Sized Dependent
  Types}.
\newblock \bibinfo{journal}{\emph{Proc. {ACM} Program. Lang.}}
  \bibinfo{volume}{1}, \bibinfo{number}{{ICFP}} (\bibinfo{year}{2017}),
  \bibinfo{pages}{33:1--33:30}.
\newblock
\href{https://doi.org/10.1145/3110277}{doi:\nolinkurl{10.1145/3110277}}


\bibitem[Adjedj et~al\mbox{.}(2026)]%
        {adjedj_adaptt_2026}
\bibfield{author}{\bibinfo{person}{Arthur Adjedj}, \bibinfo{person}{Meven
  Lennon{-}Bertrand}, \bibinfo{person}{Thibaut Benjamin}, {and}
  \bibinfo{person}{Kenji Maillard}.} \bibinfo{year}{2026}\natexlab{}.
\newblock \showarticletitle{{AdapTT}: Functoriality for Dependent Type Casts}.
\newblock \bibinfo{journal}{\emph{Proc. {ACM} Program. Lang.}}
  \bibinfo{volume}{10}, \bibinfo{number}{{POPL}} (\bibinfo{year}{2026}),
  \bibinfo{pages}{628--658}.
\newblock
\href{https://doi.org/10.1145/3776664}{doi:\nolinkurl{10.1145/3776664}}


\bibitem[Aspinall and Compagnoni(2001)]%
        {aspinall_subtyping_2001}
\bibfield{author}{\bibinfo{person}{David Aspinall} {and}
  \bibinfo{person}{Adriana~B. Compagnoni}.} \bibinfo{year}{2001}\natexlab{}.
\newblock \showarticletitle{Subtyping dependent types}.
\newblock \bibinfo{journal}{\emph{Theor. Comput. Sci.}} \bibinfo{volume}{266},
  \bibinfo{number}{1-2} (\bibinfo{year}{2001}), \bibinfo{pages}{273--309}.
\newblock
\href{https://doi.org/10.1016/S0304-3975(00)00175-4}{doi:\nolinkurl{10.1016/S0304-3975(00)00175-4}}


\bibitem[Brady(2013)]%
        {DBLP:journals/jfp/Brady13}
\bibfield{author}{\bibinfo{person}{Edwin~C. Brady}.}
  \bibinfo{year}{2013}\natexlab{}.
\newblock \showarticletitle{Idris, a general-purpose dependently typed
  programming language: Design and implementation}.
\newblock \bibinfo{journal}{\emph{J. Funct. Program.}} \bibinfo{volume}{23},
  \bibinfo{number}{5} (\bibinfo{year}{2013}), \bibinfo{pages}{552--593}.
\newblock
\href{https://doi.org/10.1017/S095679681300018X}{doi:\nolinkurl{10.1017/S095679681300018X}}


\bibitem[Cardelli(1988)]%
        {cardelli_multiple_1988}
\bibfield{author}{\bibinfo{person}{Luca Cardelli}.}
  \bibinfo{year}{1988}\natexlab{}.
\newblock \showarticletitle{A Semantics of Multiple Inheritance}.
\newblock \bibinfo{journal}{\emph{Inf. Comput.}} \bibinfo{volume}{76},
  \bibinfo{number}{2/3} (\bibinfo{year}{1988}), \bibinfo{pages}{138--164}.
\newblock
\href{https://doi.org/10.1016/0890-5401(88)90007-7}{doi:\nolinkurl{10.1016/0890-5401(88)90007-7}}


\bibitem[Cardelli et~al\mbox{.}(1991)]%
        {DBLP:conf/tacs/CardelliMMS91}
\bibfield{author}{\bibinfo{person}{Luca Cardelli}, \bibinfo{person}{Simone
  Martini}, \bibinfo{person}{John~C. Mitchell}, {and} \bibinfo{person}{Andre
  Scedrov}.} \bibinfo{year}{1991}\natexlab{}.
\newblock \showarticletitle{An Extension of System {F} with Subtyping}. In
  \bibinfo{booktitle}{\emph{Theoretical Aspects of Computer Software,
  International Conference {TACS} '91, Sendai, Japan, September 24-27, 1991,
  Proceedings}} \emph{(\bibinfo{series}{Lecture Notes in Computer Science},
  Vol.~\bibinfo{volume}{526})}, \bibfield{editor}{\bibinfo{person}{Takayasu
  Ito} {and} \bibinfo{person}{Albert~R. Meyer}} (Eds.).
  \bibinfo{publisher}{Springer}, \bibinfo{pages}{750--770}.
\newblock
\href{https://doi.org/10.1007/3-540-54415-1\_73}{doi:\nolinkurl{10.1007/3-540-54415-1\_73}}


\bibitem[Carneiro(2024)]%
        {carneiro_lean4lean_2024}
\bibfield{author}{\bibinfo{person}{Mario Carneiro}.}
  \bibinfo{year}{2024}\natexlab{}.
\newblock \bibinfo{title}{{Lean4Lean}: Verifying a Typechecker for Lean, in
  Lean}.
\newblock
\showeprint[arxiv]{2403.14064}~[cs.LO]
\urldef\tempurl%
\url{https://arxiv.org/abs/2403.14064}
\showURL{%
\tempurl}


\bibitem[Castagna(1995)]%
        {castagna_contravariance_1995}
\bibfield{author}{\bibinfo{person}{Giuseppe Castagna}.}
  \bibinfo{year}{1995}\natexlab{}.
\newblock \showarticletitle{Covariance and Contravariance: Conflict without a
  Cause}.
\newblock \bibinfo{journal}{\emph{{ACM} Trans. Program. Lang. Syst.}}
  \bibinfo{volume}{17}, \bibinfo{number}{3} (\bibinfo{year}{1995}),
  \bibinfo{pages}{431--447}.
\newblock
\href{https://doi.org/10.1145/203095.203096}{doi:\nolinkurl{10.1145/203095.203096}}


\bibitem[Coquand(2013)]%
        {coquand_presheaf}
\bibfield{author}{\bibinfo{person}{Thierry Coquand}.}
  \bibinfo{year}{2013}\natexlab{}.
\newblock \bibinfo{title}{Presheaf Model of Type Theory}.
\newblock
\urldef\tempurl%
\url{https://www.cse.chalmers.se/~coquand/presheaf.pdf}
\showURL{%
\tempurl}
\newblock
\shownote{Unpublished note}.


\bibitem[Danielsson et~al\mbox{.}(2026)]%
        {danielsson_firstclass_2026}
\bibfield{author}{\bibinfo{person}{Nils~Anders Danielsson},
  \bibinfo{person}{Na{\"i}m~Camille Favier}, {and}
  \bibinfo{person}{Ond{\v{r}}ej Kub{\'a}nek}.} \bibinfo{year}{2026}\natexlab{}.
\newblock \showarticletitle{Normalisation for First-Class Universe Levels}.
\newblock \bibinfo{journal}{\emph{Proc. {ACM} Program. Lang.}}
  \bibinfo{volume}{10}, \bibinfo{number}{{POPL}} (\bibinfo{year}{2026}).
\newblock
\href{https://doi.org/10.1145/3776645}{doi:\nolinkurl{10.1145/3776645}}


\bibitem[de~Moura and Ullrich(2021)]%
        {DBLP:conf/cade/Moura021}
\bibfield{author}{\bibinfo{person}{Leonardo de Moura} {and}
  \bibinfo{person}{Sebastian Ullrich}.} \bibinfo{year}{2021}\natexlab{}.
\newblock \showarticletitle{The Lean 4 Theorem Prover and Programming
  Language}. In \bibinfo{booktitle}{\emph{Proceedings of the 28th International
  Conference on Automated Deduction, {CADE} 2021, Virtual Event, July 12-15,
  2021}} \emph{(\bibinfo{series}{Lecture Notes in Computer Science},
  Vol.~\bibinfo{volume}{12699})},
  \bibfield{editor}{\bibinfo{person}{Andr{\'{e}} Platzer} {and}
  \bibinfo{person}{Geoff Sutcliffe}} (Eds.). \bibinfo{publisher}{Springer},
  \bibinfo{pages}{625--635}.
\newblock
\href{https://doi.org/10.1007/978-3-030-79876-5\_37}{doi:\nolinkurl{10.1007/978-3-030-79876-5\_37}}


\bibitem[de~Moura et~al\mbox{.}(2015)]%
        {DBLP:conf/cade/MouraKADR15}
\bibfield{author}{\bibinfo{person}{Leonardo~Mendon{\c{c}}a de Moura},
  \bibinfo{person}{Soonho Kong}, \bibinfo{person}{Jeremy Avigad},
  \bibinfo{person}{Floris van Doorn}, {and} \bibinfo{person}{Jakob von
  Raumer}.} \bibinfo{year}{2015}\natexlab{}.
\newblock \showarticletitle{The Lean Theorem Prover (System Description)}. In
  \bibinfo{booktitle}{\emph{Proceedings of the 25th International Conference on
  Automated Deduction, {CADE} 2015, Berlin, Germany, August 1-7, 2015}}
  \emph{(\bibinfo{series}{Lecture Notes in Computer Science},
  Vol.~\bibinfo{volume}{9195})}, \bibfield{editor}{\bibinfo{person}{Amy~P.
  Felty} {and} \bibinfo{person}{Aart Middeldorp}} (Eds.).
  \bibinfo{publisher}{Springer}, \bibinfo{pages}{378--388}.
\newblock
\href{https://doi.org/10.1007/978-3-319-21401-6\_26}{doi:\nolinkurl{10.1007/978-3-319-21401-6\_26}}


\bibitem[Dybjer and Setzer(2003)]%
        {DBLP:journals/apal/DybjerS03}
\bibfield{author}{\bibinfo{person}{Peter Dybjer} {and} \bibinfo{person}{Anton
  Setzer}.} \bibinfo{year}{2003}\natexlab{}.
\newblock \showarticletitle{Induction-recursion and Initial Algebras}.
\newblock \bibinfo{journal}{\emph{Ann. Pure Appl. Log.}} \bibinfo{volume}{124},
  \bibinfo{number}{1-3} (\bibinfo{year}{2003}), \bibinfo{pages}{1--47}.
\newblock
\href{https://doi.org/10.1016/S0168-0072(02)00096-9}{doi:\nolinkurl{10.1016/S0168-0072(02)00096-9}}


\bibitem[Fridlender and Pagano(2013)]%
        {DBLP:conf/tlca/FridlenderP13}
\bibfield{author}{\bibinfo{person}{Daniel Fridlender} {and}
  \bibinfo{person}{Miguel Pagano}.} \bibinfo{year}{2013}\natexlab{}.
\newblock \showarticletitle{A Type-Checking Algorithm for Martin-L{\"{o}}f Type
  Theory with Subtyping Based on Normalisation by Evaluation}. In
  \bibinfo{booktitle}{\emph{Typed Lambda Calculi and Applications, 11th
  International Conference, {TLCA} 2013, Eindhoven, The Netherlands, June
  26-28, 2013. Proceedings}} \emph{(\bibinfo{series}{Lecture Notes in Computer
  Science}, Vol.~\bibinfo{volume}{7941})},
  \bibfield{editor}{\bibinfo{person}{Masahito Hasegawa}} (Ed.).
  \bibinfo{publisher}{Springer}, \bibinfo{pages}{140--155}.
\newblock
\href{https://doi.org/10.1007/978-3-642-38946-7\_12}{doi:\nolinkurl{10.1007/978-3-642-38946-7\_12}}


\bibitem[Frisch et~al\mbox{.}(2008)]%
        {frisch_semantic_2008}
\bibfield{author}{\bibinfo{person}{Alain Frisch}, \bibinfo{person}{Giuseppe
  Castagna}, {and} \bibinfo{person}{V{\'{e}}ronique Benzaken}.}
  \bibinfo{year}{2008}\natexlab{}.
\newblock \showarticletitle{Semantic Subtyping: Dealing Set-Theoretically with
  Function, Union, Intersection, and Negation Types}.
\newblock \bibinfo{journal}{\emph{J. {ACM}}} \bibinfo{volume}{55},
  \bibinfo{number}{4} (\bibinfo{year}{2008}), \bibinfo{pages}{19:1--19:64}.
\newblock
\href{https://doi.org/10.1145/1391289.1391293}{doi:\nolinkurl{10.1145/1391289.1391293}}


\bibitem[Gratzer et~al\mbox{.}(2020)]%
        {DBLP:conf/lics/GratzerKNB20}
\bibfield{author}{\bibinfo{person}{Daniel Gratzer}, \bibinfo{person}{G.~A.
  Kavvos}, \bibinfo{person}{Andreas Nuyts}, {and} \bibinfo{person}{Lars
  Birkedal}.} \bibinfo{year}{2020}\natexlab{}.
\newblock \showarticletitle{Multimodal Dependent Type Theory}. In
  \bibinfo{booktitle}{\emph{{LICS} '20: 35th Annual {ACM/IEEE} Symposium on
  Logic in Computer Science, Saarbr{\"{u}}cken, Germany, July 8-11, 2020}},
  \bibfield{editor}{\bibinfo{person}{Holger Hermanns}, \bibinfo{person}{Lijun
  Zhang}, \bibinfo{person}{Naoki Kobayashi}, {and} \bibinfo{person}{Dale
  Miller}} (Eds.). \bibinfo{publisher}{{ACM}}, \bibinfo{pages}{492--506}.
\newblock
\href{https://doi.org/10.1145/3373718.3394736}{doi:\nolinkurl{10.1145/3373718.3394736}}


\bibitem[Gratzer et~al\mbox{.}(2021)]%
        {DBLP:journals/lmcs/GratzerKNB21}
\bibfield{author}{\bibinfo{person}{Daniel Gratzer}, \bibinfo{person}{G.~A.
  Kavvos}, \bibinfo{person}{Andreas Nuyts}, {and} \bibinfo{person}{Lars
  Birkedal}.} \bibinfo{year}{2021}\natexlab{}.
\newblock \showarticletitle{Multimodal Dependent Type Theory}.
\newblock \bibinfo{journal}{\emph{Log. Methods Comput. Sci.}}
  \bibinfo{volume}{17}, \bibinfo{number}{3} (\bibinfo{year}{2021}).
\newblock
\href{https://doi.org/10.46298/LMCS-17(3:11)2021}{doi:\nolinkurl{10.46298/LMCS-17(3:11)2021}}


\bibitem[Hu et~al\mbox{.}(2023)]%
        {hu_jang_pientka_2023}
\bibfield{author}{\bibinfo{person}{Jason Z.~S. Hu}, \bibinfo{person}{Junyoung
  Jang}, {and} \bibinfo{person}{Brigitte Pientka}.}
  \bibinfo{year}{2023}\natexlab{}.
\newblock \showarticletitle{Normalization by Evaluation for Modal Dependent
  Type Theory}.
\newblock \bibinfo{journal}{\emph{J. Funct. Program.}}  \bibinfo{volume}{33}
  (\bibinfo{year}{2023}), \bibinfo{pages}{e7}.
\newblock
\href{https://doi.org/10.1017/S0956796823000060}{doi:\nolinkurl{10.1017/S0956796823000060}}


\bibitem[Jang et~al\mbox{.}(2025)]%
        {mctt}
\bibfield{author}{\bibinfo{person}{Junyoung Jang}, \bibinfo{person}{Antoine
  Gaulin}, \bibinfo{person}{Jason Z.~S. Hu}, {and} \bibinfo{person}{Brigitte
  Pientka}.} \bibinfo{year}{2025}\natexlab{}.
\newblock \showarticletitle{McTT: {A} Verified Kernel for a Proof Assistant}.
\newblock \bibinfo{journal}{\emph{Proc. {ACM} Program. Lang.}}
  \bibinfo{volume}{9}, \bibinfo{number}{{ICFP}} (\bibinfo{year}{2025}),
  \bibinfo{pages}{190--221}.
\newblock
\href{https://doi.org/10.1145/3747511}{doi:\nolinkurl{10.1145/3747511}}


\bibitem[Jiang et~al\mbox{.}(2025)]%
        {jiang_noncumu_2025}
\bibfield{author}{\bibinfo{person}{Shengyi Jiang}, \bibinfo{person}{Jason Z.~S.
  Hu}, {and} \bibinfo{person}{Bruno~C. d. S.~Oliveira}.}
  \bibinfo{year}{2025}\natexlab{}.
\newblock \showarticletitle{Normalization by Evaluation for Non-cumulativity}.
\newblock \bibinfo{journal}{\emph{Proc. ACM Program. Lang.}}
  \bibinfo{volume}{9}, \bibinfo{number}{ICFP} (\bibinfo{year}{2025}),
  \bibinfo{pages}{93--127}.
\newblock
\href{https://doi.org/10.1145/3747508}{doi:\nolinkurl{10.1145/3747508}}


\bibitem[Laurent(2025)]%
        {DBLP:phd/hal/Laurent25}
\bibfield{author}{\bibinfo{person}{Th{\'{e}}o Laurent}.}
  \bibinfo{year}{2025}\natexlab{}.
\newblock \emph{\bibinfo{title}{Structural Subtyping in {MLTT.} (Sous-typage
  structurel en th{\'{e}}orie des types {\`{a}} la Martin-L{\"{o}}f)}}.
\newblock \bibinfo{thesistype}{Ph.\,D. Dissertation}.
  \bibinfo{school}{University of Montpellier, France}.
\newblock
\urldef\tempurl%
\url{https://tel.archives-ouvertes.fr/tel-05144931}
\showURL{%
\tempurl}


\bibitem[Laurent et~al\mbox{.}(2024)]%
        {laurent_functoriality_2024}
\bibfield{author}{\bibinfo{person}{Th{\'{e}}o Laurent}, \bibinfo{person}{Meven
  Lennon{-}Bertrand}, {and} \bibinfo{person}{Kenji Maillard}.}
  \bibinfo{year}{2024}\natexlab{}.
\newblock \showarticletitle{Definitional Functoriality for Dependent
  (Sub)Types}. In \bibinfo{booktitle}{\emph{Programming Languages and Systems -
  33rd European Symposium on Programming, {ESOP} 2024, Held as Part of the
  European Joint Conferences on Theory and Practice of Software, {ETAPS} 2024,
  Luxembourg City, Luxembourg, April 6-11, 2024, Proceedings, Part {I}}}
  \emph{(\bibinfo{series}{Lecture Notes in Computer Science},
  Vol.~\bibinfo{volume}{14576})}, \bibfield{editor}{\bibinfo{person}{Stephanie
  Weirich}} (Ed.). \bibinfo{publisher}{Springer}, \bibinfo{pages}{302--331}.
\newblock
\href{https://doi.org/10.1007/978-3-031-57262-3_13}{doi:\nolinkurl{10.1007/978-3-031-57262-3_13}}


\bibitem[Lennon{-}Bertrand(2021)]%
        {lennonbertrand_bidir}
\bibfield{author}{\bibinfo{person}{Meven Lennon{-}Bertrand}.}
  \bibinfo{year}{2021}\natexlab{}.
\newblock \showarticletitle{Complete Bidirectional Typing for the Calculus of
  Inductive Constructions}. In \bibinfo{booktitle}{\emph{12th International
  Conference on Interactive Theorem Proving, {ITP} 2021}}
  \emph{(\bibinfo{series}{LIPIcs}, Vol.~\bibinfo{volume}{193})}.
  \bibinfo{publisher}{Schloss Dagstuhl - Leibniz-Zentrum f{\"{u}}r Informatik},
  \bibinfo{pages}{24:1--24:19}.
\newblock
\href{https://doi.org/10.4230/LIPIcs.ITP.2021.24}{doi:\nolinkurl{10.4230/LIPIcs.ITP.2021.24}}


\bibitem[Lennon{-}Bertrand(2025)]%
        {lennonbertrand2025doescertifyconversion}
\bibfield{author}{\bibinfo{person}{Meven Lennon{-}Bertrand}.}
  \bibinfo{year}{2025}\natexlab{}.
\newblock \showarticletitle{What Does It Take to Certify a Conversion
  Checker?}. In \bibinfo{booktitle}{\emph{10th International Conference on
  Formal Structures for Computation and Deduction, {FSCD} 2025, July 14-20,
  2025, Birmingham, {UK}}} \emph{(\bibinfo{series}{LIPIcs},
  Vol.~\bibinfo{volume}{337})}, \bibfield{editor}{\bibinfo{person}{Maribel
  Fern{\'{a}}ndez}} (Ed.). \bibinfo{publisher}{Schloss Dagstuhl -
  Leibniz-Zentrum f{\"{u}}r Informatik}, \bibinfo{pages}{27:1--27:23}.
\newblock
\href{https://doi.org/10.4230/LIPICS.FSCD.2025.27}{doi:\nolinkurl{10.4230/LIPICS.FSCD.2025.27}}


\bibitem[Liesnikov and Cockx(2024)]%
        {DBLP:conf/aplas/LiesnikovC24}
\bibfield{author}{\bibinfo{person}{Bohdan Liesnikov} {and}
  \bibinfo{person}{Jesper Cockx}.} \bibinfo{year}{2024}\natexlab{}.
\newblock \showarticletitle{Building a Correct-by-Construction Type Checker for
  a Dependently Typed Core Language}. In \bibinfo{booktitle}{\emph{Programming
  Languages and Systems - 22nd Asian Symposium, {APLAS} 2024, Kyoto, Japan,
  October 22-24, 2024, Proceedings}} \emph{(\bibinfo{series}{Lecture Notes in
  Computer Science}, Vol.~\bibinfo{volume}{15194})},
  \bibfield{editor}{\bibinfo{person}{Oleg Kiselyov}} (Ed.).
  \bibinfo{publisher}{Springer}, \bibinfo{pages}{63--83}.
\newblock
\href{https://doi.org/10.1007/978-981-97-8943-6\_4}{doi:\nolinkurl{10.1007/978-981-97-8943-6\_4}}


\bibitem[Liu and Weirich(2026)]%
        {liu_conversion_2026}
\bibfield{author}{\bibinfo{person}{Yiyun Liu} {and} \bibinfo{person}{Stephanie
  Weirich}.} \bibinfo{year}{2026}\natexlab{}.
\newblock \showarticletitle{Algorithmic Conversion with Surjective Pairing: {A}
  Syntactic and Untyped Approach}.
\newblock \bibinfo{journal}{\emph{Proc. {ACM} Program. Lang.}}
  \bibinfo{volume}{10}, \bibinfo{number}{{POPL}} (\bibinfo{year}{2026}),
  \bibinfo{pages}{865--894}.
\newblock
\href{https://doi.org/10.1145/3776672}{doi:\nolinkurl{10.1145/3776672}}


\bibitem[Luo(1994)]%
        {luo_ecc_1994}
\bibfield{author}{\bibinfo{person}{Zhaohui Luo}.}
  \bibinfo{year}{1994}\natexlab{}.
\newblock \bibinfo{booktitle}{\emph{Computation and Reasoning: A Type Theory
  for Computer Science}}.
\newblock \bibinfo{publisher}{Clarendon Press}.
\newblock
\showISBNx{978-0-19-853835-6}
\href{https://doi.org/10.1093/oso/9780198538356.001.0001}{doi:\nolinkurl{10.1093/oso/9780198538356.001.0001}}


\bibitem[Luo(1999)]%
        {luo_coercive_1999}
\bibfield{author}{\bibinfo{person}{Zhaohui Luo}.}
  \bibinfo{year}{1999}\natexlab{}.
\newblock \showarticletitle{Coercive Subtyping}.
\newblock \bibinfo{journal}{\emph{J. Log. Comput.}} \bibinfo{volume}{9},
  \bibinfo{number}{1} (\bibinfo{year}{1999}), \bibinfo{pages}{105--130}.
\newblock
\href{https://doi.org/10.1093/LOGCOM/9.1.105}{doi:\nolinkurl{10.1093/LOGCOM/9.1.105}}


\bibitem[Luo(2012)]%
        {luo_universe_notes_2012}
\bibfield{author}{\bibinfo{person}{Zhaohui Luo}.}
  \bibinfo{year}{2012}\natexlab{}.
\newblock \bibinfo{title}{Notes on Universes in Type Theory}.
\newblock
\urldef\tempurl%
\url{https://ncatlab.org/ufias2012/files/LuoUniverse.pdf}
\showURL{%
\tempurl}
\newblock
\shownote{Notes for a talk at the Institute for Advanced Study, Princeton}.


\bibitem[Luo et~al\mbox{.}(2013)]%
        {luo_coercive_2013}
\bibfield{author}{\bibinfo{person}{Zhaohui Luo}, \bibinfo{person}{Sergei
  Soloviev}, {and} \bibinfo{person}{Tao Xue}.} \bibinfo{year}{2013}\natexlab{}.
\newblock \showarticletitle{Coercive subtyping: Theory and implementation}.
\newblock \bibinfo{journal}{\emph{Inf. Comput.}}  \bibinfo{volume}{223}
  (\bibinfo{year}{2013}), \bibinfo{pages}{18--42}.
\newblock
\href{https://doi.org/10.1016/J.IC.2012.10.020}{doi:\nolinkurl{10.1016/J.IC.2012.10.020}}


\bibitem[Martin{-}L{\"{o}}f(1984)]%
        {DBLP:books/daglib/0000395}
\bibfield{author}{\bibinfo{person}{Per Martin{-}L{\"{o}}f}.}
  \bibinfo{year}{1984}\natexlab{}.
\newblock \bibinfo{booktitle}{\emph{Intuitionistic Type Theory}}.
  \bibinfo{series}{Studies in proof theory}, Vol.~\bibinfo{volume}{1}.
\newblock \bibinfo{publisher}{Bibliopolis}.
\newblock
\showISBNx{978-88-7088-228-5}


\bibitem[Palmgren(1998)]%
        {palmgren_universes_1998}
\bibfield{author}{\bibinfo{person}{Erik Palmgren}.}
  \bibinfo{year}{1998}\natexlab{}.
\newblock \showarticletitle{On Universes in Type Theory}.
\newblock In \bibinfo{booktitle}{\emph{Twenty {Five} {Years} of {Constructive}
  {Type} {Theory}}}. \bibinfo{publisher}{Oxford University Press}.
\newblock
\showISBNx{978-0-19-850127-5 978-0-19-191654-0}
\href{https://doi.org/10.1093/oso/9780198501275.003.0012}{doi:\nolinkurl{10.1093/oso/9780198501275.003.0012}}


\bibitem[Pujet and Tabareau(2022)]%
        {DBLP:journals/pacmpl/PujetT22}
\bibfield{author}{\bibinfo{person}{Lo{\"{\i}}c Pujet} {and}
  \bibinfo{person}{Nicolas Tabareau}.} \bibinfo{year}{2022}\natexlab{}.
\newblock \showarticletitle{Observational Equality: Now for Good}.
\newblock \bibinfo{journal}{\emph{Proc. {ACM} Program. Lang.}}
  \bibinfo{volume}{6}, \bibinfo{number}{{POPL}} (\bibinfo{year}{2022}),
  \bibinfo{pages}{1--27}.
\newblock
\href{https://doi.org/10.1145/3498693}{doi:\nolinkurl{10.1145/3498693}}


\bibitem[Pujet and Tabareau(2023)]%
        {DBLP:journals/pacmpl/PujetT23}
\bibfield{author}{\bibinfo{person}{Lo{\"{\i}}c Pujet} {and}
  \bibinfo{person}{Nicolas Tabareau}.} \bibinfo{year}{2023}\natexlab{}.
\newblock \showarticletitle{Impredicative Observational Equality}.
\newblock \bibinfo{journal}{\emph{Proc. {ACM} Program. Lang.}}
  \bibinfo{volume}{7}, \bibinfo{number}{{POPL}} (\bibinfo{year}{2023}),
  \bibinfo{pages}{2171--2196}.
\newblock
\href{https://doi.org/10.1145/3571739}{doi:\nolinkurl{10.1145/3571739}}


\bibitem[Sa{\"i}bi(1997)]%
        {saibi_inheritance_1997}
\bibfield{author}{\bibinfo{person}{Amokrane Sa{\"i}bi}.}
  \bibinfo{year}{1997}\natexlab{}.
\newblock \showarticletitle{Typing Algorithm in Type Theory with Inheritance}.
  In \bibinfo{booktitle}{\emph{Proceedings of the 24th {ACM} {SIGPLAN-SIGACT}
  Symposium on Principles of Programming Languages, {POPL} 1997}}.
  \bibinfo{publisher}{{ACM}}, \bibinfo{pages}{292--301}.
\newblock
\href{https://doi.org/10.1145/263699.263742}{doi:\nolinkurl{10.1145/263699.263742}}


\bibitem[Scherer and Abel(2011)]%
        {scherer_universe_2011}
\bibfield{author}{\bibinfo{person}{Gabriel Scherer} {and}
  \bibinfo{person}{Andreas Abel}.} \bibinfo{year}{2011}\natexlab{}.
\newblock \bibinfo{booktitle}{\emph{Universe Subtyping in Martin-L{\"o}f Type
  Theory}}.
\newblock \bibinfo{type}{{T}echnical {R}eport}.
  \bibinfo{institution}{Ludwig-Maximilians-University Munich}.
\newblock
\newblock
\shownote{Internship report,
  \url{https://gallium.inria.fr/~scherer/doc/munich-report-2011-scherer-abel.pdf}}.


\bibitem[Sozeau et~al\mbox{.}(2020)]%
        {coqcoqcorrect}
\bibfield{author}{\bibinfo{person}{Matthieu Sozeau}, \bibinfo{person}{Simon
  Boulier}, \bibinfo{person}{Yannick Forster}, \bibinfo{person}{Nicolas
  Tabareau}, {and} \bibinfo{person}{Th{\'{e}}o Winterhalter}.}
  \bibinfo{year}{2020}\natexlab{}.
\newblock \showarticletitle{Coq Coq correct! verification of type checking and
  erasure for Coq, in Coq}.
\newblock \bibinfo{journal}{\emph{Proc. ACM Program. Lang.}}
  \bibinfo{volume}{4}, \bibinfo{number}{POPL} (\bibinfo{year}{2020}),
  \bibinfo{pages}{8:1--8:28}.
\newblock
\href{https://doi.org/10.1145/3371076}{doi:\nolinkurl{10.1145/3371076}}


\bibitem[{The Agda Team}(2024)]%
        {agda}
\bibfield{author}{\bibinfo{person}{{The Agda Team}}.}
  \bibinfo{year}{2024}\natexlab{}.
\newblock \bibinfo{title}{Agda 2.6.4.3}.
\newblock
\urldef\tempurl%
\url{https://wiki.portal.chalmers.se/agda/pmwiki.php}
\showURL{%
\tempurl}


\bibitem[{The Rocq Development Team}(2025)]%
        {rocq}
\bibfield{author}{\bibinfo{person}{{The Rocq Development Team}}.}
  \bibinfo{year}{2025}\natexlab{}.
\newblock \bibinfo{booktitle}{\emph{The Rocq Prover}}.
\newblock


\bibitem[Timany and Sozeau(2017)]%
        {DBLP:journals/corr/abs-1710-03912}
\bibfield{author}{\bibinfo{person}{Amin Timany} {and} \bibinfo{person}{Matthieu
  Sozeau}.} \bibinfo{year}{2017}\natexlab{}.
\newblock \showarticletitle{Consistency of the Predicative Calculus of
  Cumulative Inductive Constructions (pCuIC)}.
\newblock \bibinfo{journal}{\emph{CoRR}}  \bibinfo{volume}{abs/1710.03912}
  (\bibinfo{year}{2017}).
\newblock
\showeprint[arXiv]{1710.03912}
\urldef\tempurl%
\url{http://arxiv.org/abs/1710.03912}
\showURL{%
\tempurl}


\bibitem[Wieczorek and Biernacki(2018)]%
        {wieczorek_coq_2018}
\bibfield{author}{\bibinfo{person}{Pawel Wieczorek} {and}
  \bibinfo{person}{Dariusz Biernacki}.} \bibinfo{year}{2018}\natexlab{}.
\newblock \showarticletitle{A Coq Formalization of Normalization by Evaluation
  for Martin-L{\"{o}}f Type Theory}. In \bibinfo{booktitle}{\emph{Proceedings
  of the 7th {ACM} {SIGPLAN} International Conference on Certified Programs and
  Proofs, {CPP} 2018, Los Angeles, California, USA, January 8-9, 2018}},
  \bibfield{editor}{\bibinfo{person}{June Andronick} {and}
  \bibinfo{person}{Amy~P. Felty}} (Eds.). \bibinfo{publisher}{{ACM}},
  \bibinfo{pages}{266--279}.
\newblock
\href{https://doi.org/10.1145/3167091}{doi:\nolinkurl{10.1145/3167091}}


\bibitem[Zwanenburg(1999)]%
        {zwanenburg_pts_1999}
\bibfield{author}{\bibinfo{person}{Jan Zwanenburg}.}
  \bibinfo{year}{1999}\natexlab{}.
\newblock \showarticletitle{Pure Type Systems with Subtyping}. In
  \bibinfo{booktitle}{\emph{Typed Lambda Calculi and Applications, 4th
  International Conference, {TLCA} '99}} \emph{(\bibinfo{series}{Lecture Notes
  in Computer Science}, Vol.~\bibinfo{volume}{1581})}.
  \bibinfo{publisher}{Springer}, \bibinfo{pages}{381--396}.
\newblock
\href{https://doi.org/10.1007/3-540-48959-2\_27}{doi:\nolinkurl{10.1007/3-540-48959-2\_27}}


\end{thebibliography}
